\newif\ifsicomp
\sicompfalse

\newif\ifgreyscale
\greyscalefalse

\ifsicomp
\documentclass{siamart220329}
\usepackage{amsmath,amsfonts,amssymb}
\else
\documentclass[11pt]{article}
\usepackage{amsmath,amsthm,amsfonts,amssymb}
\usepackage{fullpage}
\fi
\usepackage[utf8]{inputenc}
\usepackage[textwidth=0.8in,colorinlistoftodos]{todonotes}
\usepackage{enumitem}
\usepackage{mathtools}
\usepackage{physics}
\usepackage{thmtools,thm-restate}
\usepackage{enumitem}
\usepackage{graphicx}
\usepackage{subcaption}
\numberwithin{equation}{section}
\usepackage{forest}
\usepackage{tikz}
\usetikzlibrary{arrows}
\ifsicomp\else
\usepackage[pagebackref,hypertexnames=false]{hyperref}
\usepackage[capitalize,nameinlink]{cleveref}
\fi

\usepackage{microtype}
\usepackage{mleftright}
\usepackage[mathscr]{euscript}

\makeatletter
\def\@footnotecolor{red}
\define@key{Hyp}{footnotecolor}{%
 \HyColor@HyperrefColor{#1}\@footnotecolor%
}
\def\@footnotemark{%
    \leavevmode
    \ifhmode\edef\@x@sf{\the\spacefactor}\nobreak\fi
    \stepcounter{Hfootnote}%
    \global\let\Hy@saved@currentHref\@currentHref
    \hyper@makecurrent{Hfootnote}%
    \global\let\Hy@footnote@currentHref\@currentHref
    \global\let\@currentHref\Hy@saved@currentHref
    \hyper@linkstart{footnote}{\Hy@footnote@currentHref}%
    \@makefnmark
    \hyper@linkend
    \ifhmode\spacefactor\@x@sf\fi
    \relax
  }%
\makeatother

\hypersetup{
  linkcolor  = blue!80!black,
  citecolor  = green!50!black,
  footnotecolor = red!90!black,
  colorlinks = true,
}

\ifsicomp
\renewcommand{\cref}{\Cref}
\newsiamthm{claim}{Claim}
\newsiamremark{remark}{Remark}
\newsiamthm{introtheorem}{Theorem}
\newsiamthm{introcorollary}{Corollary}
\newsiamthm{openquestion}{Open question}
\else
\newtheorem{theorem}{Theorem}[section]
\newtheorem{lemma}[theorem]{Lemma}
\newtheorem{definition}[theorem]{Definition}
\newtheorem{corollary}[theorem]{Corollary}

\newtheorem{claim}[theorem]{Claim}
\newtheorem{introtheorem}{Theorem}

\theoremstyle{definition}
\newtheorem{remark}[theorem]{Remark}
\newtheorem{openquestion}{Open question}
\fi

\newcommand{\eps}{\varepsilon}
\newcommand{\set}[1]{\left\{#1\right\}}

\newcommand{\bitset}{\{0,1\}}

\newcommand{\E}{\mathbb E}

\newcommand{\N}{\mathbb N}

\newcommand{\poly}{\mathrm{poly}}

\newcommand{\supp}{\mathrm{supp}}
\newcommand{\adist}{\bar{\Delta}}
\newcommand{\dist}{\Delta}

\newcommand{\sets}{\mathcal{S}}
\newcommand{\daisy}{\mathcal{D}}
\newcommand{\petals}{\mathcal{P}}

\newcommand{\tuples}{\mathcal{T}}
\newcommand{\ones}{\mathcal{O}}

\ifsicomp
\title{A Structural Theorem for Local Algorithms with Applications to Coding, Testing, and Verification\footnote{An extended abstract of this paper appeared in SODA 2021 as ``A Structural Theorem for Local Algorithms with Applications to Coding, Testing, and Privacy'' \cite{DGL21}.}}
\author{Marcel Dall'Agnol\thanks{Princeton University
(\email{dallagnol@cs.princeton.edu}).}
\and Tom Gur\thanks{University of Cambridge (\email{tom.gur@cl.cam.ac.uk
}). Supported by the UKRI Future Leaders Fellowship MR/S031545/1 and EPSRC New Horizons Grant EP/X018180/1.}
\and Oded Lachish\thanks{Birkbeck, University of London (\email{o.lachish@bbk.ac.uk}).}}
\else
\title{A Structural Theorem for Local Algorithms \\ with Applications to Coding, Testing, and Verification\footnote{An extended abstract of this paper appeared in SODA 2021 as ``A Structural Theorem for Local Algorithms with Applications to Coding, Testing, and Privacy'' \cite{DGL21}.}}
\author{Marcel Dall'Agnol \\ University of Warwick \\ \texttt{msagnol@pm.me} \and
Tom Gur\thanks{Tom Gur is supported by the UKRI Future Leaders Fellowship MR/S031545/1 and EPSRC New Horizons Grant EP/X018180/1.} \\ University of Cambridge\\ \texttt{tom.gur@cl.cam.ac.uk}  \and
Oded Lachish \\ Birkbeck, University of London \\ \texttt{o.lachish@bbk.ac.uk}}
\fi
\date{}

\usepackage{amsmath}
\begin{document}
\maketitle
\thispagestyle{empty} 

\begin{abstract}
We prove a general structural theorem for a wide family of local algorithms, which includes property testers, local decoders, and PCPs of proximity. Namely, we show that the structure of every algorithm that makes $q$ adaptive queries and satisfies a natural robustness condition admits a sample-based algorithm with $n^{1- 1/O(q^2 \log^2 q)}$ sample complexity, following the definition of Goldreich and Ron (TOCT 2016). We prove that this transformation is nearly optimal. Our theorem also admits a scheme for constructing privacy-preserving local algorithms.

Using the unified view that our structural theorem provides, we obtain results regarding various types of local algorithms, including the following.
\begin{itemize}
\item We strengthen the state-of-the-art lower bound for relaxed locally decodable codes, obtaining an \emph{exponential} improvement on the dependency in query complexity; this resolves an open problem raised by Gur and Lachish (SICOMP 2021).

\item We show that any (constant-query) testable property admits a sample-based tester with sublinear sample complexity; this resolves a problem left open in a work of Fischer, Lachish, and Vasudev (FOCS 2015), bypassing an exponential blowup caused by previous techniques in the case of adaptive testers.

\item We prove that the known separation between proofs of proximity and testers is essentially maximal; this resolves a problem left open by Gur and Rothblum (ECCC 2013, Computational Complexity 2018) regarding sublinear-time delegation of computation.
\end{itemize}

Our techniques strongly rely on relaxed sunflower lemmas and the Hajnal–Szemer\'{e}di theorem.

\bigskip

adaptivity, sunflower lemmas.

\end{abstract}

\newpage
\thispagestyle{empty} 
\setcounter{tocdepth}{2}
\tableofcontents

\newpage

\pagenumbering{arabic} 
\section{Introduction}
Sublinear-time algorithms are central to the theory of algorithms and computational complexity. Moreover, with the surge of massive datasets in the last decade, understanding the power of computation in sublinear time rapidly becomes crucial for real-world applications. Indeed, in recent years this notion received a great deal of attention, and algorithms for a plethora of problems were studied extensively.

Since algorithms that run in sublinear time cannot even afford to read the entirety of their input, they are forced to make decisions based on a small local view of the input and are thus often referred to as \emph{local algorithms}.\footnote{This terminology makes explicit that the (sublinear) parameter of interest is the algorithm's query complexity, as opposed to, say, space (as in the streaming model) or time.} Prominent notions of local algorithms include \emph{property testers} \cite{RS96,GGR98}, which are probabilistic algorithms that solve approximate decision problems by only probing a minuscule portion of their input; \emph{locally decodable codes} (LDCs) \cite{KT00} and \emph{locally testable codes} (LTCs) \cite{GS06}, which are codes that admit algorithms that, using a small number of queries to their input, decode individual symbols and test the validity of the encoding, respectively; and \emph{probabilistically checkable proofs} (PCPs) \cite{FGLSS91,AS98,ALMSS98}, which are encodings of NP-proofs that can be verified by examining only (say) $3$ bits of the proof.

While the foregoing notions are often grouped under the umbrella term of local algorithms, they are in fact very distinct. Indeed, local algorithms perform fundamentally different tasks, such as testing, self-correcting, decoding, computing a local function, or verifying the correctness of a proof. Moreover, the tasks are often performed under different promises (e.g., proximity to a valid codeword in the case of LDCs, and deciding whether an object has a property or is far from having it in the case of property testing).

Nevertheless, despite the aforementioned diversity, one of our main \emph{conceptual} contributions is capturing a fundamental structural property that is common to all of algorithms above and beyond, which in turn implies sufficient structure for obtaining our main result. We build on work of Fischer, Lachish and Vasudev \cite{FLV15} as well as Gur and Lachish \cite{GL21}, which imply an essentially equivalent structure for \emph{non-adaptive} testers and local decoders, respectively; our generalisation captures both and extends beyond them to the adaptive setting, as well as to other classes of algorithms.

More specifically, we first formalise the notion of local algorithms in the natural way: we define them simply as probabilistic algorithms that compute some function $f(x)$, with high probability, by making a small number of queries to the input $x$. We then observe that, except for degenerate cases, having a promise on the input is necessary for algorithms that make a sublinear number of queries to it.

Finally, we formalise a natural robustness condition that captures this phenomenon and is shared by most reasonable interpretations of local algorithms.

\subsection{Robust local algorithms}
We say that a local algorithm is \emph{robust} if its output is stable under minor perturbations of the input.\footnote{Note that the notion of robustness is a priori orthogonal to locality; however, as robustness is arguably the main structural property of local algorithms, we restrict the discussion to their intersection.} To make the discussion more precise, we define a $(\rho_0,\rho_1)$\emph{-robust local algorithm} $M$ for computing a partial function $f \colon \mathcal{P} \to \bitset$ (where $\mathcal{P} \subset \bitset^n)$ as a local algorithm that satisfies the following: for every input $w$ that is $\rho_0$-close to $x$ (which may or may not belong to $\mathcal{P}$) such that $f(x)=0$, we have $M^w = 0$ with high probability; and, for every $w$ that is $\rho_1$-close to $x$ such that $f(x)=1$, we have $M^w = 1$ w.h.p. We remark that our results extend to larger alphabets (see \cref{sec:localalg}, where we formally define robust local algorithms).

We illustrate the expressivity of robust local algorithms via two examples: property testing and locally decodable codes. We remark that similarly, locally testable codes, locally correctable codes, relaxed LDCs, PCPs of proximity, and other notions can all be cast as robust local algorithms (see \cref{sec:PT,sec:codes-def,sec:PCP-def}).

\begin{figure}[h]
 \center
 \begin{subfigure}[t]{0.5\textwidth}
   \includegraphics[width=\textwidth]{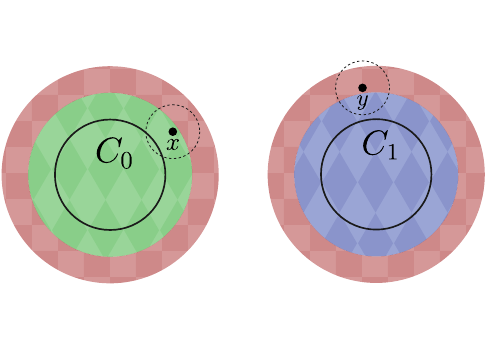}
  \caption{\small A local decoder for the code $C$ with decoding radius $\delta$. Codewords whose $i^\text{th}$ message bit equals 0 (resp. 1) for a fixed $i$ comprise $C_0$ (resp. $C_1$). The decoder is robust in the $\delta/2$-neighbourhood of $C$, with tiled \ifgreyscale\else blue and green \fi patterns. Inputs in the checkerboard \ifgreyscale\else red \fi area are within distance $\delta$ from $C$, but their $\delta/2$-neighbourhoods are not.}
  \label{fig:LDC}
 \end{subfigure}
 \hfill
 \begin{subfigure}[t]{0.4\textwidth}
   \includegraphics[width=\textwidth]{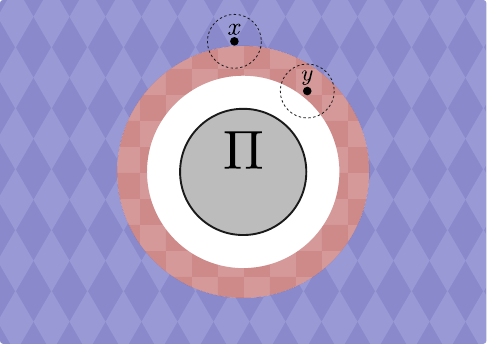}
  \caption{\small An $\eps$-tester for property $\Pi$. Inputs in the tiled \ifgreyscale\else blue \fi area are $2\eps$-far from $\Pi$, where the tester is robust. While inputs in the checkerboard \ifgreyscale\else red \fi area are rejected by the tester, this is not necessarily the case for their $\eps$\nobreakdash-neighbourhoods.}
  \label{fig:PT}
 \end{subfigure}
 \caption{Casting local decoders and property testers as robust local algorithms.}
\end{figure}

\paragraph{Locally decodable codes\ifsicomp\else.\fi} An LDC is a code that admits algorithms for decoding each individual bit of the message of a moderately corrupted codeword; that is, a code $C \colon \bitset^k \to \bitset^n$ with \emph{decoding radius} $\delta$ for which there exists a probabilistic algorithm $D$ that, given an index $i \in[k]$, queries a string $w$ promised to be $\delta$-close to a codeword $C(x)$ and outputs $D^{w}(i) = x_i$ with high probability.

Observe that $D$ can be viewed as a $(\delta/2,\delta/2)$-robust local algorithm for local decoding with respect to decoding radius $\delta/2$, that is, for the function $f \colon [k] \times \mathcal{P} \to \bitset$ where $\mathcal{P} = B_{\delta/2}(\Im C)$ is the $\delta/2$-neighbourhood of $C$.\footnote{Note that $f$ additionally receives a coordinate $i \in [k]$ as explicit input, which is allowed by the formal definition of robustness (see \cref{def:robust}).} This is because the $\delta/2$-neighbourhood of any point that is $\delta/2$-close to a codeword is still within the decoding radius, and thus the algorithm decodes the same value when given either a codeword or a string in its neighbourhood; see \cref{fig:LDC}.

\paragraph{Property testing\ifsicomp\else.\fi} Property testers are algorithms that solve approximate decision problems by only probing a minuscule part of their input, and are one of the most widely studied types of sublinear algorithms (see, e.g., the textbook \cite{G17a}).

An $\eps$-tester $T$ for a property $\Pi$ queries a string $x$ and, with high probability, outputs $1$ if $x \in \Pi$, and outputs $0$ if $x$ is $\eps$-far from $\Pi$. Here, unlike with local decoders, there is no robustness at all with respect to $1$-inputs:\footnote{Unless the tester is \emph{tolerant}: an $(\eps_1, \eps_2)$-tolerant tester is, by definition, $\eps_1$-robust with respect to its $1$-inputs.} we can only cast $T$ as an $(\eps,0)$-robust local algorithm for the function $f \colon \mathcal{P} \to \bitset$ where $\mathcal{P}$ is the union of $\Pi$ and the complement of its $2\eps$-neighbourhood. We refer to such robustness as \emph{one-sided}.

Note that while $T$ does not $\eps$-test $\Pi$ robustly, it is robust when viewed as a $2\eps$-tester: doubling the proximity parameter ensures the $\eps$-neighbourhood of each $0$-input is still rejected (see \cref{fig:PT}).

\ifsicomp\else\paragraph{}\fi By the previous discussion, our scope includes local algorithms that only exhibit \emph{one-sided} robustness. Accordingly, we define \emph{robust local algorithms} (without specifying the robustness parameters) as $(\rho_0,\rho_1)$-robust local algorithms where $\max\{\rho_0,\rho_1\} = \Omega(1)$.\footnote{Although \cref{infthm:main} assumes an algorithm satisfying this condition, we remark that a weaker one suffices. Supposing (without loss of generality) that $\rho_0 \geq \rho_1$, only \emph{a single input $x$} must imply $M^w = 0$ when $w$ is $\Omega(1)$-close to $x$; then the result follows even for $\rho_0 = \Theta(n^{-1/q}) = o(1)$, where $q$ is the query complexity of $M$.} We stress that while dealing with one-sided robustness is significantly more technically involved, our results also hold for this type of local algorithm.

\subsection{Main result}
\label{sec:intro:main}
In this work, we capture structural properties that are common to all robust local algorithms and leverage it to obtain a transformation that converts them into (uniform) \emph{sample-based algorithms}.

Sample-based algorithms are provided with uniformly distributed labeled samples, or alternatively, query each coordinate independently with some fixed probability. Adopting the latter perspective, the \emph{sample complexity} of such an algorithm is the expected number of coordinates that it samples.

 In the following, we use $n$ to denote the input size and assume the alphabet $\Sigma$ over which the input is defined is not too large (e.g., $\abs{\Sigma} \leq n^{1/q^4}$ suffices).

\begin{introtheorem}[\cref{thm:main}, informally stated]
\label{infthm:main}
    Every robust local algorithm with query complexity $q$ can be transformed into a \emph{sample-based} local algorithm with sample complexity $ n^{1- 1/O(q^2 \log^2 q)}$.
   \end{introtheorem}

We stress that the robustness in \cref{infthm:main} is only required on \emph{part of the input space} (i.e., need only be one-sided); indeed, otherwise the structural properties captured become much more restrictive (and are not shared by, e.g., property testers).

Moreover, we prove that the transformation in \cref{infthm:main} is optimal up to a quadratic factor in the dependency on the query complexity; that is, $q$-query robust local algorithms cannot be transformed into sample-based algorithms with sample complexity $n^{1- 1/o(q)}$ (see \cref{sec:map} for a more precise statement).

Our proof of \cref{infthm:main} strongly relies on analysing the query behaviour of robust local algorithms by partitioning their local views into relaxed sunflowers and using volume lemmas that are implied by their robustness. We build on the Hajnal–Szemer\'{e}di theorem to analyse sampling from relaxed sunflowers (see \cref{sec:techniques} for a detailed technical overview).

By the generality of our definition, we can apply \cref{infthm:main} to a wide family of well-studied algorithms such as locally testable codes, locally decodable and correctable codes, relaxed LDCs, universal LTCs, PCPs of proximity, and more (see \cref{sec:localalg} for details on how to cast these algorithms as robust local algorithms).

We note that \cite{FLV15} and \cite{GL21} obtain an essentially equivalent transformation for testers and decoders, respectively, through ``lossy'' versions of our relaxed sunflower lemmas (they extract \emph{one} relaxed sunflower from the local views, rather than partition them) that applies to non-adaptive algorithms; by a trivial transformation from adaptive to non-adaptive algorithms that incurs an exponential increase in the query complexity, these previous works show transformations whose sample-based algorithms have complexity $n^{1-1/\exp(q)}$, which we reduce to $n^{1-1/\poly(q)}$  (indeed, as far down as $n^{1-1/\Tilde{O}(q^2)}$).

\paragraph{Motivation\ifsicomp\else.\fi} The notion of sample-based local algorithms was first defined in \cite{GGR98}, and its systematic study was initiated by Goldreich and Ron \cite{GR16}. This is an intrinsically interesting model of computation with practical potential, as obtaining random samples is much easier than implementing full query access to a large input. Moreover, sample-based local algorithms admit schemes for multi-computation (i.e., simultaneously computing multiple functions of the input using the same queries), as well as schemes for private local computation, on which we elaborate below.

\paragraph{Private local computation\ifsicomp\else.\fi} We wish to highlight an interesting application of \cref{infthm:main} to privacy. Suppose a client wishes to compute a function of data that is stored on a server, e.g., decode a symbol of a code or test whether the data has a certain property. Typically, the query behaviour of a local algorithm may leak information on which function the client attempts to compute. However, since sample-based algorithms probe their input uniformly, they can be used to compute the desired function without revealing any information on which function was computed, e.g., which coordinate was decoded or which property was tested.

Furthermore, since \cref{infthm:main} transforms any robust local algorithm into a sample-based local algorithm that probes its input \emph{obliviously to the function it computes},\footnote{We stress that such obliviousness does not \emph{not} correspond to a differential privacy guarantee.} then (after standard error-reduction) we can apply \cref{infthm:main} to \emph{many algorithms at once} and reuse the samples to obtain a local algorithm that computes multiple functions at the same time (see \cref{sec:apps:pt}).

\subsection{Applications}
\label{sec:apps}
We proceed to our main applications, which range over three fields of study: coding theory, property testing and probabilistic proof systems. 
We remark that our testing application follows as a direct corollary of \cref{infthm:main}, but adapting it to decoders and proof systems require additional arguments.

\subsubsection{Relaxed locally decodable codes}
Locally decodable codes play an important role in contemporary coding theory. Since their systematic study was initiated by Katz and Trevisan \cite{KT00}, they made a profound impact on several areas of theoretical computer science (see, e.g., \cite{T04,Y12,KS17} and references therein), and led to practical applications in distributed storage \cite{HSXOCGLY12}.

Despite the success and attention that LDCs received in the last two decades, the best construction of $O(1)$-query LDCs has \emph{super-polynomial} blocklength (cf.\ \cite{E12}, building on \cite{Y08}). This barrier led to the study of \emph{relaxed} LDCs, which were introduced in the foundational work of Ben-Sasson, Goldreich, Harsha, Sudan, and Vadhan \cite{BGHSV06}. In a recent line of research, relaxed LDCs and variants thereof were applied to PCPs \cite{MR10,DH13,RR20a}, property testing \cite{CG18}, data structures \cite{CGW13} and probabilistic proof systems (e.g., \cite{GR17,DGRT22}).

Loosely speaking, this relaxation allows the local decoder to abort on a small fraction of the indices, yet crucially, still avoid errors. More accurately, a \emph{relaxed} LDC $C \colon \bitset^k \to \bitset^n$ with  \emph{decoding radius} $\delta$ is a code that admits a probabilistic algorithm, a decoder, which on input index $i \in [k]$ makes queries to a string $w \in\bitset^n$ that is $\delta$-close to a codeword $C(x)$ and satisfies the following:
(1) if the input is a valid codeword (i.e., $w = C(x)$), the decoder outputs $x_i$ with high probability; and
(2) otherwise, with high probability, the decoder must either output $x_i$ or a special ``abort'' symbol $\bot$, indicating it detected an error and is unable to decode.\footnote{As observed in \cite{BGHSV06}, these two conditions suffice for obtaining a third condition which guarantees that the decoder only outputs $\bot$ on an arbitrarily small fraction of the coordinates.}

This seemingly modest relaxation allows for obtaining dramatically stronger parameters. Indeed, \cite{BGHSV06} constructed a $q$-query relaxed LDC with blocklength $n = k^{1+1/\Omega(\sqrt{q})}$, and raised the problem of whether it is possible to obtain better rates; the best known construction, obtained in recent work of Asadi and Shinkar \cite{AS21}, improves it to $n = k^{1+1/\Omega(q)}$.  We stress that \emph{proving lower bounds on relaxed LDCs is significantly harder than on standard LDCs}, and indeed, the first non-trivial lower bound was only recently obtained in \cite{GL21}, which shows that, to obtain query complexity $q$, the code must have blocklength
\begin{equation*}
    n \geq k^{1 + \frac{1}{O\left(2^{2q} \cdot \log^2 q\right)}} \:.
\end{equation*}
This shows that $O(1)$-query relaxed LDCs cannot obtain quasilinear length, a question raised in \cite{G11}, but still leaves exponential room for improvement in the dependency on query complexity (note that even for $q=O(1)$ this strongly affects the asymptotic behaviour). Indeed, eliminating this exponential dependency was raised as the main open problem in \cite{GL21}.

Fortunately, our technical framework is general enough to capture \emph{relaxed} LDCs as well, and in turn, our main application for coding theory resolves the aforementioned open problem by obtaining a lower bound with an \emph{exponentially} better dependency on the query complexity. Along the way, we also extend the lower bound to hold for relaxed decoders with \emph{two-sided error}, resolving another problem left open in \cite{GL21}.

\begin{introtheorem}[\cref{cor:rldcsample}, informally stated]
  \label{infthm:rldc}
  Any relaxed LDC $C\colon \bitset^k \rightarrow \bitset^n$ with constant decoding radius $\delta$ and query complexity $q$ must have blocklength at least
    \begin{equation*}
		n \geq k^{1 + \frac{1}{O\left(q^2 \log^2 q\right)}} \:.
	\end{equation*}
\end{introtheorem}

This also makes significant progress towards resolving the problem due to \cite{BGHSV06}, by narrowing the gap between lower and upper bounds to merely a quadratic factor.

\subsubsection{Property testing}
\label{sec:apps:pt}

Recall that a standard $\eps$-tester for a property $\Pi$ is endowed with the ability to make queries, accepting inputs in $\Pi$ and rejecting inputs that are $\eps$-far from $\Pi$. An alternative definition, first given in \cite{GGR98}, only provides the tester with uniformly distributed labeled samples (or, equivalently, with uniform and independent queries to each coordinate). Such algorithms are called \emph{sample-based testers} and have received attention recently \cite{GR16,FGL14,BGS15,FLV15,CFSS17,BMR19a,BMR19}.\footnote{More accurately, these are \emph{uniform} sample-based testers (in contrast to the \cite{BGS15} tester, which queries coordinates in a random subspace). We adopt the original terminology of \cite{GR16} for simplicity.}

As an immediate corollary of \cref{infthm:main}, we obtain that any constant-query testable property (up to $\sqrt[5]{\log n}$-query, in fact) admits a sample-based tester with sublinear sample complexity. This resolves a problem left open in a work of Fischer, Lachish, and Vasudev \cite{FLV15}, who showed a similar statement for non-adaptive testers, by extending it to the adaptive setting.\footnote{\cite{FLV15} applies to adaptive testers with an exponential blowup in query complexity, which our result avoids.}

\begin{introtheorem}[\cref{cor:test}, informally stated]
	\label{infthm:PT}
  Any property $\Pi \subseteq \bitset^n$ that is $\eps$-testable with $q$ queries admits a sample-based $2\eps$-tester with sample complexity $n^{1-1/O(q^2 \log^2 q)}$.
\end{introtheorem}

This also admits an application to \emph{adaptive multi-testing}, where the goal is to simultaneously test a large number of properties. In \cref{sec:PT} we show that as a corollary of \cref{infthm:PT} we can multi-test, with sublinear query complexity, exponentially many properties, namely $k = \exp\left(n^{1/\omega(q^2 \log^2q)}\right)$, that are each testable with $q$ adaptive queries.

\subsubsection{Proofs of proximity}

Proofs of proximity \cite{RVW13} are probabilistic proof systems that allow for delegation of computation in sublinear time. They were studied extensively in recent years, finding applications in cryptography with both computational \cite{KR15} and information-theoretic security \cite{RRR21,BRV18}.

In the non-interactive setting, we have a verifier that wishes to ascertain the validity of a given statement, using a short (sublinearly long) explicitly given proof, and a sublinear number of queries to its input. Since the verifier cannot even read the entire input, it is only required to reject inputs that are far from being valid. Thus, the verifier is only assured of the proximity of the statement to a correct one. Such proof systems can be viewed as the NP (or, more accurately, MA) analogue of property testing, and are referred to as MA proofs of proximity (MAPs).

As such, one of the most fundamental questions regarding proofs of proximity is their relative strength in comparison to testers; that is, whether verifying a proof for an approximate decision problem can be done significantly more efficiently than solving it. One of the main results in \cite{GR18} is that this can indeed be the case. Namely, there exists a property $\Pi$ which: (1) admits an adaptive MAP with proof length $O(\log n)$ and query complexity $q=O(1)$; and (2) requires at least $n^{1-1/\Omega(q)}$ queries to be tested without access to a proof.\footnote{We remark that the bound in \cite{GR18} is stated in a slightly weaker form. However, it is straightforward to see that the proof achieves the bound stated above. See \cref{sec:map}.}

In \cref{sec:map} we use \cref{infthm:main} to show that the foregoing separation is nearly tight.
\begin{introtheorem}[\cref{thm:mappt}, informally stated]
\label{infthm:mappt}
	Any property $\Pi \subseteq \bitset^n$ that admits an adaptive MAP with query complexity $q$ and proof length $p$ also admits a tester with query complexity $p \cdot n^{1-1/O(q^2 \log^2q)}$.
\end{introtheorem}

Interestingly, we remark that we rely on \cref{infthm:mappt} to prove the (near) optimality of \cref{infthm:main} (see \cref{sec:map} for details).

\subsection{Discussion and open questions}
\label{sec:open}
Our work leaves several interesting directions and open problems that we wish to highlight. Firstly, we stress that our structural theorem is extremely general, and indeed the robustness condition that induces the structure required by \cref{infthm:main} appears to hold for most reasonable interpretations of robust local algorithms. While we gave applications to coding theory, property testing, and probabilistic proof systems, it would be interesting to see whether our framework (or a further generalisation of it) could imply applications to other families of local algorithms, such as PAC learners, local computation algorithms (LCAs), and beyond.

\begin{openquestion}
	Can \cref{infthm:main} and the framework of robust local algorithms be used to obtain query-to-sample transformations for PAC learners and LCAs?
\end{openquestion}

One promising direction that we did not explore is on rate lower bounds on PCPs of proximity (PCPPs). Such bounds are notoriously hard to get, and indeed the only such bounds we are aware of are those in \cite{BHLM09}, which are restricted to special setting of $3$-query PCPPs. We remark that our framework captures PCPPs, and that in light of the rate lower bounds it allowed us to obtain for relaxed LDCs, it seems feasible to obtain rate lower bounds on PCPPs as well.

\begin{openquestion}
	Can we obtain rate lower bounds on $q$-query PCPs of proximity for $q>3$?
\end{openquestion}

Another interesting question involves the optimality of our transformation. Recall that \cref{infthm:main} transforms $q$-query robust local algorithms into sample-based local algorithms with sample complexity $n^{1- 1/O(q^2 \log^2 q)}$, whereas in \cref{sec:map} we show that any such transformation must yield an algorithm with sample complexity $n^{1- 1/\Omega(q)}$. This still leaves a quadratic gap in the dependency on query complexity. We remark that closing this gap could lead to fully resolving an open question raised in \cite{BGHSV06} regarding the power of relaxed LDCs.

\begin{openquestion}
	What is the optimal sample complexity obtained by a transformation from robust local algorithms to sample-based local algorithms?
\end{openquestion}

Note, moreover, that we focus on query (or sample) complexities and provide a computationally inefficient transformation, iterating over exponentially many input strings; the \emph{computational} cost of such transformations is an interesting problem in its own regard.

\begin{openquestion}
	Are there efficient transformations from robust to sample-based algorithms?
\end{openquestion}

Finally, in \cref{sec:intro:main} we discuss an application of our main result to privacy-preserving local computation, where one can compute one out of a large collection of functions without revealing information regarding which function was computed.\footnote{Borrowing terminology from zero-knowledge, we can call this notion of privacy  \emph{perfect} -- as opposed to statistical or computational -- as the query pattern is completely independent from the choice of function.}
While \cref{infthm:main} implies such a scheme that only requires probing the input in sublinear locations, the number of probes is quite high. Moreover, by the near-tightness of our result, we cannot expect a significant improvement of this scheme.

Nevertheless, we find it very interesting to explore whether for structured families of functions (as, for example, admitted by the canonical tester for dense graphs in \cite{GT03}), or for statistical or computational notions of privacy (i.e., where the distributions of queries obtained from each function are statistically close, or indistinguishable to a computationally bounded adversary), the query complexity of this scheme can be significantly reduced.

\begin{openquestion}[Private testing of small families]
	Do there exist schemes for private local computation with small query complexity?
\end{openquestion}

\subsection*{Organisation}
The rest of the paper is organised as follows.
In \cref{sec:techniques}, we provide a high-level technical overview of the proof of our main result and its applications.
In \cref{sec:preliminaries}, we briefly discuss the preliminaries for the technical sections.
In \cref{sec:localalg}, we present our definition of robust local algorithms and show how to cast various types of algorithms in this framework.
In \cref{sec:lemmas}, we provide an arsenal of technical tools, including relaxed sunflower lemmas and a sampling lemma that builds on the Hajnal–Szemer\'{e}di theorem.
In \cref{sec:proof}, we use the foregoing tools to prove \cref{infthm:main}.
Finally, in \cref{sec:applications}, we derive our applications to coding theory, property testing, and proofs of proximity.

\section{Technical overview}
\label{sec:techniques}

In this section, we outline the techniques used and developed in the course of proving \cref{infthm:main} and its applications. Our techniques build on and simplify ideas from \cite{FLV15,GL21}, but are significantly more general and technically involved, and in particular, offer novel insight regarding adaptivity in local algorithms.

Our starting point, which we outline in \cref{sec:techniques:nonadaptive}, generalises the techniques of \cite{GL21} (which are, in turn, inspired by \cite{FLV15}) to the setting of robust local algorithms. Then, in \cref{sec:techniques:challenge}, we identify a key technical bottleneck in previous works: \emph{adaptivity}. We discuss the fundamental challenges that adaptivity imposes, and in \cref{sec:techniques:adaptive-approach} we present our strategy for meeting these challenges and the tools that we develop for dealing with them, as well as describe our construction. Subsequently, in \cref{sec:techniques:analysis}, we provide an outline of the analysis of our construction, which relies on the Hajnal–Szemer\'{e}di theorem to sample from structured set systems we call \emph{daisies}. 

\paragraph{The setting\ifsicomp\else.\fi}
Recall that our goal is to transform a robust (query-based) local algorithm into a sample-based algorithm with sublinear sample complexity. Towards this end, let $M$ be a $(\rho_{0},\rho_{1})$\nobreakdash-robust local algorithm for computing a function $f \colon \bitset^{n} \to \bitset$.\footnote{In general, the function $f$ may depend on an explicitly given parameter (e.g., an index for decoding in the case of relaxed LDCs), but for simplicity of notation, we omit this parameter in the technical overview.
} Since we also need to deal with \emph{one-sided robustness}, assume without loss of generality that $\rho_{1}=0$ and $\rho \coloneqq \rho_0 = \Omega(1)$. Recall that the algorithm $M$ receives query access to a string $x \in \bitset^n$, flips at most $r$ random coins, makes at most $q$ queries to this string and outputs $f(x) \in \bitset$ with probability at least $1 - \sigma$.

For simplicity of exposition, we assume that the error rate is $\sigma = \Theta(1/q)$, the query complexity is constant ($q=O(1)$), and the randomness complexity $r$ is bounded by $\log(n) + O(1)$. We remark that the analysis trivially extends to non-constant values of $q$, and that we can achieve the other assumptions via simple transformations, which we provide in \cref{sec:preprocessing}, at the cost of logarithmic factors in $q$. In the following, our aim is to construct a \emph{sample-based} local algorithm $N$ for computing the function $f$, with sample complexity $O\big(n^{1 - 1/2q^2}\big) = n^{1 - 1/O(q^2)}$.

\subsection{The relaxed sunflowers method}
\label{sec:techniques:nonadaptive}
As a warm-up, we first suppose that the algorithm $M$ is \emph{non-adaptive}. (This section gives an overview of the techniques of \cite{GL21}, which suffice in the non-adaptive case.) Then we can simply represent $M$ as a distribution $\mu$ over a collection of query sets $\mathcal{S}$, where each $S\in\mathcal{S}$ is a subset of $[n]$ of size $q$, and predicates $\set{f_S\colon \bitset^q \to \bitset}_{S \in \mathcal{S}}$, as follows. The algorithm $M$ draws a set $S \in \mathcal{S}$ according to $\mu$, queries $S$, obtains the \emph{local view} $x_{|S}$ (i.e., $x$ restricted to the coordinates in $S$), and outputs $f_{S}(x_{|S})$.

Consider an algorithm $N$ that samples each coordinate of the string $x$ independently with probability $p=1/n^{1/2q^2}$ (and aborts in the rare event that this exceeds the desired sample complexity).\footnote{This choice of $p$ will be made clear in \cref{sec:techniques:analysis}; see \cref{foot:sampleprob}.} Naively, we would have liked $N^x$ to emulate an invocation of the algorithm $M$ by sampling the restriction of $x$ to a query set $S \sim \mu$.

Indeed, if the distribution $\mu$ is ``well spread'', the probability of obtaining such a local view of $M$ is high. Suppose, for instance, that all of the query sets are pairwise disjoint. In this case, the probability of $N$ sampling any particular local view is $p^q$, and we expect $N$ to obtain $\Omega(p^q n) = \Omega(n^{1-1/2q})$ local views (recall that the support size of $\mu$, i.e., the number of query sets, is $O(n)$ by our assumption of $\log n + O(1)$ randomness complexity). However, if $\mu$ is concentrated on a small number of coordinates, it is highly unlikely that $N$ will obtain a local view of $M$. For example, if $M$ queries the first coordinate of $x$ with probability $1$, then we can obtain a local view of $M$ with probability at most $p$, which is negligible.

Fortunately, we can capitalise on the robustness condition to deal with this problem. We first illustrate how to do so for an easy special case, and then deal with the general setting.

\paragraph{Special case: sunflower query set\ifsicomp\else.\fi} Suppose that $\mu$ is concentrated on a small coordinate set $K$ and is otherwise disjoint, i.e., the support of $\mu$ is a \emph{sunflower} with kernel $K$ of size at most $\rho n$; see \cref{fig:sunflower}. Since the query sets are disjoint outside of $K$, by the discussion above we will sample many sets except for the coordinates in $K$ (i.e., sample the petals of the sunflower). Recall that if $x$ is such that $f(x)=0$, then the $(\rho,0)$-robust algorithm $M$ outputs $0$, with high probability, on any input $y$ that is $\rho$-close to $x$. Thus, even if we arbitrarily assign values to $K$ and use them to complete sampled petals into full local views, we can emulate an invocation of $M$ that will output as it would on $x$.

If \emph{all inputs} in the promise of $M$ were robust (as is the case for LDCs, but \emph{not} for testers, relaxed LDCs,\footnote{The type of robustness that relaxed LDCs admit is slightly more subtle, since it deals with a larger alphabet that allows for outputting $\bot$. See discussion in \cref{sec:rldc}.} and PCPPs), then the above would suffice. However, recall that we are not ensured robustness when $x$ is such that $f(x)=1$.
To deal with that, we can enumerate over all possible assignments to the kernel $K$, considering the local views obtained by completing sampled petals into full local views by using each kernel assignment to fill in the values that were not sampled. Observe that: (1) when the input $x$ is a $1$-input and $N$ considers the kernel assignment that coincides with $x$, a majority of local views (a fraction of at least $1 - \sigma$) will lead $M^x$ to output $1$; and (2), when $x$ is a $0$-input, a minority of local views (a fraction of at most $\sigma$) will lead $M^x$ to output $1$ \emph{under any kernel assignment}.

The sample-based algorithm $N$ thus outputs $1$ if and only if it sees, for some kernel assignment, a majority of local views that lead $M$ to output 1. Recall that there is asymmetry in the robustness of $M$ (while $0$-inputs are robust, $1$-inputs are not), which translates into asymmetric output conditions for $N$. Note, also, that correctness of this procedure for $0$-inputs requires that not even a single kernel assignment would lead $N$ to output incorrectly; but our assumption on the error rate ensures that the probability of sampling a majority of petals whose local views will lead to an error is sufficiently small to tolerate a union bound over all kernel assignments, as long as $\abs{K}$ is small enough.

\paragraph{General case: extracting a heavy daisy from the query sets\ifsicomp\else.\fi}

Of course, the combinatorial structure of the query sets of a local algorithm is not necessarily a sunflower and may involve many complex intersections. While we could use the \emph{sunflower lemma} to extract a sunflower from the collection of query sets, we stress that the size of such a sunflower is \emph{sublinear}, which is not enough in our setting (as we deal with constant error rate).

Nevertheless, we can exploit the robustness of $M$ even if its query sets only have the structure of a relaxed sunflower, referred to as a \emph{daisy}, with a small kernel. Loosely speaking, a $t$-daisy is a sunflower in which the kernel is not necessarily the intersection of all petals, but is rather a small subset such that every element outside the kernel is contained in at most $t$ petals;\footnote{In \cref{def:relaxed-sunflower}, a $t$-daisy has $t$ as a function from $[q]$ to $\N$ and allows for a tighter bound on the number of intersecting petals. We use the simplified definition of \cite{GL21} in this technical overview.} see \cref{fig:daisy} (and see \cref{sec:daisy} for a precise definition).

\begin{figure}
 \center
 \begin{subfigure}[t]{0.4\textwidth}
   \center
   \includegraphics[width=0.5\textwidth]{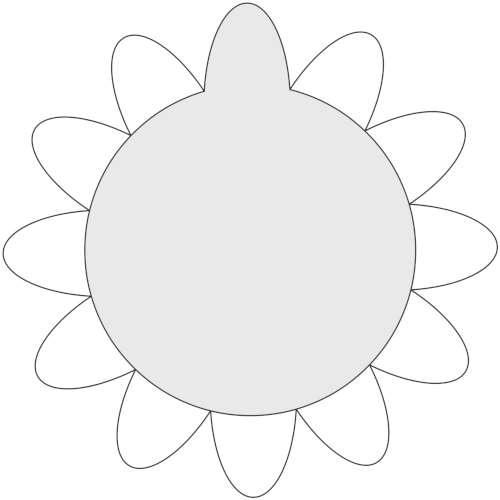}
  \caption{\small A \emph{sunflower} with one of its sets shaded. The intersection of any two sets results in the same set, the kernel.}
  \label{fig:sunflower}
 \end{subfigure}
 \hfill
 \begin{subfigure}[t]{0.5\textwidth}
   \center
   \includegraphics[width=0.4\textwidth]{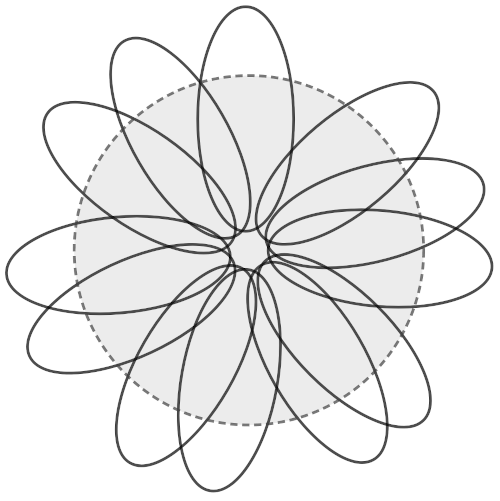}
  \caption{\small A \emph{daisy} with its kernel shaded, whose boundary is the dashed line. Outside the kernel each point is covered by a bounded number of petals.}
  \label{fig:daisy}
 \end{subfigure}
 \caption{Sunflowers and daisies.}
\end{figure}

Using a \emph{daisy lemma} \cite{FLV15,GL21}, we can extract from the query sets (the support of $\mu$) of the robust local algorithm $M$ a $t$-daisy $\daisy$  with $t$ roughly equal to $n^{i/q}$ and a kernel $K$ of size roughly $n^{1-i/q}$, where $i \in [q]$ bounds the size of the petals of $\daisy$. Moreover, the \emph{weight} $\mu(\daisy) = \sum_{S \in \daisy} \mu(S)$ is significantly larger than the error rate $\sigma$ of $M$ (recall that we assumed a sufficiently small $\sigma = \Theta(1/q)$). Thus, even if the daisy contains all local views that lead to an error, their total weight would still be small with respect to that of local views leading to a correct decision; hence, the query sets in the daisy $\daisy$ well-approximate the behaviour of $M$, and we can disregard the sets in the support of $\mu$ that do not belong to $\daisy$ at the cost of a negligible increase to the error rate.

Crucially, the intersection bound $t$ implies that \emph{sampling a daisy is similar to sampling a sunflower}: since petals do not intersect heavily, with high probability many of them are fully queried (as is the case with sunflowers). The bound on $\abs{K}$, on the other hand, allows us to implement the sampling-based algorithm we discussed for the sunflower case, except with respect to a daisy. The kernel is sufficiently small so that the output of $M$ is unchanged under any assignment to $K$, and suffices to tolerate a union bound when considering all possible assignments to $K$.

It follows that the daisy $\daisy$ provides enough ``sunflower-like'' structure for the sample-based algorithm $N$ defined previously to succeed, with high probability, when it only considers the query sets in $\daisy$ and enumerates over all assignments to its kernel.

\subsection{The challenge of adaptivity}
\label{sec:techniques:challenge}

Let us now attempt to apply the transformation laid out in the previous section to a robust local algorithm $M$ that makes $q$ \emph{adaptive} queries. In this case, $M$ may choose to query distinct coordinates depending on the answers to its previous queries, and thus \emph{there is no single distribution $\mu$} that captures its query behaviour.

Observe that now, rather than inducing a distribution on sets, the algorithm $M$ induces a distribution over \emph{decision trees} of depth $q$, as the behaviour of a randomised query-based algorithm $M^x$ can be described by choosing a decision tree according to its random string, then performing the adaptive queries according to the evaluation of that tree on the input $x \in \bitset^n$. By our assumption on the randomness complexity of $M$, this distribution is supported on $\Theta(n)$ decision trees. Note that for any \emph{fixed} input $x$, the decision tree collapses to a path, and hence the distribution over decision trees induces a distribution over query sets, which we denote $\mu_x$ (see \cref{fig:dec-tree}).

A naive way of transitioning from decision trees to sets is by querying all of the branches of each decision tree. Alas, doing so would increase the query complexity of $M$ exponentially from $q$ to (more than) $2^q$, which would in turn lead to a sample-based algorithm with a much larger sample complexity than necessary. Thus, we need to deal with the far more involved structure induced by distributions over decision trees, which imposes significant technical challenges. For starters, since our technical framework inherently relies on a combinatorial characterisation of algorithms, we first need to find a method of transitioning from decision trees to (multi-)sets without increasing the query complexity of the local algorithm $M$.

To this end, a key idea is to enumerate over all random strings and their corresponding decision trees, and extract all $q$-sets (i.e., sets of size $q$) corresponding to each branch of each tree. This leaves us with a combinatorial \emph{multi-set} $\sets$ (as multiple random strings may lead to the same decision tree, and branches of distinct decision trees may query the same set) with $\Theta(2^q \cdot n) = \Theta(n)$ query sets, of size $q$ each, corresponding to all possible query sets induced by all possible input strings.\footnote{We remark that this treatment of multi-sets allows us to significantly simplify the preparation for combinatorial analysis that was used in previous works involving sunflowers and daisies.} Note that $\sets$ contains the elements of the support of $\mu_x$ for all inputs $x \in \bitset^n$ and that, for any fixed input $x$, the vast majority of these query sets may not be relevant to this input: each $S \in \sets \setminus \supp(\mu_x)$ corresponds to a branch of a decision tree that the bits of $x$ would have not led to query.

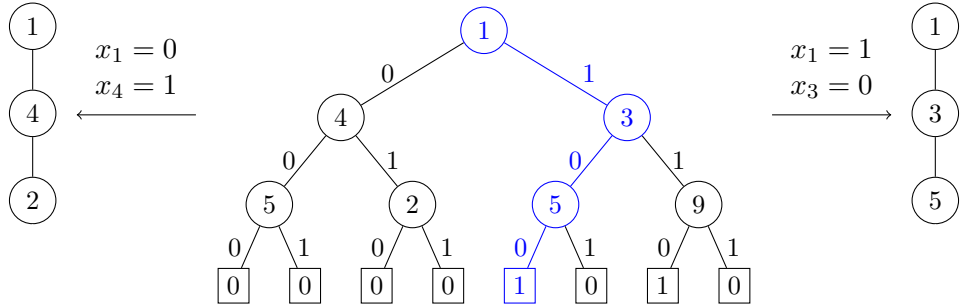
\begin{figure}
  \center
  \begin{tikzpicture}
    \node at (0,0) (tree) {\includegraphics{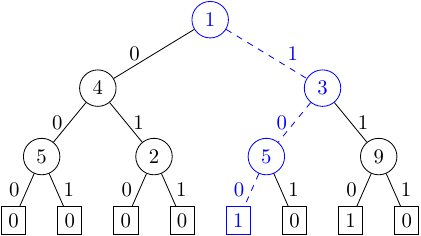}};
    \node at (6,0) (path1) {\includegraphics{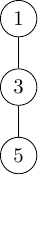}};
    \node at (-6,0) (path2) {\includegraphics{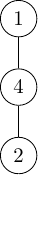}};
    \node (A) [above=0.4 of tree.east] {};
    \node (B) [above=0.4 of path1.west] {};
    \draw[->] (A) -- (B)
    node[above, pos=0.5] {$\begin{array}{c}x_1 = 1\\x_3 = 0\end{array}$};
    \node (C) [above=0.4 of tree.west] {};
    \node (D) [above=0.4 of path2.east] {};
    \draw[->] (C) -- (D)
    node[above, pos=0.5] {$\begin{array}{c}x_1 = 0\\x_4 = 1\end{array}$};
  \end{tikzpicture}
  \caption{\small Decision tree of a 3-local algorithm. When the input $x$ is such that $x_1 = 1$, $x_3 = 0$ and $x_5 = 0$, the branch highlighted in blue (and dashed) queries $\set{1,3,5}$ and outputs 1. When $x_1 = 0$ and $x_4 = 1$, this tree \emph{induces} the query set $\set{1,2,4}$; when $x_1 = 1$ and $x_3 = 0$, it induces the set $\set{1,3,5}$. This ``collapsing'' of the query behaviour is illustrated on either side of the tree.}
  \label{fig:dec-tree}
\end{figure}

This already poses a significant challenge to our approach, as we would have liked to extract a heavy daisy $\daisy$ from the collection $\sets$ which well-approximates the query sets of $M$ \emph{independently of any input}. However, it could be the case that the sets that are relevant to an input $x$ (i.e., $\supp(\mu_x)$) induce a completely different daisy (with potentially different kernels over which we'll need to enumerate) than the relevant sets for a different input $y$ that differs from $x$ on the values in the kernel, and so it is not clear at all that there exists a single daisy that well-approximates the query behaviour of the adaptive algorithm $M$ for all inputs.

Furthermore, the above also causes problems with the kernel enumeration process. For each assignment $\kappa$ to the kernel $K$, denote by $x_\kappa \in \bitset^n$ the word that takes the values of $\kappa$ in $K$ and the values of $x$ outside of $K$. Recall that the crux of our approach is to simulate executions of $M^{x_\kappa}$, for each kernel assignment $\kappa$, using the values of the sampled petals and plugging in the kernel assignment to complete these petals into local views (assignments to full query sets). Hence, since relevant sets corresponding to different kernel assignments may be distinct, it is unclear how to rule according to the local views that each of them induce.

We overcome these challenges in the next section with a more sophisticated extraction of daisies that, crucially, \emph{does not discard any query sets} of the adaptive algorithm $M$. Specifically, we will partition the (multi)-collection of all possible query sets into a collection of daisies and \emph{simultaneously analyse all daisies in the partition} to capture the adaptive behaviour of the algorithm.

\subsection{Capturing adaptivity in daisy partitions}
\label{sec:techniques:adaptive-approach}

Relying on techniques from \cite{FLV15,GL21}, we can not only extract a single heavy daisy, but rather \emph{partition} a (multi-)collection of query sets into a family of daisies, with strong structural properties on which we can capitalise. This allows us to apply our combinatorial machinery without dependency on a particular input, and \emph{analyse all daisies simultaneously}.

\paragraph{Daisy partition lemma\ifsicomp\else.\fi}
A refinement of the daisy lemma in \cite{GL21}, which we call a \emph{daisy partition lemma} (\cref{thm:daisy-partition}), partitions a multi-set $\sets$ of $q$-sets into $q + 1$ daisies $\set{\daisy_i : 0 \leq i \leq q}$ (see \cref{fig:daisy-partition}) with the following structural properties.
\begin{enumerate}
  \item $\daisy_1$ is a $n^{1/q}$-daisy, and for $i > 1$, each $\daisy_i$ is a $t$-daisy with $t = n^{(i-1)/q}$;
  \item The kernel $K_0$ of $\daisy_0$ coincides with that of $\daisy_1$, and, for $i > 0$, the kernel $K_i$ of $\daisy_i$ satisfies $\abs{K_i} \leq q\abs{\sets}\cdot n^{-i/q}$;
  \item The petal $S \setminus K_i$ of every $S \in \daisy_i$ has size exactly $i$.
\end{enumerate}
Moreover, the kernels form an incidence chain $K_q = \varnothing \subseteq K_{q-1} \subseteq \cdots \subseteq K_1 = K_0$. Note that $\daisy_0$ is vacuously a $t$-daisy for any $t$, since its petals are empty; and that our assumption on the randomness complexity of $M$ implies $\abs{K_i} = O(n^{1-i/q})$ when $i > 0$.

\begin{figure}[ht]
 \center
 \begin{subfigure}[t]{0.45\textwidth}
   \center
   \includegraphics[width=0.7\textwidth]{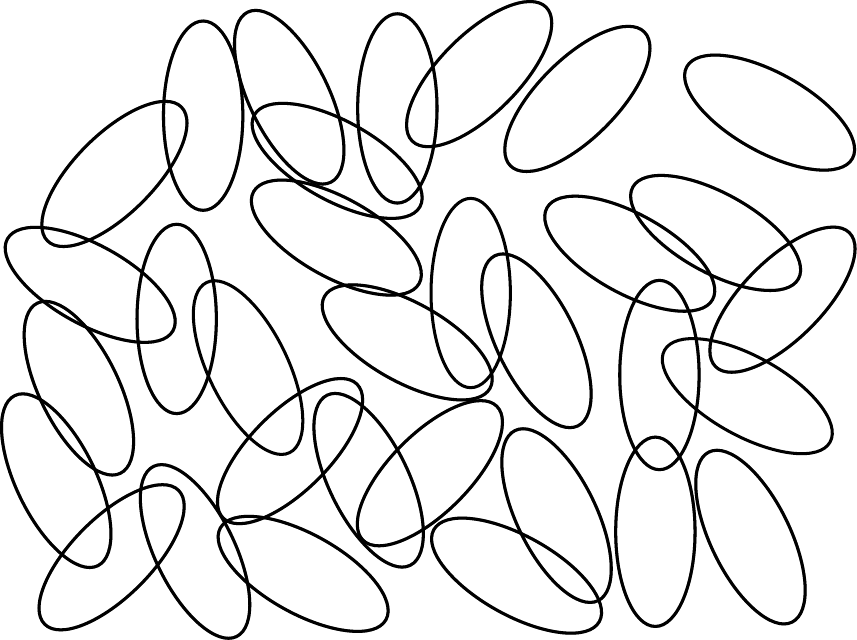}
  \caption{\small A collection $\sets$ of $3$-sets before being partitioned.}
  \label{fig:collection}
 \end{subfigure}
 \hfill
 \begin{subfigure}[t]{0.45\textwidth}
   \center
   \includegraphics[width=0.7\textwidth]{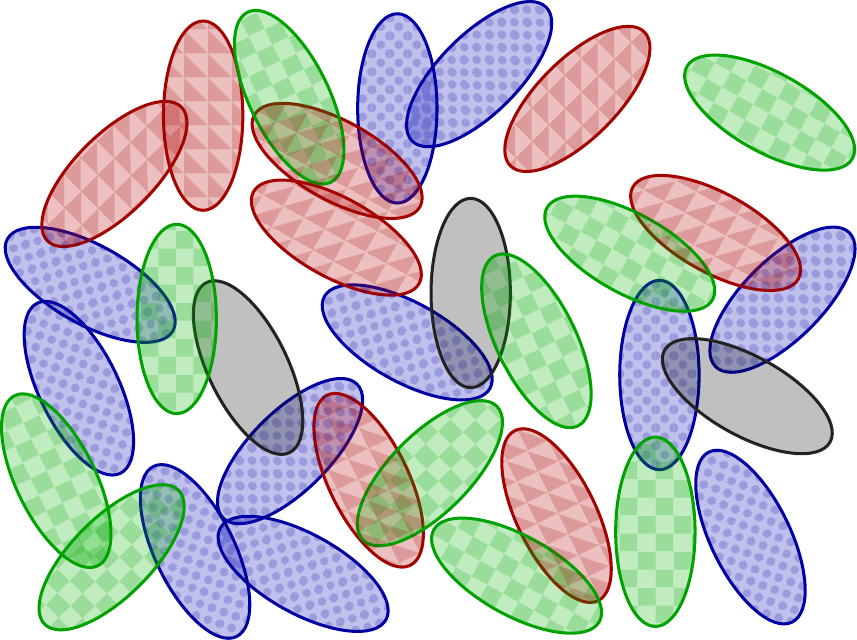}
  \caption{\small The collection $\sets$ partitioned into 4 daisies: $\daisy_0$ (shaded flat\ifgreyscale\else, grey\fi), $\daisy_1$ (checkerboard\ifgreyscale\else, green\fi), $\daisy_2$ (dotted\ifgreyscale\else, blue\fi) and $\daisy_3$ (triangle pattern\ifgreyscale\else, red\fi).}
  \label{fig:collection-partitioned}
 \end{subfigure}\\
 \center
 \begin{subfigure}[t]{0.22\textwidth}
   \center
   \includegraphics[width=\textwidth]{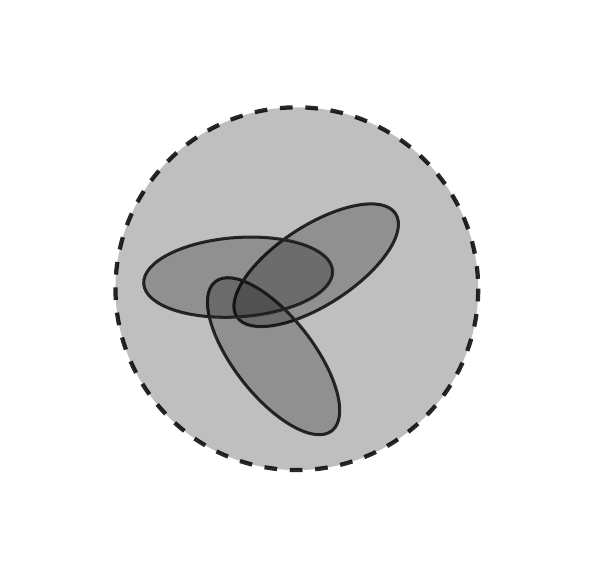}
  \caption{\small $\daisy_0$, whose sets are entirely contained in the kernel $K_0$.}
  \label{fig:daisy0}
 \end{subfigure}
 \hfill
 \begin{subfigure}[t]{0.22\textwidth}
   \center
   \includegraphics[width=\textwidth]{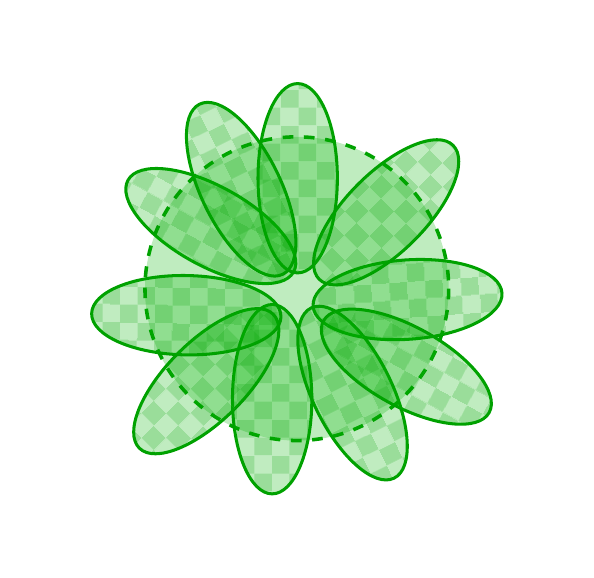}
  \caption{\small $\daisy_1$ with $K_1 = K_0$, where each $S \in \daisy_1$ has a petal $S \setminus K_1$ of size 1.}
  \label{fig:daisy1}
 \end{subfigure}
 \hfill
 \begin{subfigure}[t]{0.22\textwidth}
   \center
   \includegraphics[width=\textwidth]{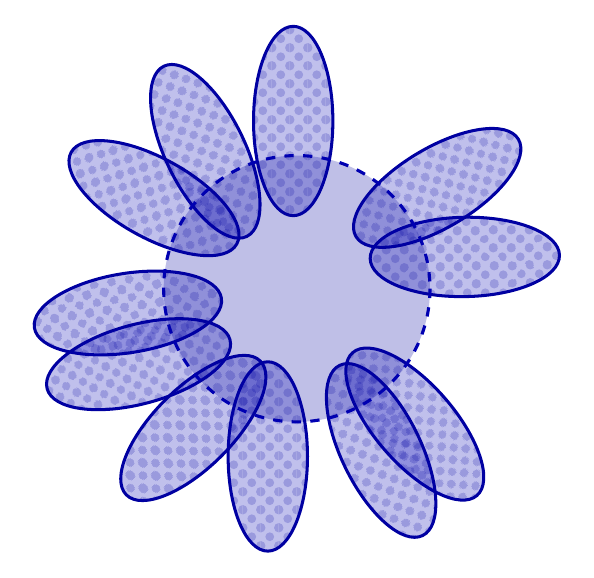}
  \caption{\small $\daisy_2$ with $K_2 \subseteq K_1$, where each $S \in \daisy_2$ has a petal $S \setminus K_2$ of size 2.}
  \label{fig:daisy2}
 \end{subfigure}
 \hfill
 \begin{subfigure}[t]{0.22\textwidth}
   \center
   \includegraphics[width=\textwidth]{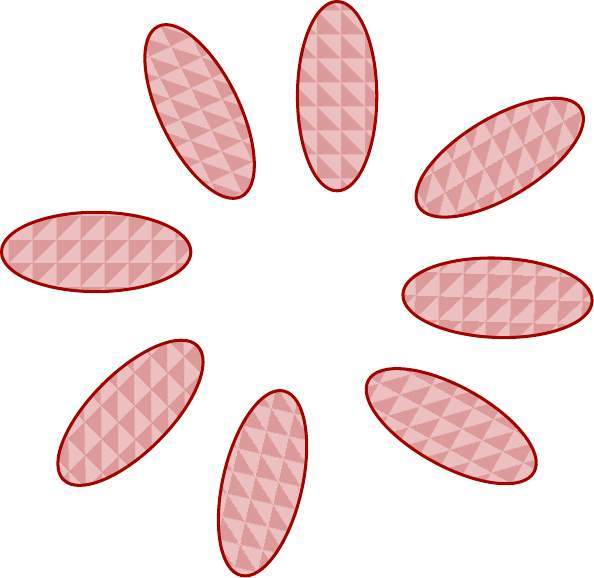}
  \caption{\small $\daisy_3$, with (empty) kernel $K_3 = \varnothing$. Each query set $S \in \daisy_3$ has a petal $S \setminus K_3 = S$ of size 3.}
  \label{fig:daisy3}
 \end{subfigure}
 \caption{Daisy partition.}
 \label{fig:daisy-partition}
\end{figure}

We may thus apply the daisy partition lemma to $\sets$ and assert that, \emph{for any input $x$}, there exists \emph{some} $i \in \set{0,\ldots,q}$ such that $\mu_x(\daisy_i)$ is larger than $1/q$ (recall that, for all $x$, the support of $\mu_x$ is contained in $\sets$); that is, each input may lead to a different heavy daisy, but there will always be at least one daisy that well-approximates the behaviour of the algorithm on input $x$. Alas, with only a local view of the input word, we are not able to tell which daisies are heavy and which are not.

It is clear, then, that a sample-based algorithm that makes use of the daisy partition has to rule not only according to a single daisy, but rather according to all of them. But \emph{how} exactly it should do so is a nontrivial question to answer, given that there are multiple daisies (and kernels) potentially interfering with one another.

\paragraph{Adaptivity in daisy partitions\ifsicomp\else.\fi}
A natural approach for dealing with multiple daisies simultaneously is by enumerating over every assignment to \emph{all} kernels (i.e., to $\cup_i K_i$) and, for each such assignment, obtaining local views from all daisies and ruling according to the aggregated local views. Note that the incidence chain structure implies that enumerating over assignments to $K_0$ suffices, since  each assignment to $K_0$ induces assignments to $K_i$ for all $i$.

However, this approach leads to fundamental difficulties. Recall that correctness of the sample-based algorithm on $0$-inputs depends on \emph{no kernel assignment} causing an output of $1$. Although for any assignment to $K_i$ this happens with sufficiently small probability to ensure it is unlikely to happen on all $2^{\abs{K_i}}$ assignments simultaneously, this does not hold true for assignments to larger kernels. More precisely, since $\abs{K_{i-1}}$ may be larger than $\abs{K_i}$ by a factor of $n^{1/q}$, an error rate that is preserved by $2^{\abs{K_i}}$ assignments becomes unbounded if the number of assignments increases to $2^{\abs{K_{i-1}}}$. This leads us to only consider, for query sets in $\daisy_i$, assignments to $K_i$ rather than to the union of all kernels.

Put differently, we construct an algorithm that deals with each daisy independently, and whose correctness follows from a delicate analysis that aggregates local views taken from all daisies, which we outline in \cref{sec:techniques:analysis}. We begin by considering a sample-based local algorithm $N$ that extends the strategy we used for a single daisy as follows.
On input $x \in \bitset^n$, it:
\begin{enumerate}[label=(\arabic*),noitemsep]
  \item samples each coordinate of the string $x$ independently with probability $p=1/n^{1/2q^2}$;
  \item for each $i \in \set{0,\ldots,q}$ and each assignment $\kappa$ to the kernel $K_i$ of the daisy $\daisy_i$, outputs $1$ if a majority of local views leads $M$ to output $1$; and
  \item outputs $0$ if no such majority is ever found.
\end{enumerate}

First, note that since the algorithm $N$ is constructed in a \emph{white-box} manner, it has access to the description of all decision trees induced by the query-based algorithm $M$. Hence $N^x$ is able to determine which local views correspond to a valid execution of $M$. Denoting by $Q$ the set of coordinates that were sampled, an assignment $\kappa$ to $K_i$ induces, for each query set $S \subset Q \cup K_i$, the assignment $x_{\kappa|S}$; the sample-based algorithm $N$ can check whether each such $S$ is a relevant query set (i.e., belongs to the support of $\mu_{x_\kappa}$) by verifying it arises from some branch of a decision tree of $M$ that $x_\kappa$ would have led to query. This allows $N$ to ignore the non-relevant query sets and overcome the difficulty pointed out in the previous section.\footnote{We remark that in the accurate description of our construction (see \cref{sec:construction}), we capture all the information contained in the decision trees via \emph{tuples} that contain, besides the query set, the assignment that led to it being queried as well as the output of the algorithm when it does so. The daisy partition lemma then allows to partition these tuples based on the structure of the sets they contain.}

However, we remain with the issue that motivated searching for heavy daisies in the first place: there is no guarantee that every $\daisy_i$ well-approximates the algorithm $M$ on all inputs. This is due to the use of relative estimates: if $x$ is a $0$-input and $\mu_{x}(\daisy_i)$ is smaller than the error rate $\sigma$, even when $N$ considers the correct kernel assignment $x_{|K_i}$ with respect to $x$, it may find a majority of local views that leads $M^x$ to output $1$; indeed, nothing prevents all the ``bad'' query sets, which lead $M^x$ to erroneously output $1$, from being placed in the same daisy $\daisy_i$.

The solution is to use a simpler decision rule: absolute rather than relative. We count the number of local views leading to an output of $1$, outputting $1$ if and only if it crosses a threshold. The upper bound $\sigma$ on the weight of ``bad'' query sets limits their number, and a large enough threshold prevents them from causing an incorrect output even if \emph{no local view} leads to the correct one. Note that a different threshold $\tau_i$ is needed for each daisy $\daisy_i$, since the probability of sampling petals decreases as $i$ increases. The thresholds $\tau_i$ must thus be carefully set to take this into account.

Finally, note that whenever the daisy $\daisy_0$ leads to an output of $1$, this happens (almost) \emph{independently of the input}: the assignment to every $S \in \daisy_0$ is determined solely by the assignment to $K_0$, because $S \subset K_0$. Therefore, the sample-based algorithm $N$ disregards $\daisy_0$ in its execution.

\paragraph{The algorithm\ifsicomp\else.\fi}
By the discussion above, we obtain the following description for the sample-based algorithm $N^x$ (with some parameters that we will set later).
\begin{enumerate}
  \item Sample each coordinate of $x$ independently with probability $p=1/n^{1/2q^2}$. If the number of samples exceeds the desired sample complexity, abort.
  \item For every $i \in [q]$ and every assignment $\kappa$ to $K_i$, perform the following steps.
  \begin{enumerate}
    \item Count the number of sets in $\daisy_i$ with local views that lead $M$ to output 1, which are \emph{relevant for the assignment $\kappa$} and the queried values. If $i = 1$, discard the sets whose petals are shared by at least $\alpha$ local views.\footnote{The extra condition for $i = 1$ is necessary to deal with the looser intersection bound $t = n^{1/q} > n^{(i-1)/q}$ on $\daisy_1$. We discuss this in the next section.}
    \item If the number is larger than the threshold $\tau_i$, output $1$.
  \end{enumerate}
  \item If every assignment to every kernel failed to trigger an output of 1, then output $0$.
\end{enumerate}

In the next section we will present key technical tools that we develop and apply to analyse this algorithm, as well as discuss the parameters $\tau_i = \gamma_i \cdot n p^i$ (where $\gamma_i = \Theta(1)$) and $\alpha = \Theta(1)$, and show it indeed suffices for the problem we set out to solve.

\subsection{Analysis using a volume lemma and the Hajnal–Szemer\'{e}di theorem}
\label{sec:techniques:analysis}
To establish the correctness of the aforementioned sample-based algorithm, we shall first need two technical lemmas about sampling daisies. We will then proceed to provide an outline of the analysis of our algorithm.

\subsubsection{Two technical lemmas}
We sketch the proofs of two simple, yet important technical lemmas that will be paramount to our analysis: (1) a lemma that allows us to transition from arguing about probability mass to arguing about combinatorial volume; and (2) a lemma that allows us to efficiently analyse sampling petals of daisies with complex intersection patterns.

\paragraph{The volume lemma\ifsicomp\else.\fi}
We start by showing how to derive from the probability mass of query sets (i.e., the probability under $\mu_x$ when the input is $x$) a bound on the volume that the union of these query sets cover. This is provided by the following \emph{volume lemma}, which captures what is arguably \emph{the defining structural property of robust local algorithms}.

Recall that the sample-based algorithm $N$ uses the query sets of a $(\rho,0)$-robust algorithm $M$ with error rate $\sigma$, which comprise the support of the distributions $\mu_x$ for all inputs $x$. Intuitively, these sets cannot be too concentrated (i.e., cover little volume), as otherwise slightly corrupting a word (in less than $\rho n$ coordinates) could require $M$ to output differently, a behaviour that is prevented by the robustness of $M$. This intuition is captured by the following \emph{volume lemma}.
\begin{lemma}[\cref{lem:volume}, informally stated]
  Let $x \in \bitset^n$ be a non-robust input (a $1$-input in our case) and $\sets$ be a subcollection of query sets in the support of $\mu_x$. If $\sets$ covers little volume (i.e., $\cup \sets < \rho n$), then it has small weight (i.e., $\mu_x(\sets) < 2 \sigma$).
\end{lemma}

We stress that the \emph{robustness of the $0$-inputs yields the volume lemma for $1$-inputs}.\footnote{This is a rather subtle consequence of adaptivity; in the nonadaptive setting a symmetric volume lemma for $b$-inputs can be shown using robustness on $b$-inputs, for $b \in \bitset$.} Note that the contrapositive of the volume lemma yields a desirable property for our sample-based algorithm: for any (non-robust) $1$-input $x$, the query sets in $\supp(\mu_x)$ must cover a large amount of volume, so that we can expect to sample many such sets.

\paragraph{The Hajnal–Szemer\'{e}di theorem\ifsicomp\else.\fi} Once we establish that a daisy covers a large volume, it remains to argue how this affects the probability of sampling petals from this large daisy, which is a key component of our algorithm. Recall that sampling the petals of a sunflower is trivial to do. However, with the complex intersection patterns that the petals of a daisy could have, we need a tool to argue about sampling petals of daisies.

First, recall that the daisy partition lemma ensures that each $\daisy_i$ is a $t$-daisy where $t = n^{\max\{1,i-1\}/q}$, for all $i$. Observe that if $\daisy_i$ is a $1$-daisy (which we call a \emph{simple} daisy), that is, each point outside the kernel $K_i$ is contained in at most one set $S \in \daisy_i$, then the sets in $\daisy_i$ have pairwise disjoint petals, so sampling them is \emph{exactly} like sampling petals of a sunflower: these petals are sampled independently from one another, and we expect their number to be concentrated around the expectation of $p^i\abs{\daisy_i}$ (recall that all petals have size $i$).

Of course, there is no guarantee that $\daisy_i$ is a simple daisy, though we expect it to \emph{contain} a simple daisy if it is large enough. Indeed, greedily removing intersecting sets yields a simple daisy of size $\Theta(\abs{\daisy_i}/t)$, but this does not suffice for our purposes because most of the sets in $\daisy_i$ are discarded.

Instead, we rely on the \emph{Hajnal–Szemer\'{e}di theorem} to obtain a ``lossless'' transition from a $t$-daisy to a collection of simple daisies, from which sampling petals is easy. The Hajnal–Szemer\'{e}di theorem shows that for every graph $G$ with $m$ vertices and maximum degree $\Delta$, and for any $k \geq \Delta + 1$, there exists a $k$-colouring of the vertices of $G$ such that every colour class has size either $\lfloor m/k\rfloor$ or $\lceil m/k\rceil$. By applying this theorem to the incidence graph of the \emph{petals} of query sets (i.e., the graph with vertex set $\daisy_i$ where we place an edge between $S$ and $S'$ when $(S \cap S') \setminus K_i \neq \varnothing$), which satisfies $\Delta(G) \leq 2t$ (see \cref{clm:hbound}) we obtain a partition of $\daisy_i$ into $t$ simple daisies of the same size (up to an additive error of 1), and hence obtain stronger sampling bounds.

\subsubsection{Analysis}

Note that the probability that $N$ samples too many coordinates (thus aborts) is exponentially small, hence we assume hereafter that this event did not occur.

We proceed to sketch the high-level argument of the correctness of the sampled-based algorithm $N$, described in the previous section, making use of tools above. This follows from two claims that hold with high probability: (1) \emph{correctness on non-robust inputs}, which ensures that when $x$ is a $1$-input (i.e., is non-robust), there exists $i \in [q]$ such that when $N$ considers the kernel assignment $x_{|K_i}$ (which coincides with the input), the number of local views that lead to output $1$ crosses the threshold $\tau_i$; and (2) \emph{correctness on robust inputs}, which, on the other hand, ensures that when $x$ is a $0$-input (i.e., is robust), for \emph{every kernel} $K_i$ and \emph{every kernel assignment}, the number of local views that lead to output $1$ \emph{does not cross} the threshold $\tau_i$.

In the following, we remind that when the sample-based algorithm $N$ considers a particular assignment $\kappa$ to a kernel and counts the number of local views that lead to output $1$, the algorithm only considers views that are \emph{relevant} to $x_\kappa$ (the input $x$ where the values of its kernel are replaced by $\kappa$); that is, local views that arise from some branch of a decision tree of the adaptive algorithm $M$ that would have led it to query these local views. While $N$ does not know all of $x$, after collecting samples from $x$ and considering the kernel assignment $\kappa$, it can check which local views are relevant to $x_\kappa$ (see discussion in \cref{sec:techniques:adaptive-approach}).

\paragraph{Correctness on non-robust inputs\ifsicomp\else.\fi}
We start with the easier case, where $x$ is a non-robust input (in our case, $f(x)=1$). We show that there exists $i \in [q]$ such that when $N$ considers the kernel assignment $x_{|K_i}$, the number of local views that lead to output $1$ crosses the threshold $\tau_i = \gamma_i \cdot n p^i$. We begin by recalling that $N$ disregards the daisy $\daisy_0$, whose query sets are entirely contained in the kernel $K_0$, and arguing that this leaves sufficiently many query sets that lead to output $1$.

Indeed, while we could not afford this if $\daisy_0$ was heavily queried by $M$ given the $1$-input $x$ (i.e., if $\mu_x(\daisy_0)$ is close to $1 - \sigma$), an application of the volume lemma shows this is not the case: since $\abs{K_0} = o(n)$, this volume is smaller than $\rho n$, implying $\mu_x(\daisy_0) < 2\sigma$ \emph{for all $1$-inputs $x$}.

Apart from $\daisy_0$, the query sets in the daisy $\daisy_1$ whose petals are shared by at least $\alpha$ local views (for a parameter $\alpha$ to be discussed shortly) are also discarded, and we need to show that the loss incurred by doing so is negligible as well.

This is accomplished with a slightly more involved application of the volume lemma:

since the sets of $\daisy_1$ have petals of size $1$, the subcollection $\mathcal{C} \subseteq \daisy_1$ of sets that are discarded covers a volume of at most $\abs{K_1} + \abs{\mathcal{C}}/\alpha$. For a sufficiently large choice of a constant $\alpha>0$, we have $\abs{\mathcal{C}}/\alpha \leq \rho n/2$ (recall that $\mathcal{C} \subseteq \supp(\mu_x)$ and $\abs{\supp(\mu_x)} = \Theta(n)$ by the assumption on the randomness complexity of $M$). Since $\abs{K_1} = o(n)$ and in particular $\abs{K_1} < \rho n/2$, applying the volume lemma to $\mathcal{C}$ shows that $\mu_x(\mathcal{C}) < 2\sigma$.

Finally, the total weight of all query sets in $\supp(\mu_x)$ that lead to output $0$ is at most $\sigma$ (by definition of the error rate $\sigma$). This implies that the subcollection of $\supp(\mu_x)$ that leads to output $1$ \emph{and is not disregarded} has weight at least $1 - 2\sigma - 2\sigma - \sigma = 1 - 5\sigma$, and, for a sufficiently small value of $\sigma$ (recall that $\sigma = \Theta(1/q)$), we have $1 - 5\sigma \geq 1/2$.

We now shift perspectives, and in effect use the volume lemma in the contrapositive direction: large weights imply large volumes. By a simple averaging argument, it follows that at least one daisy $\daisy_i$ has weight at least $(1-5\sigma)/q \geq 2\sigma$, and thus, by the volume lemma, \emph{covers at least $\rho n$ coordinates}.

Therefore, since $\abs{\supp(\mu_x)} = \Theta(n)$ and \emph{$\mu_x$ is uniform over a multi-collection of query sets}, this daisy contains $\Theta(n)$ ``good'' sets (that lead to output $1$ and were not discarded). For the analysis, using the Hajnal–Szemer\'{e}di theorem, we partition the $t$-daisy $\daisy_i$ into $t$ \emph{simple} daisies of size $\Theta(n/t)$. Each such simple daisy has \emph{disjoint} petals of size $i$, so that $\Omega(n p^i/t)$ petals will be sampled except with probability $\exp(-\Omega(n p^i/t))$.
Finally, this implies that, by setting $\gamma_i = \Theta(1)$ small enough, \emph{when $N$ considers the kernel assignment $x_{|K_i}$ to $K_i$}, at least $\tau_i = \gamma_i \cdot n p^i$ petals are sampled except with probability
\begin{equation*}
  O(t) \cdot \exp\left(-\Omega\left(\frac{n p^i}{t}\right)\right) = \exp\left(-\Omega\left(n^{1-\frac{\max\set{1,i-1}}{q}-\frac{i}{2q^2}}\right)\right) = o(1).
\end{equation*}
(Recall that $t = n^{\max\set{1,i-1}/q}$ and the sampling probability is $p = 1/n^{1/2q^2}$.)

\paragraph{Correctness on robust inputs\ifsicomp\else.\fi} It remains to show the harder case: when $x$ is a robust input (in our case, $f(x)=0$), then \emph{all kernel assignments to all daisies} will make the local views that lead to output $1$ fail to cross the threshold. This case is harder since, by the asymmetry of $N$ with respect to robust and non-robust inputs, here we need to prove a claim for all kernel assignments to all daisies, whereas in the non-robust case above we only had to argue about the existence of a single assignment to a single kernel.

We begin with a simple observation regarding $\daisy_0$, then analyse the daisies $\daisy_i$ for $i > 1$, and deal with the more delicate case of $\daisy_1$ last. Recall that $\daisy_0$ is disregarded by the algorithm $N$, and that by the asymmetry of $N$ with respect to 0- and 1-inputs, this only makes the analysis on robust inputs \emph{easier}. Indeed, $N^x$ is correct when no kernel assignment to any of the $\daisy_i$'s leads to crossing the threshold $\tau_i$ of local views on which the query-based algorithm $M$ outputs $1$. Thus, by discarding the query sets in $\daisy_0$, we only increase the probability of not crossing these thresholds.

Fix $i > 0$ and an arbitrary kernel assignment $\kappa$ to $K_i$.

Then, the relevant sets that $N$ may sample are those in the support of $\mu_{x_\kappa}$ (recall that $x_\kappa$ is the word obtained by replacing the bits of $x$ whose coordinates lie in $K_i$ by $\kappa$). Since $\abs{K_i} = o(n)$, it follows by the robustness of $x$ that $x_\kappa$ is $\rho$-close to $x$, and thus the weight of the collection $\ones \subseteq \supp(\mu_{x_\kappa})$ of query sets that lead to output $1$ is at most $\sigma$. For the sake of this technical overview, we focus on the worst-case scenario, where all of these ``bad'' sets are in the daisy $\daisy_i$ (i.e., $\ones \subseteq \daisy_i$) and $\abs{\ones}$ is as large as possible (i.e., $\abs{\ones} = \Theta(n)$), and show that even that will not suffice to cross the threshold $\tau_i$.

By the randomness complexity of the algorithm $M$, the size of $\ones$ is $\Theta(n)$. We apply the Hajnal–Szemer\'{e}di theorem and partition $\ones$ into $\Theta(t)$ simple daisies of size $\Theta(n/t)$. Recall that the petals of query sets in $\daisy_i$ have size $i$ and are disjoint; therefore, each of these simple daisies has $\gamma_i \cdot n p^i/t$ sampled petals with probability only $\exp(-\Omega(n p^i/t))$.\footnote{We stress that \emph{the expected number of sampled petals is smaller} in the robust case than in the non-robust one. This is what allows us to show the total number of queried petals is at least $\tau_i =  \gamma_i n p^i$ with probability $\exp(-\Omega(n p^i/t))$ in the robust case but $1 - \exp(-\Omega(n p^i/t))$ in the non-robust, for the same constant $\gamma_i$.} By an averaging argument, the total number of sampled petals crosses $\tau_i = \gamma_i \cdot np^i$ with probability at most
\begin{equation*}
  O(t) \cdot \exp\left(-\Omega\left(\frac{n p^i}{t}\right)\right)= \exp\left(-\Omega\left(n^{1-\frac{i-1}{q}-\frac{i}{2q^2}}\right)\right) = \exp\left(-\Omega\left(n^{1-\frac{i}{q} + \frac{1}{2q}}\right)\right) \:;
\end{equation*}
recall that $t = n^{(i-1)/q}$ and the sampling probability is $p = 1/n^{1/2q^2}$, so $p^i \geq 1/n^{1/2q}$.
Since the daisy partition lemma yields a bound of $O(n^{1-i/q})$ for the size of the kernel $K_i$, a union bound over all $2^{\abs{K_i}}$ kernel assignments ensures the threshold is crossed with probability $o(1)$.\footnote{\label{foot:sampleprob}Recall that we set the sampling probability to be $p \coloneqq 1/n^{1/\beta}$ with $\beta = 2q^2$. This choice is justified as follows: the union bound requirement that $2^{\abs{K_i}}$ multiplied by the probability of crossing the threshold be small translates into $1/q - i/\beta > 0$ for all $i \in [q-1]$. Then $i = q - 1$ requires $\beta = \Omega(q^2)$.}

We now analyse $\daisy_1$ and stress that the need for a separate analysis arises from the looser intersection bound on this daisy: $\daisy_1$ is a $t$-daisy with $t = n^{1/q}$, whereas for all other $i$ the bound is $t = n^{(i-1)/q}$. This implies that there is no ``gap'' between the expected number of queried petals in each simple daisy $\Theta(np/t) = \Theta(n^{1-1/q-1/(2q^2)}) = o(n^{1-1/q})$ and the size of the kernel $\abs{K_1} = O(n^{1-1/q})$, so a union bound as in the case $i > 1$ does not suffice.

This is precisely what the ``capping'' performed by $N$ on $\daisy_1$ is designed to address:
 
the query sets $\ones \subseteq \supp(\mu_{x_\kappa})$ that lead to output $1$ will only be counted by $N$ \emph{if their petals are shared by at most $\alpha$ query sets}. Then, by the Hajnal–Szemer\'{e}di theorem, we partition $\ones$ into $\alpha = \Theta(1)$ simple daisies of size $\Theta(n)$. Each simple daisy will have more than $\tau_i/\alpha = \Theta(np)$ queried petals with probability at most $\exp(-\Omega(np))$, so that the total number of such petals across all simple daisies exceeds $\tau_j$ with probability at most $\exp(-\Omega(np))$. This provides the necessary gap: as $\Theta(n p) = \Omega(n^{1-1/2q^2})$ and $\abs{K_1} = O(n^{1-1/q})$, a union bound over all $2^{\abs{K_1}}$ assignments to $K_1$ shows the threshold is crossed with probability $o(1)$.

\ifsicomp\else\paragraph{}\fi This concludes our high-level proof of correctness, and thus of \cref{infthm:main} (see \cref{sec:analysis} for the full proof). For an overview of how to derive our applications from \cref{infthm:main}, see \cref{sec:applications}.

\section{Preliminaries}
\label{sec:preliminaries}

Throughout this paper, constants are denoted by Greek lowercase letters, such as $\alpha$, $\beta$ and $\gamma$; capital letters of the Latin alphabet generally denote sets (e.g. $P$ and $S$) or algorithms (e.g., $M$ and $N$). For each $n \in \N$, we denote by $[n]$ the set $\set{1,\ldots,n}$. Sets $S$ such that $\abs{S} = q$ are called $q$-sets. The complement of $S$ is denoted $\overline{S}$. We will use $\Sigma$ to denote an alphabet.

As integrality issues do not substantially change any of our results, equality between an integer and an expression (that may not necessarily evaluate to one) is assumed to be rounded to the nearest integer.

\paragraph{Multi-sets of sets\ifsicomp\else.\fi} To prevent ambiguity, we call (multi-)sets comprised of objects other than points (such as sets, trees or tuples) (multi-)\emph{collections} in this work, and denote them by the calligraphic capitals $\daisy, \sets, \tuples$.

\paragraph{Distance and proximity\ifsicomp\else.\fi} We denote the \emph{absolute distance} between two strings $x, y \in \Sigma^n$ (over the alphabet $\Sigma$) by $\adist(x,y) \coloneqq \abs{\left\{ x_i \neq y_i \;\colon\; i \in [n] \right\}}$ and their \emph{relative distance} (or simply \emph{distance}) by $\dist(x,y) \coloneqq \frac{\adist(x,y)}{n}$. If $\dist(x,y) \leq \eps$, we say that $x$ is \emph{$\eps$-close} to $y$, and otherwise we say that $x$ is \emph{$\eps$-far} from $y$. The (Hamming) ball of radius $\eps$ around $x$ is $B_\eps(x) = \set{y \in \Sigma^n : y \text{ is } \eps\text{-close to } x}$, and the ball $B_\eps(S)$ around a set $S \subseteq \Sigma^n$ is the union over $B_\eps(x)$ for each $x \in S$.

\paragraph{Probability\ifsicomp\else.\fi} The uniform distribution over a set $S$ is denoted $U_S$. We write $X \sim \mu$ to denote a random variable $X$ with distribution $\mu$; the probability of event $[X = s]$ is interchangeably referred to as \emph{weight}, and is denoted $\Pr[X = s]$ (the underlying distribution will be clear from context), and the expected value of $X$ is $\E[X]$. We also write $\abs{\mu}$ as shorthand for $\abs{\supp(\mu)}$, the support size of $\mu$. Below is the version of the Chernoff bound that will be used in this work.

\begin{lemma}[Chernoff bound]
  Let $X_1, \ldots, X_k$ be independent Bernoulli random variables distributed as $X$. Then, for every $\delta \in [0,1]$,
  \begin{align*}
    \Pr\left[\frac{1}{k}\sum_{i = 1}^k X_i \geq (1 + \delta) \E[X]\right] &\leq e^{-\frac{\delta^2 k \E[X]}{3}} \text{ and}\\
    \Pr\left[\frac{1}{k}\sum_{i = 1}^k X_i \leq (1 - \delta) \E[X]\right] &\leq e^{-\frac{\delta^2 k \E[X]}{2}}.
  \end{align*}
\end{lemma}

\paragraph{Algorithms\ifsicomp\else.\fi} We denote by $M^x(z)$ the output of algorithm $M$ given direct access to input $z$ and query access to string $x$. Probabilistic expressions that involve a randomised algorithm $M$ are taken over the inner randomness of $M$ (e.g., when we write $\Pr[M^x(z) = s]$, the probability is taken over the coin tosses of $M$). The number of coin tosses $M$ makes is its \emph{randomness complexity}. The maximal number of queries $M$ performs over all strings $x$ and outcomes of its coin tosses is interchangeably referred to as its \emph{query complexity} or \emph{locality} $q$. When $q = o(n)$, where $n$ is the length of the string $x$, we say $M$ is a ($q$-)\emph{local} algorithm. If the queries performed by $M$ are determined in advance (so that no query depends on the result of any other query), $M$ is \emph{non-adaptive}; otherwise, it is \emph{adaptive}. Finally, if $M$ queries each coordinate independently with some probability $p$, we say it is a \emph{sample-based} algorithm. Since we will want to have an absolute bound on the sample complexity (i.e., the number of coordinates sampled) of sample-based algorithms, we allow them to cap the number of coordinates they sample.

\paragraph{Notation\ifsicomp\else.\fi} We will use the following, less standard, notation. An \emph{assignment} to a set $S$ (over alphabet $\Sigma$) is a function $a\colon S \rightarrow \Sigma$, which may be equivalently seen as a vector in $\Sigma^{\abs{S}}$ whose coordinates correspond to elements of $S$ in increasing order. Its restriction to $P \subseteq S$ is denoted $a_{|P}$. If $x$ is an assignment to $S$ and $\kappa$ is an assignment to $P \subseteq S$, the \emph{partially replaced assignment} $x_{\kappa} \in \Sigma^n$ is that which coincides with $\kappa$ in $P$ and with $x$ in $S \setminus P$ (i.e., $x_{\kappa|P} = \kappa$ and $x_{\kappa|S\setminus P} = x_{S \setminus P}$).

\paragraph{Adaptivity\ifsicomp\else.\fi} Adaptive local algorithms are characterised in two equivalent manners: a standard description via decision trees and an alternative that makes more direct use of set systems. Let $M$ be a $q$-local algorithm for a decision problem (i.e., which outputs an element of $\bitset$) with oracle access to a string over alphabet $\Sigma$ and no access to additional input (what follows immediately generalises by enumerating over explicit inputs).

The behaviour of $M$ is completely characterised by a collection $\set{(T_s, s) : s \in \bitset^r}$ of \emph{decision trees}, where $r$ is the randomness complexity of $M$; the trees are $\abs{\Sigma}$-ary, have depth $q$, edges labeled by elements of $\Sigma$, inner nodes labeled by elements of $[n]$ and leaves labeled by elements of $\bitset$. The execution of $M^x$ proceeds as follows. It (1) flips $r$ random coins, obtains a string $s \in \bitset^r$ and chooses the decision tree $T_s$ to execute; (2) beginning at the root, for $q$ steps, queries the coordinate given by the label at the current vertex and follows the edge whose label is the queried value; and (3) outputs the bit given by the label of the leaf thus reached.

Although this offers a complete description of an adaptive algorithm, we choose to use an alternative that is amenable to \emph{daisy lemmas} (see \cref{sec:daisy}). This is obtained by describing each decision tree with the collection of its branches. More precisely, from $\set{(T_s, s) : s \in \bitset^r}$ we construct $\set{(S_{st}, a_{st}, b_{st}, s, t) : s \in \bitset^r, t \in [\abs{\Sigma}^q]}$, where $t$ identifies which branch the tuple is obtained from. $S_{st}$ is the $q$-set queried by the $t^\text{th}$ branch of $T_s$, while $a_{st}$ is the assignment to $S_{st}$ defined by the edges of this branch and $b_{st} \in \bitset$ is the output at its leaf. We remark that the decision trees may be reconstructed from their branches, so that this is description is indeed equivalent (though we will not need this fact).

We now define the distribution of an algorithm, as well as its distribution under a fixed input.
\begin{definition}[Induced distribution]
  \label{def:distb}
  Let $M$ be a $q$-local algorithm with randomness complexity $r$ described by the collection of decision trees $\set{T_s : s \in \bitset^r}$. The \emph{distribution} $\Tilde{\mu}^M$ of $M$ is given by sampling $s \in \bitset^r$ uniformly at random and taking $T_s$.

  Fix an arbitrary input $x$ to $M$ and, for all $s \in \bitset^r$, let $(S_{st}, x_{|S_{st}}, b_{st}, s, t)$ be the unique tuple defined by the branch of $T_s$ followed on input $x$. We may thus discard $t$ and the tuple $(S_s, x_{|S_s}, b_s, s)$ is well defined. The distribution $\mu_x^M$ is given by sampling $s \in \bitset^r$ uniformly at random and taking the set $S_s$ (the first element of the tuple $(S_s, x_{|S_s}, b_s, s)$).
\end{definition}
We note that the contents of the tuple $(S_s, x_{|S_s}, b_s, s)$ describe exactly how $M$ will behave on input $x$ and random string $s$.

\section{Robust local algorithms}
\label{sec:localalg}

We now formally introduce \emph{robust local algorithms}, which capture a wide class of sublinear algorithms, ranging from \emph{property testing} to \emph{locally decodable codes}. Our main result (\cref{infthm:main}) holds for any robust local algorithm, and indeed, we obtain our results for coding theory, testing, and proofs of proximity as direct corollaries.

While local algorithms are very well studied, their definition is typically context-dependent, where they are required to perform different tasks (e.g., test, self-correct, decode, perform a local computation) under different promises (e.g., proximity to encoded inputs, being either ``close'' to or ``far'' from sets). However, structured promises on the input are (with the exception of degenerate cases) necessary for algorithms that only make a sublinear number of queries. This feature leads naturally to the notion of robustness, which, loosely speaking, a local algorithm satisfies if its output is stable under small perturbations.

In the next subsection, we provide a precise definition of robust local algorithms. Then, in the subsequent subsections, we show how this notion captures property testing, locally testable codes, locally decodable and correctable codes, PCPs of proximity, and other local algorithms.

\subsection{Definition}
We begin by defining \emph{local algorithms}, which are probabilistic algorithms that receive query access to an input $x$ and explicit parameter $z$ and are required to compute a partial function $f(z,x)$ (which represents a promise problem) by only making a small number of queries to the input $x$.

\begin{definition}[Local algorithms]
  \label{def:localalg}
  Let $\Sigma$ be a finite alphabet, $Z$ a finite set and $\set{\mathcal{P}_z : z \in Z}$ a family of sets $\mathcal{P}_z \subseteq \Sigma^n$ indexed by $Z$. With $\mathcal{P} \coloneqq \set{(z, x) : z \in Z, x \in \mathcal{P}_z}$, let $f\colon \mathcal{P} \to \bitset$ be a partial function.\footnote{We remark that allowing only rectangles $\mathcal{P} = Z \times \mathcal{Q}$ as the domain of $f$ suffices for most of our applications (e.g., testers and local decoders), but not all. For example, in a MAP for a property $\Pi$, there may be inputs $x \in \Pi$ that are only contained in $\mathcal{P}_z$ for a single $z \in Z$. (See \cref{sec:map}.)} A \emph{$q$-local algorithm $M$ for computing $f$} with \emph{error rate} $\sigma$ receives explicit access to $z \in Z$, query access to $x \in \mathcal{P}_z$, makes at most $q$ queries to $x$ and satisfies
  \begin{equation*}
		\Pr[M^x(z) = f(z,x)] \geq 1 - \sigma.
  \end{equation*}
\end{definition}

The parameter $q$ is called the \emph{query complexity} of $M$, to which we also refer as \emph{locality}. Throughout, when we refer to a \emph{local algorithm}, we mean a $q$-local algorithm with $q=o(n)$. Another important parameter is the \emph{randomness complexity} of $M$, defined as the maximal number of coin tosses it makes over all $(z, x) \in Z \times \Sigma^n$ (note that an execution $M^x(z)$ is well-defined even if $x \notin \mathcal{P}_z$).

The following definition formalises the aforementioned natural notion of \emph{robustness}, which is the structural property that underlies local computation.

\begin{definition}[Robustness]
  \label{def:robust}
  Let $\rho > 0$. A local algorithm $M$ for computing $f \colon \mathcal{P} \to \bitset$ is \emph{$\rho$-robust} at the point $(z, x) \in \mathcal{P}$ if $\Pr[M^w(z) = f(z,x)] \geq 1 -\sigma$ for all $w \in B_\rho(x)$. We say that $M$ is \emph{$(\rho_0, \rho_1)$-robust} if, for all $z \in Z$ and $b \in \bitset$, $M$ is $\rho_b$-robust at every $x$ such that $f(z,x) = b$.
 \end{definition}
  If a local algorithm $M$ is $(\rho_0, \rho_1)$-robust and $\max\set{\rho_0, \rho_1} = \Omega(1)$ (a constant independent of $n$), we simply call $M$ \emph{robust}.
Note that non-trivial robustness is only possible because $f$ is a partial function; that is, the local algorithm $M$ solves a \emph{promise problem} where, for every parameter $z$, the algorithm is promised to receive an input from $\mathcal{P}_{z,0} \coloneqq f^{-1}(z,0)$ on which it should output $0$, or an input from $\mathcal{P}_{z,1} \coloneqq f^{-1}(z,1)$ on which it should output $1$.

\begin{remark}[One-sided robustness]
For our main result (\cref{infthm:main}), it suffices to have \emph{one-sided robustness}, i.e., $(\rho_0, \rho_1)$-robustness where only one of $\rho_0, \rho_1$ is non-zero. For example, in the setting of property testing with proximity parameter $\eps$ we only have $(\eps, 0)$-robustness (see \cref{sec:PT} for details). To simplify notation, we refer to $(\rho, 0)$-robust local algorithms as $\rho$-robust.
\end{remark}

\begin{remark}[Larger alphabets]
  The definition of local algorithms can be further generalised to a constant-size output alphabet $\Gamma$, in which case the partial function is $f\colon \Sigma^n \rightarrow \Gamma$; we assume $\Gamma = \bitset$ for simplicity of exposition, but note that our results extend to larger output alphabets in a straightforward manner.
\end{remark}

We proceed to show how to capture various well-studied families of sublinear algorithms (such as testers, local decoders, and PCPs) using the notion of robust local algorithms.

\subsection{Property testing}
\label{sec:PT}

Property testing \cite{RS96,GGR98} deals with probabilistic algorithms that solve approximate decision problems by making a small number of queries to their input. More specifically, a tester is required to decide whether its input is in a set $\Pi$ (i.e., has the property $\Pi$) or whether it is $\eps$-far from any input in $\Pi$.

\begin{definition}[Testers]
  \label{def:tester}
  An $\eps$-\emph{tester} with \emph{error rate} $\sigma$ for a \emph{property} $\Pi \subseteq \Sigma^n$ is a probabilistic algorithm $T$ that receives query access to a string $x \in \Sigma^n$. The tester $T$ performs at most $q = q(\eps, n)$ queries to $x$ and satisfies the following two conditions.

  \begin{enumerate}
  \item If $x \in \Pi$, then
      $\Pr \left[T^x = 1 \right] \ge 1 - \sigma$.
  \item For every $x$ that is $\eps$-far from $\Pi$ (i.e., $x \in \overline{B_\eps(\Pi)}$), then $\Pr \left[T^x = 0 \right] \ge 1 - \sigma$.
  \end{enumerate}
\end{definition}

We are interested in the regime where $\eps = \Omega(1)$ (i.e., $\eps$ is a fixed constant independent of $n$), and assume it to be the case in the remainder of this discussion.

Note that testers are \emph{not} robust with respect to inputs in the property $\Pi$, as changing a single coordinate of an input $x \in \Pi$ could potentially lead to an input outside $\Pi$. Moreover, an $\eps$-tester does not immediately satisfy one-sided robustness, as inputs that are on the boundary of the $\eps$-neighbourhood of $\Pi$ are not robust (see figure \cref{fig:PT}).

However, by increasing the value of the proximity parameter by a factor of $2$, we can guarantee that every point that is $2\eps$-far from $\Pi$ satisfies the robustness condition. The following claim formalises this statement and shows that testers can be cast as robust local algorithms.

\begin{claim}
	\label{clm:testrla}
	An $\eps$-tester $T$ for property $\Pi \subseteq \Sigma^n$ is an $(\eps,0)$-robust local algorithm, with the same parameters, for computing the function $f$ defined as follows.
	\begin{equation*}
		f(x) = \left\{\begin{array}{ll}1,&\text{if }x \in \Pi\\ 0,&\text{if }x \text{ is }2\eps\text{-far from } \Pi.\end{array}\right.
\end{equation*}
\end{claim}

\begin{proof}
By definition, the tester $T$ is a local algorithm for computing $f$; denote its error rate by $\sigma$. We show it satisfies (one-sided) robustness with respect to $f$. Let $x \in \Sigma^n$ be an input that is $2\eps$-far from $\Pi$, and consider $y \in B_\eps(x)$. By the triangle inequality, we have that $y$ is $\eps$-far from $\Pi$. Thus, $\Pr \left[T^y = 0 \right] \ge 1 - \sigma$, and so $T$ is an $(\eps,0)$-robust local algorithm for $f$.
\end{proof}

\begin{remark}[Robustness vs proximity tradeoff]
	The notion of a tester with proximity parameter $\eps$ and that of an $\eps$-robust tester with proximity parameter $2\eps$ coincide. Moreover, there is a tradeoff between the size of the promise captured by the partial function $f$ and the robustness parameter $\rho$: taking any $\eps' > \eps$, the tester $T$ is a $\rho$-robust local algorithm with $\rho = \eps' - \eps$ for computing the function
	\begin{equation*}
		f(x) = \left\{\begin{array}{ll}1,&\text{if }x \in \Pi\\ 0,&\text{if }x \text{ is }\eps'\text{-far from } \Pi.\end{array}\right.
  \end{equation*}
	As $\eps'$ increases, the robustness parameter $\rho$ increases and the size of the domain of definition of $f$ decreases. In particular, taking $\eps' = \eps$ makes $T$ a $(0,0)$-robust algorithm (i.e., an algorithm that is not robust).
\end{remark}

\subsection{Local codes}
\label{sec:codes-def}
We consider error-correcting codes that admit local algorithms for various tasks, such as testing, decoding, correcting, and computing functions of the message. Recall that a \emph{code} ${C \colon \bitset^k \to \bitset^n}$ is an injective mapping from \emph{messages} of length $k$ to codewords of \emph{blocklength} $n$. The \emph{rate} of the code $C$ is defined as $k/n$. The \emph{relative distance} of the code is the minimum, over all distinct messages $x,y \in\Sigma^k$, of $\dist(C(x),C(y))$. We shall sometimes slightly abuse notation and use $C$ to denote the set of all of its codewords $\set{C(x) : x\in\Sigma^k} \subset \Sigma^n$. Note that we focus on \emph{binary} codes, but remind that the extension to larger alphabets is straightforward.
In the following, we show how to cast the prominent notions of local codes as robust local algorithms.

\subsubsection{Locally testable codes}
Locally testable codes (LTCs) \cite{GS06} are codes that admit algorithms that distinguish codewords from strings that are far from being valid codewords, using a small number of queries.

\begin{definition}[Locally Testable Codes (LTCs)]
	A code $C \colon \bitset^k \to \bitset^n$ is locally testable, with respect to proximity parameter $\eps$ and error rate $\sigma$, if there exists a probabilistic algorithm $T$ that makes $q$ queries to a purported codeword $w$ such that:
	  \begin{enumerate}
  \item If $w = C(x)$ for some $x \in \bitset^k$, then $\Pr \left[T^w = 1 \right] \ge 1 - \sigma$.
  \item For every $w$ that is $\eps$-far from $C$, we have $\Pr \left[T^x = 0 \right] \ge 1 - \sigma$.
  \end{enumerate}
\end{definition}

Note that the algorithm $T$ that an LTC admits is simply an $\eps$-tester for the property of being a valid codeword of $C$. Thus, by \cref{clm:testrla}, we can directly cast $T$ as a robust local algorithm.

\subsubsection{Locally decodable codes}
Locally decodable codes (LDCs) \cite{KT00} are codes that admit algorithms for decoding each individual bit of the message of a moderately corrupted codeword by only making a small number of queries to it.

\begin{definition}[Locally Decodable Codes (LDCs)]
  \label{def:ldc}
  A code $C \colon \bitset^k \to \bitset^n$ is locally decodable with \emph{decoding radius} $\delta$ and error rate $\sigma$ if there exists a probabilistic algorithm $D$ that, given index $i \in[k]$, makes $q$ queries to a string $w$ promised to be $\delta$-close to a codeword $C(x)$, and satisfies
    \begin{equation*}
      \Pr[D^{w}(i) = x_i] \ge 1 - \sigma.
    \end{equation*}
\end{definition}

Note that local decoders are significantly different from local testers and testing in general. Firstly, decoders are given a promise that their input is \emph{close} to a valid codeword (whereas testers are promised to either receive a perfectly valid input, or one that is far from being valid). Secondly, a decoder is given an index as an explicit parameter and is required to perform a different task (decode a different bit) for each parameter (see \cref{fig:LDC}).

Nevertheless, local decoders can also be cast as robust local algorithms. In fact, unlike testers, they satisfy \emph{two-sided} robustness (i.e., both 0-inputs and 1-inputs are robust). In the following, note that since inputs near the boundary of the decoding radius are \emph{not} robust, we reduce the decoding radius by a factor of 2.

\begin{claim}
	\label{clm:LDC}
	A local decoder $D$ with decoding radius $\delta$ for the code $C \colon \bitset^k \to \bitset^n$ is a $(\delta/2,\delta/2)$-robust local algorithm for computing the function $f$ defined as follows.
	\begin{equation*}
		f(z,w) = x_z, \text{ if }x \in \bitset^k \text{ is such that }w \text{ is }\delta/2\text{-close to } C(x).
	\end{equation*}
\end{claim}
\begin{proof}
	Take any $w \in \bitset^n$ that is $\delta/2$-close to $C(x)$. Then, $(w,z)$ is in the domain of definition of $f$ for all explicit inputs $z \in [k]$. Now let $w' \in B_{\delta/2}(w)$ and note that $w'$ is still within the decoding radius of $D$. Hence, the decoder $D$ outputs $x_z$ with probability $1-\sigma$, as required. Moreover, this holds regardless whether $x_z = 0$ or $x_z = 1$, and so $D$ is $(\delta/2,\delta/2)$-robust.
\end{proof}

\begin{remark}[Robustness vs decoding radius tradeoff]
	A local decoder has decoding radius $\delta$ if and only if it is $\delta/2$-robust with decoding radius $\delta/2$, and a tradeoff between promise size and robustness parameter likewise holds in this case: for any $\delta' < \delta$, the decoder $D$ is a $(\delta-\delta', \delta - \delta')$\nobreakdash-robust algorithm for the restriction of $f$ to the $\delta'$-neighbourhood of the code $C$. In particular, $D$ is a $(\delta,\delta)$-robust algorithm with the domain of $f$ defined to be the code $C$.
\end{remark}

\subsubsection{Relaxed locally decodable codes}
Relaxed locally decodable codes (relaxed LDCs) \cite{BGHSV06} are codes that admit a natural relaxation of the notion of local decoding, in which the decoder is allowed to output a special abort symbol $\bot$ on a small fraction of indices, indicating it detected an inconsistency, but never erring with high probability.

\begin{definition}[Relaxed LDCs]\label{def:rldc}
  A code $C\colon\bitset^k \to \bitset^n$ whose distance is $\delta_C$ is a $q$-local relaxed LDC with success rate $\rho$, decoding radius $\delta \in (0,\delta_C/2)$ and error rate $\sigma \in (0, 1/3]$ if there exists a randomised algorithm $D$, known as a \emph{relaxed decoder}, that, on input $i\in [k]$, makes at most $q$ queries to an oracle $w$ and satisfies the following conditions.

  \begin{enumerate}
  \item \textsf{Completeness}: For any $i\in [k]$ and $w = C(x)$, where $x\in\bitset^k$,
  \begin{equation*}
    \Pr[D^{w}(i) = x_i] \ge 1 - \sigma\;.
  \end{equation*}

  \item \textsf{Relaxed Decoding}: For any $i\in [k]$ and $w \in\bitset^{n}$ that is $\delta$-close to a (unique) codeword $C(x)$,
    \begin{equation*}
      \Pr[D^{w}(i) \in \{x_i,\bot\}] \ge 1 - \sigma\;.
    \end{equation*}

    \item \textsf{Success Rate}: There exists a constant $\xi>0$ such that, for any $w\in\bitset^n$ that is $\delta$-close to a codeword $C(x)$, there exists a set $I_w\subseteq [k]$ of size at least $\xi k$ such that for every $i\in I_w$,
    \begin{equation*}
        \Pr[D^w(i)=x_i]\geq 2/3 \;.
    \end{equation*}
  \end{enumerate}
\end{definition}

Note that strictly speaking, the special abort symbol makes it so that relaxed local decoders do not fully fit \cref{def:localalg}, as the input-output mapping $f$ becomes one-to-many. Nevertheless, a simple generalisation of local algorithms, which allows an additional abort symbol, enables us to capture relaxed LDCs as robust local algorithms as well. We show this in \cref{sec:rldc}.

\subsubsection{Locally correctable codes}
The notion of locally correctable codes (LCCs) is closely related to that of LDCs, except that rather than admitting an algorithm that can decode any individual \emph{message} bit, LCCs admit an algorithm that can correct any corrupted \emph{codeword} bit of a moderately corrupted codeword.

\begin{definition}[Locally Correctable Codes (LCCs)]
  \label{def:lcc}
  A code $C \colon \bitset^k \to \bitset^n$ is locally correctable with \emph{correcting radius} $\delta$ and error rate $\sigma$ if there exists a probabilistic algorithm $D$ that, given index $j \in[n]$, makes $q$ queries to a string $w$ promised to be $\delta$-close to a codeword $C(x)$ and satisfies
    \begin{equation*}
      \Pr[D^{w}(j) = C(x)_j] \ge 1 - \sigma.
    \end{equation*}
\end{definition}

A straightforward adaptation of \cref{clm:LDC} yields the following claim.
\begin{claim}
	A local corrector $D$ with correcting radius $\delta$ for the code $C \colon \bitset^k \to \bitset^n$ is a $(\delta/2,\delta/2)$-robust local algorithm for computing the function $f$ defined as follows.
	\begin{equation*}
		f(z,w) = C(x)_z, \text{ if }x \in \bitset^k \text{ is such that }w \text{ is }\delta/2\text{-close to } C(x).
	\end{equation*}
\end{claim}

\subsubsection{Universal locally testable codes}
Universal locally testable codes (universal LTCs) \cite{GG18} are codes that admit local tests for membership in numerous possible subcodes, allowing for testing properties of the encoded message.

\begin{definition}[Universal LTCs]
	A universal LTC $C \colon \bitset^k \to \bitset^n$ for a family of functions $\mathcal{F} = \left\{ f_i : \bitset^k \to \bitset \right\}_{i \in [M]}$ is a code such that for every $i \in [M]$ the subcode $\{ C(x) \: : \: f_i(x) = 1 \}$ is locally testable.
\end{definition}

 Note that ULTCs trivially generalise LTCs, as well as generalise relaxed LDCs (see details in \cite[Appendix A]{GG18}). Since universal testers can be viewed as algorithms that receive an explicit parameter $i \in [M]$ and invoke an  $\eps$-tester for the property $\{ C(x) \: : \: f_i(x) = 1 \}$, then by applying \cref{clm:testrla} to each value of the parameter $i$ they can be cast as robust local algorithms.

\subsection{PCPs of proximity}
\label{sec:PCP-def}
PCPs of proximity (PCPPs) \cite{BGHSV06} are probabilistically checkable proofs wherein the verifier is given query access not only to the proof, but also to the input. The PCPP verifier is then required to probabilistically check whether the statement is correct by only making a constant number of queries to both input and proof.

\begin{definition}
\label{def:pcpp}
	A PCP of proximity (PCPP) for a language $L$ with proximity parameter $\eps$, error rate $\sigma$ and query complexity $q$, consists of a probabilistic algorithm $V$, called the verifier, that receives query access to \emph{both} an input $x \in \Sigma^n$ and a proof $\pi \in \bitset^m$. The verifier $V$ is allowed to make $q$ queries to $(x, \pi)$ and satisfies the following:
	\begin{enumerate}
		\item for every $x \in L$ there exists a proof $\pi$ such that $\Pr[V^{(x,\pi)}=1] \geq 1 - \sigma$; and
		\item for every $x$ that is $\eps$-far from $L$ and every proof $\pi$, it holds that $\Pr[V^{(x,\pi)}=0] \geq 1-\sigma$.
	\end{enumerate}
\end{definition}

We observe that PCPs of proximity with canonical proofs \cite{GS06} (i.e., such that the verifier rejects statement-proof pairs that are far from being the concatenation of a valid statement with a valid proof for it) admit verifiers that are robust local algorithms. Using the tools of \cite{DGG19}, who show that PCPPs can be endowed with the canonicity property at the cost of polynomial blowup in proof length, we can obtain robust local algorithms for general PCPPs.

\begin{claim}
	A PCPP for a language $L \subseteq \Sigma^n$ with proximity parameter $\eps>0$, error rate $\sigma$ and query complexity $q$ can be transformed into a PCPP for $L$ with proximity parameter $2\eps$, whose verifier is an $(\eps,0)$-robust local algorithm with the same query complexity and error rate.
\end{claim}
\begin{proof}[Sketch of proof.]
	Let $V$ be a PCPP verifier with proximity parameter $\eps$ and error rate $\sigma$ for $L \subseteq \Sigma^n$, that makes at most $q$ queries to its input-proof pair $(x, \pi) \in \Sigma^n \times \bitset^m$. By \cite[Section 3]{DGG19}, there exists a PCPP verifier $V'$ for $L$ with $\poly(m)$ proof length (as well as proximity parameter $\eps$, error rate $\sigma$ and query complexity $q$) that satisfies the following strengthening of the conditions in \cref{def:pcpp}: there is a set of \emph{canonical proofs} $\set{\pi_x}_{x \in L}$ such that
	\begin{enumerate}[noitemsep]
		\item for every $x \in L$, it holds that $\Pr[V^{(x,\pi_x)}=1] \geq 1 - \sigma$; and
	  \item for every $(x,\pi)$ that is $\eps$-far from $(y, \pi_y)$ for all $y \in L$, it holds that $\Pr[V^{(x,\pi)}=0] \geq 1-\sigma$.
	\end{enumerate}
	In other words, $V'$ is an $\eps$-tester for the property $\Pi \coloneqq \set{(x, \pi_x) : x \in L}$, and we invoke \cref{clm:testrla}.
\end{proof}

\paragraph{Non-interactive proofs of proximity\ifsicomp\else.\fi} MA proofs of proximity (MAPs) \cite{GR18,FGL14} are proof systems that can be viewed as a property testing analogue of NP proofs. The setting of MAPs is very similar to that of PCPPs, with the distinction that the purported proof is of sublinear size and is given explicitly, i.e., the MAP verifier can read the entire proof. We remark that an equivalent description of a MAP as a covering by partial testers \cite{FGL14} is used in this work, so that \emph{every fixed proof string} defines a tester, and \cref{clm:testrla} applies. We cover this in \cref{sec:map}.

\section{Technical lemmas}
\label{sec:lemmas}
In the section we provide an arsenal of technical tools for analysing robust local algorithms, which we will then use to prove our main result in \cref{sec:proof}. The order in which we present the tools is according to their importance, starting with the most central lemmas.

Specifically, in \cref{sec:daisy} we discuss the notion of relaxed sunflowers that we shall need, called daisies, then state and prove a daisy partition lemma for multi-collections of sets.
In \cref{sec:HS}, we apply the Hajnal-Szemer\'{e}di theorem to derive a sampling lemma for daisies.
In \cref{sec:volume}, we prove a simple yet vital volume lemma for robust local algorithms, which will be used throughout our analysis.
Finally, in \cref{sec:preprocessing} we adapt generic transformations (amplification and randomness reduction) to our setting of robust local algorithms.

\subsection{Relaxed sunflowers}
\label{sec:daisy}

We discuss the central technical tool used in the transformation to sample-based algorithms, which is a relaxation of combinatorial sunflowers, referred to as \emph{daisies} \cite{FLV15,GL21}. We extend the definition of daisies to multi-sets, then state and prove the particular variant of a daisy lemma that we shall need.

\begin{definition}[Daisy]
	\label{def:relaxed-sunflower}
	Suppose $\sets$ is a multi-collection of subsets of $[n]$ (i.e., subsets may repeat).
	$\sets$ is an $h$-daisy (where $h\colon\N\rightarrow \N$) with petals of size $j$ and \emph{kernel} $K \subseteq [n]$ if the following holds: every $S\in \sets$ has a \emph{petal} $S\setminus K$ with $\abs{S\setminus K}=j$ and, for every $k \in [j]$, there exists a subset $P_k \subseteq S \setminus K$ with $\abs{P_k} \geq k$ whose elements are contained in at most $h(k)$ sets from $\sets$.
	
	\noindent A daisy with pairwise disjoint petals ($1$-daisy) is referred to as a \emph{simple daisy}.
\end{definition}

We remark that the notion of a daisy relaxes the standard definition of a sunflower in two ways: (1) the kernel is not required to equal the pairwise intersection of all sets in the collection, its structure is unconstrained; and (2) the \emph{petals} $\petals = \set{S \setminus K : S \in \daisy}$ need not be pairwise disjoint, but rather, each point outside of the kernel can be contained in at most $h(j)$ sets of $\daisy$; see \cref{fig:daisy}. Note that \cref{def:relaxed-sunflower}, in contrast to sunflowers (for which pairwise disjointness disallows multiple copies of a same set), applies to \emph{multi-sets}.

These relaxations, unlike in the case of sunflowers, allow us to arbitrarily partition any collection of subsets into a collection of daisies with strong structural properties, as \cref{thm:daisy-partition} shows.
\begin{lemma}[Daisy partition lemma for multi-collections]
\label{thm:daisy-partition}
	Let $\sets$ be a multi-collection of $q$-sets of $[n]$, and define the function $h\colon\N\rightarrow\N$ as follows:
	\begin{equation*}
		h(k) =  n^{\frac{\max\{1,k-1\}}{q}} \:.
	\end{equation*}
	Then, there exists a collection $\set{\daisy_j : 0 \leq j \leq q}$ such that
	\begin{enumerate}
		\item $\set{\daisy_j}$ is a partition of $\sets$, i.e., $\bigcup_{j=0}^{q}\daisy_j = \sets$ and $\daisy_j \cap \daisy_k = \emptyset$ when $j \neq k$.
		\item For every $0 \leq j \leq q$, there exists a set $K_j \subseteq [n]$ of size $\abs{K_j} \leq q\abs{\sets}\cdot n^{-\max\{1,j\}/q}$ such that $\daisy_j$ is an $h$-daisy with kernel $K_j$ and petals of size $j$. Moreover, the kernels form an incidence chain $\varnothing = K_q \subseteq K_{q-1} \subseteq \cdots \subseteq K_1 = K_0$.
	\end{enumerate}
\end{lemma}
\begin{proof}
	We construct the collections $\set{\daisy_j : 0 \leq j \leq q}$ in a greedy iterative manner, as follows.

	\begin{enumerate}
		\item Define $\sets_0 \coloneqq \sets$.
		\item Inductively define, for each $0 \leq j \leq q-1$:
		\begin{enumerate}
		\item \emph{Kernel construction:} Define $K_j$ as the set of points in $[n]$ that are contained in at least $h(j+1)$ sets from $\sets$.
		\item \emph{Daisy construction:} Set $\daisy_j$ to be all the sets $S \in \sets_j$ such that $\abs{S \setminus K_j} = j$.
		\item Set $\sets_{j+1}$ to be $\sets_j \setminus \daisy_j$.
		\end{enumerate}
		\item Finally, set $\daisy_q = \sets_{q}$ and $K_q = \emptyset$.
	\end{enumerate}

	We now prove that this construction yields daisies with the required properties. For ease of notation, let $d_i$ be the number of sets in $\sets$ containing $i$ for each $i \in [n]$.
	By definition, $\sets_q\subseteq \sets_{q-1}\subseteq\dots\subseteq\sets_0$ and $\daisy_j \subseteq \sets_j$ for all $j$. Since $\daisy_{j-1} \cap \sets_j = \varnothing$ for all $j \in [q]$, it follows that $\daisy_j \cap \daisy_k = \varnothing$ when $j \neq k$. Also, since $\sets = \sets_q \cup \bigcup_{0 \leq j < q} \daisy_j$ and $\sets_q = \daisy_q$, we have $\sets =  \bigcup_{0 \leq j \leq q} \daisy_j$.

	Since $\sets$ is comprised of $q$-sets,
	\begin{equation*}
		\sum_{i \in [n]}d_i = q  \abs{\sets}.
	\end{equation*}
	By the kernel construction, for each $j\in \{0,1,\dots,q\}$, $K_j$ is the set of all $i \in [n]$ such that $d_i \geq h(j+1)$, which implies $\sum_{i \in K_j}d_i \geq \abs{K_j}\cdot h(j+1)$. Therefore,
	\begin{equation*}
		q \abs{\sets} = \sum_{i \in [n]} d_i \geq \sum_{i \in K_j}d_i \geq \abs{K_j}\cdot h(j+1)
	\end{equation*}
	and thus $\abs{K_j} \leq q \abs{\sets} / h(j+1) = q \abs{\sets}\cdot n^{-\max\set{1,j}/q}$. Note, also, that the kernel construction ensures not only $K_j \subseteq K_{j-1}$ when $j \in [q]$, but also $K_0 = K_1$ because $h(1) = h(2)$.

	Since the petals of each $S \in \daisy_j$ have size exactly $j$ by construction, all that remains to be proven is the intersection condition on the petals that makes $\daisy_j$ an $h$-daisy; namely, that for every $k \in [j]$, there exists a subset $P_k \subseteq S \setminus K$ with $\abs{P_k} \geq k$ whose elements are contained in at most $h(k)$ sets from $\daisy_j$. Assume for the sake of contradiction that this condition does not hold.

	Let $j \in [q]$ be the smallest value for which $\daisy_j$ is not an $h$-daisy and $S \in \daisy_j$ be a set that violates the intersection condition; then take $k \leq j$ to be the smallest subset size such that every $P_k \subseteq S \setminus K_j$ with size $k$ has an element $i \in P_k$ with $d_i > h(k)$ (equivalently said, $j$ and $k$ are minimal such that the subset $L \subset S \setminus K_j$, comprised of all $i$ with $d_i \leq h(k)$, has size at most $k - 1$).

	Suppose first that $k = 1$. Then, every $i \in S\setminus K_j$ is such that $d_i > h(2) = h(1)$ and thus $S \setminus K_j \subseteq K_0$. But this implies $S \in \daisy_0$ (since $\abs{S \setminus K_0} = 0$), a contradiction because the intersection condition holds by vacuity on empty petals.

	Suppose now that $k > 1$.	The subset
	\begin{equation*}
		L \coloneqq \set{i \in S \setminus K_j : d_i \leq h(k)}
	\end{equation*}
	contains at most $k-1$ points; however, by minimality of $k$, at least $k-1$ distinct points $i \in L$ satisfy $d_i \leq h(k-1) \leq h(k)$. Therefore, $\abs{L} = k - 1$ and
	\begin{equation*}
		L = \set{i \in S \setminus K_j : d_i \leq h(k-1)}\;.
	\end{equation*}
	By the definition of $L$, every $i \in S \setminus (K_j \cup L)$ satisfies $d_i > h(k)$, so that $i \in K_{k-1}$; therefore, $S \subseteq K_{k-1} \cup K_j \cup L$. Since the kernels form an incidence chain, $K_j \subseteq K_{k-1}$ and thus $S \setminus K_{k-1} = L$. But then $\abs{S \setminus K_{k-1}} = \abs{L} = k - 1$, so that $S \in \daisy_{k-1}$, contradicting the fact that $S \in \daisy_j$ (because $k - 1 < j$ and $\set{\daisy_j}$ is a partition).
\end{proof}

The following claim shows an upper bound on the total number of sets in an $h$-daisy that may intersect a given petal. It will be useful in order to partition a daisy into \emph{simple} daisies, as the next section will show.

\begin{claim}
  \label{clm:hbound}
	Let $\sets$ be a multi-collection of $q$-sets and $\set{\daisy_j : 0 \leq j \leq q}$ be a daisy partition obtained by an application of \cref{thm:daisy-partition}. Then, for every $j \in [q]$ and $S \in \daisy_j$, the number of sets in $\daisy_j$ whose petals intersect $S \setminus K_j$ (including $S$ itself) is at most $2 h(j)=2n^{\max\{1,j-1\}/q}$.
\end{claim}
\begin{proof}
	Let $S$ be an arbitrary set in $\daisy_j$.
	We name the elements in $S \setminus K_j$ by $u_1,u_2,\dots,u_j$ (by \cref{thm:daisy-partition}, every $S \in \daisy_j$ satisfies $\abs{S \setminus K_j} =j$).
	For every $k\in [j]$, let $d_k$ be the number of sets of $\daisy_j$ that $u_k$  is a member of.

	Assume without loss of generality that $d_k \leq d_\ell$ for every $k$ and $\ell$ in $[j]$ such that $k<\ell$, as otherwise we can rename $u_1,u_2,\dots,u_j$ so that this holds.

	By the definition of an $h$-daisy, for every $\ell\in [j]$, there exists a set of $\ell$ elements $k \in [j]$ that satisfy $d_k \leq h(\ell)$.
	Thus, $[\ell]$ is such a set and we know that $d_{\ell} \leq h(\ell)$.
	As the number of petals that intersect $S \setminus K_j$ is at most $\sum_{k=1}^{j} d_{k}$, we get that
	\begin{equation*}
		\sum_{k=1}^{j} d_{k}\leq\sum_{k = 1}^j h(k) = \sum_{k = 1}^j n^{\frac{\max\set{1,k-1}}{q}} \leq 2h(j).
	\end{equation*}
	The last equality follows directly from the value of $h(k)$.
\end{proof}

\subsection{Sampling daisies and the Hajnal-Szemer\'{e}di theorem}
\label{sec:HS}

Concentration of measure is an essential ingredient in our proofs, which we first illustrate via a simplified example. Consider a collection of singletons that comprise the petals of a combinatorial sunflower: sets $P_1, \ldots P_k$, all disjoint and of size 1, contained in the ground set $[n]$. If we perform binomial sampling of the ground set (sampling each $i \in [n]$ independently with probability $p$), the Chernoff bound ensures that the number of sampled petals is close to its expectation. Defining $X_i$ as the random variable that indicates whether $P_i$ was sampled, we have lower and upper tail bounds that guarantee the number of queried petals is concentrated around $pn$ except with exponentially small probability. Note, too, that the same holds for larger petals: if $P_i$ is a $j$-set for all $i$, the number of queried petals is concentrated around $p^j n$.

Now consider the case where $P_1, \ldots, P_k$ are petals of a \emph{daisy}. In this case the Chernoff bound does not apply, since the indicator random variables $X_i$ are no longer independent; however, the structure of a daisy ensures there is not too much intersection among these petals, which gives means to control the correlation between these random variables. It is thus reasonable to expect that \emph{sampling a daisy is similar to sampling a sunflower}. This intuition is formalised by making use of the Hajnal-Szemer\'{e}di theorem \cite{HS70}, which we state next.

\begin{theorem}
  Let $G$ be a graph with $m$ vertices and maximum degree $\Delta$. Then, for any $k \geq \Delta + 1$, there exists a $k$-colouring of the vertices of $G$ such that every colour class has size either $\lfloor m/k\rfloor$ or $\lceil m/k\rceil$.
\end{theorem}

We remind that integrality does not cause issues in our analyses, and we thus assume all colour classes have size $n/k$. By encoding the sets of a daisy as the vertices of an ``intersection graph'', the fact that petals have bounded intersection translates into a graph with bounded maximum degree. Applying the Hajnal-Szemer\'{e}di theorem to this graph, we are able to partition the original daisy into a small number of large \emph{simple} daisies.

\begin{lemma}
	\label{lem:simple-daisy}
  A daisy $\daisy$ with kernel $K$, such that each one of its petals has a non-empty intersection with at most $t - 1$ other petals, can be partitioned into $t$ simple daisies with the same kernel, each of size $\abs{\daisy}/t$.
\end{lemma}
\begin{proof}
  Construct a graph $G$ with vertex set $\daisy$ by placing an edge between vertices $S$ and $S'$ when ${(S \cap S') \setminus K \neq \varnothing}$. By definition, the maximum degree of $G$ is $\Delta(G) \leq t - 1$. The Hajnal-Szemer\'{e}di theorem implies $G$ is colourable with $t$ colours, where each colour class has size $\abs{G}/t$. This partition of the vertex set corresponds to a partition of the daisy $\daisy$ into simple daisies $\set{\sets_j : j \in [t]}$, each of size $\abs{\daisy}/t$.
\end{proof}

\subsection{The volume lemma}
\label{sec:volume}

This section proves a key lemma that makes use of daisies to prove a certain structure on the sets that a \emph{robust} local algorithm may query. Loosely speaking, the volume lemma ensures that in order for a collection of sets to be queried with high enough probability, it must cover a sufficiently large fraction of the input's coordinates.

Let $M$ be a $q$-local algorithm that computes a partial function $f$ with error rate $\sigma$ (we assume the explicit input to be fixed and omit it). Recall that, for each input $x \in \Sigma^n$, the algorithm $M$ queries according to a distribution $\mu_x$ over a multi-collection of $q$-sets, as defined in \cref{def:distb}.

\begin{lemma}[Volume lemma]
  \label{lem:volume}
  Fix $x \in \Sigma^n$ in the domain of $f$. If there exists a $\rho$-robust $y \in \Sigma^n$ such that $f(y) \neq f(x)$, then every collection $\sets \subseteq \supp(\mu_x)$ such that $\abs{\cup \sets} = \abs{\cup_{S \in \sets} S} < \rho n$ satisfies $\mu_x(\sets) \leq 2 \sigma$.
\end{lemma}
\begin{proof}
  Suppose, by way of contradiction, that there exists $\sets \subseteq \supp(\mu_x)$ such that $\mu_x(\sets) > 2 \sigma$ and $\abs{\cup \sets} < \rho n$.

  For notational simplicity, assume without loss of generality that $f(x) = 1$, and take a $\rho$-robust $y \in \Sigma^n$ such that $f(y) = 0$. Define $w$ to match $x$ in the coordinates covered by $\cup \sets$, and to match $y$ otherwise. Then $w$ is $\rho$-close to $y$, so that $M$ outputs $0$ when its input is $w$ with probability at least $1-\sigma$.

	When the algorithm samples a decision tree that makes it query $S \in \sets$, then $M$ behaves \emph{exactly} as it would on input $x$, which happens with probability at least $\mu_x(\sets) > 2 \sigma$. But the algorithm outputs 1 on input $x$ with probability at most $\sigma$, and thus outputs 0 on input $w$ with probability greater than $\sigma$, in contradiction with the robustness of $y$.
\end{proof}

\begin{remark}
  The volume lemma requires an arbitrary $\rho$-robust $y$ with $f(y) \neq f(x)$. It thus suffices that \emph{a single} such $\rho$-robust point exists for the volume lemma to hold \emph{for every $x'$ such that $f(x') = f(x)$}.
\end{remark}

\subsection{Generic transformations}
\label{sec:preprocessing}

This section provides two standard transformations that improve parameters of an algorithm: error reduction (\cref{clm:err-red}) and randomness reduction (\cref{clm:rand-red}), which, applied in conjunction, imply \cref{lem:prep-alg}. These apply generally to randomised algorithms for decision problems, and, when applied to robust local algorithms, \emph{both transformations compute the same function and preserve robustness}. We defer the proofs in this section to \cref{sec:deferred}.

The following claim is an adaptation of a basic fact regarding randomised algorithms: performing independent runs and selecting the output via a majority rule decreases the error probability exponentially.

\begin{claim}[Error reduction]
\label{clm:err-red}
Let $M$ be a $(\rho_0, \rho_1)$-robust algorithm for $f\colon \mathcal{P} \rightarrow \bitset$ (where $\mathcal{P} \subset Z \times \Sigma^n$) with error rate $\sigma \leq 1/3$, query complexity $q$ and randomness complexity $r$.

For any $\sigma' > 0$, there exists a $(\rho_0, \rho_1)$-robust algorithm $N$ for computing the same function with error rate $\sigma'$, query complexity $O(q \log(1/\sigma')/\sigma)$ and randomness complexity $O(r \log(1/\sigma')/\sigma)$.
\end{claim}

Next, we state a transformation that yields an algorithm with twice the error rate and significantly reduced randomness complexity. This, in turn, provides an upper bound on the number of $q$-sets queried by the algorithm, such that an application of \cref{thm:daisy-partition} to this multi-collection yields daisies with kernels of sublinear size. Such a bound on the size of the kernels is crucial to ensure correctness of the sample-based algorithm we construct in \cref{sec:construction}. Our proof adapts the technique of Goldreich and Sheffet \cite{GS10}, which in turn builds on the work of Newman \cite{N91a}.

\begin{claim}[Randomness reduction]
\label{clm:rand-red}
  Let $M$ be a $(\rho_0, \rho_1)$-robust algorithm for $f\colon \mathcal{P} \rightarrow \bitset$ (where $\mathcal{P} \subset Z \times \Sigma^n$) with error rate $\sigma$, query complexity $q$ and randomness complexity $r$.

  There exists a $(\rho_0, \rho_1)$-robust algorithm $N$ for computing the same function with error rate $2\sigma$ and query complexity $q$, whose distribution $\Tilde{\mu}^N$ has support size $3n \ln \abs{\Sigma}/\sigma$. In particular, the randomness complexity of $N$ is bounded by $\log (n/\sigma) + \log \log \abs{\Sigma} + 2$.
\end{claim}

In the next section, we need a combination of error reduction and randomness reduction, which the following lemma provides.

\newcommand{\errorrateInv}{4q}
\newcommand{\errorrate}{1/(\errorrateInv)}
\newcommand{\errorratefrac}{\frac{1}{\errorrateInv}}
\newcommand{\suppsize}{6n \ln \abs{\Sigma}/\sigma}
\newcommand{\suppsizefrac}{\frac{6n \ln \abs{\Sigma}}{\sigma}}
\newcommand{\suppsizeq}{48q \ln \abs{\Sigma} n}

\begin{lemma}
  \label{lem:prep-alg}
  Assume there exists a $\rho$-robust algorithm $M$ for computing $f$ with query complexity $\ell$, error rate $1/3$ and arbitrary randomness complexity.
  Then there exists a $\rho$-robust $q$-local algorithm $M'$ for $f$ with error rate
	\begin{equation*}
		\sigma = \errorratefrac
	\end{equation*}
	such that $\frac{q}{\log 8q} = O(\ell)$, or, equivalently,
  \begin{equation*}
    q = O(\ell \log \ell).
  \end{equation*}
  Moreover, the distribution of $M'$ is uniform over a multi-collection of decision trees of size $\suppsize$.
\end{lemma}

\section{Proof of \texorpdfstring{\cref{infthm:main}}{Theorem \ref{infthm:main}}}
\label{sec:proof}

This section contains the main technical contribution of our work: a proof that every robust local algorithm with query complexity $q$ can be transformed into a \emph{sample-based} local algorithm with sample complexity $ n^{1- 1/O(q^2 \log^2 q)}$. We begin by providing a precise statement of \cref{infthm:main}. In the following, we remind that when the error rate of an algorithm is not stated, it is assumed to be $1/3$.

\newcommand{\gammaValue}{24 \abs{\Sigma}^q \ln \abs{\Sigma}}
\begin{theorem}[\cref{infthm:main}, restated]
  \label{thm:main}
  Suppose there exists a $(\rho_0, \rho_1)$-robust local algorithm $M$ for the function $f\colon \mathcal{P} \rightarrow \bitset$ (where $\mathcal{P} \subset Z \times \Sigma^n$) with query complexity $\ell$ and $\max\set{\rho_0, \rho_1} = \Omega(1)$.
  Then, there exists a \emph{sample-based} algorithm $N$ for $f$ with sample complexity $\gamma\cdot n^{1-1/2q^2}$,
  where $q = O(\ell \log \ell)$ and $\gamma = O(\abs{\Sigma}^q \ln \abs{\Sigma})$.
\end{theorem}

Note that when $q = \Omega(\sqrt{\log n})$ or $\abs{\Sigma}^{q} = \Omega\big(n^{1/2q^2}\big)$, the algorithm we obtain samples linearly many coordinates, and the statement becomes trivial.

Therefore, hereafter we assume that the query complexity of $M$ satisfies $\ell \leq \sqrt[5]{\log n}$ (so $q = \Theta(\sqrt[5]{\log n} \log\log n) = o(\sqrt{\log n})$) and the alphabet size satisfies $\abs{\Sigma} \leq 2^{\sqrt[6]{\log n}}$ (so $\abs{\Sigma}^q \leq n^{1/q^3}$).

We proceed to prove \cref{thm:main}. Specifically, in \cref{sec:construction} we construct a sample-based local algorithm $N$ from the $(\rho_0, \rho_1)$-robust local algorithm $M$ in the hypothesis of \cref{thm:main}; in \cref{sec:analysis}, we analyse our sample-based algorithm $N$; and in \cref{sec:conclude} we conclude the proof by showing the lemmas proved throughout the analysis indeed imply correctness of $N$.

\subsection{Construction}
\label{sec:construction}
\newcommand{\query}{\ell \log \ell}
\newcommand{\querysq}{\ell^2 \log^2 \ell}

\newcommand{\capparam}{12\ln\abs{\Sigma} / (\rho\cdot\sigma)}
\newcommand{\capparamfrac}{\frac{12\ln\abs{\Sigma}}{\rho\cdot\sigma}}
\newcommand{\capparamInv}{\frac{\rho\cdot\sigma}{12\ln\abs{\Sigma}}}

\newcommand{\threshold}[1]{\frac{\abs{\Tilde{\mu}}}{2q} \cdot p^{#1}}
\newcommand{\thresholdx}[1]{\frac{\abs{\mu_x}}{2q} \cdot p^{#1}}
\newcommand{\sizeKj}[1]{12q^2\abs{\Sigma}^q\cdot n^{1-{#1}/q}}
\newcommand{\sizeTau}{12qn\abs{\Sigma}^q}
\newcommand{\sampprob}{\gamma\cdot n^{-1/2q^2}}

Hereafter, let $f\colon \mathcal{P} \to \bitset$ be the function in the hypothesis of \cref{thm:main}.
As the following treatment is the same for all explicit inputs $z \in Z$, we assume it to be fixed and omit it from the notation.
We also assume without loss of generality that $\rho_0$ is a constant strictly greater than $0$ (if this is not the case we simply exchange the 0 and 1 values in the truth table of $f$). We set $\rho = \rho_0$.

Let $M$ be the algorithm in the hypothesis of \cref{thm:main}. We apply \cref{lem:prep-alg} and obtain a $(\rho_0, \rho_1)$-robust local algorithm $M'$ for the same problem, with query complexity $q = O(\query)$ and error rate $\sigma = 1/4q$. The algorithm $N$ we describe below has white box access to the local algorithm $M'$.
We next explain how it extracts information from it.

Upon execution, $M'$ chooses a decision tree uniformly at random according to the outcome of its coin flips; this uniform distribution is denoted $\Tilde{\mu} = \Tilde{\mu}^{M'}$, whose support size is $\abs{\Tilde{\mu}}$.
For every decision tree and every one of its branches, define a \emph{description tuple} $(S, a_S, b, s)$, where $s$ is the random string that will cause the use of this tree, $S$ is the set of all the queries in this branch, $a_S$ is the assignment to these queries that will result in $M'$ using this specific branch and $b$ is the value $M'$ returns when this occurs (as per \cref{def:distb}).

We assume that for every description-tuple $(S, a_S, b, s)$ the size of $S$ is exactly $q$.  This can be assumed without loss of generality since it is possible to convert $M'$ into an algorithm such that every decision tree and every one of its branches makes $q$ distinct queries: if the same query appears more than once on a branch of a tree, all but the first appearance can be removed by choosing the continuation that follows the (already known) value that leads to the algorithm using this branch. In addition, a tree can be expanded by adding queries, so that every branch has exactly $q$ distinct queries. Both of these changes do not change any of the parameters of the algorithm beyond ensuring that it will query exactly $q$ coordinates.

The algorithm $N$ we describe next only makes use of description tuples $(S, a_S, b, s)$ such that $b=1$. To this end we set
\begin{equation*}
\tuples = \{(S, a_S, b, s): (S, a_S, b, s) \text{ is a description tuple such that } b = 1\}.
\end{equation*}
Algorithm $N$ also requires the multi-collection $\sets$ defined as follows:
\begin{equation*}
\sets  = \{S: (S, a_S, b, s) \in \tuples\}.
\end{equation*}
Specifically, it applies \cref{thm:daisy-partition} to get a daisy partition of $\sets$.
When the algorithm extracts $\tuples$ and $\sets$ from $M'$ and computes a daisy partition for $\sets$, it
preserves the information that allows it to associate the set of a tuple of $(S, a_S, b, s)$  to the unique daisy $S$ is contained in.

The construction proceeds in two stages: \emph{preprocessing} and \emph{execution}. Recall that, for any input $x \in \Sigma^n$ and assignment $\kappa$ to a subset $S \subseteq [n]$, we denote by $x_\kappa$ the word that assigns the same values as $\kappa$ on $S$ and the same values as $x$ on $[n] \setminus S$.

\paragraph{Preprocessing\ifsicomp\else.\fi} $N$ has access to $M'$, with which it computes $\tuples$ and $\sets$. Applying \cref{thm:daisy-partition} to $\sets$, the algorithm obtains the \emph{daisy partition}
\begin{equation*}
	\daisy = \set{\daisy_j : 0 \leq j \leq q},
\end{equation*}
so that each tuple in $\tuples$ is associated with $\daisy_j$ for exactly one $j \in \set{0, \ldots, q}$. Set
\begin{equation*}
	p \coloneqq \sampprob,
\end{equation*}
the \emph{sampling probability}, where $\gamma = \gammaValue$; and, for every $j \in [q]$, set
\begin{equation*}
	\tau_j \coloneqq \threshold{j},
\end{equation*}
the \emph{thresholds}, which will be used in the execution stage.

\paragraph{Execution\ifsicomp\else.\fi} When $N$ receives query access to a string $x \in \Sigma^n$, it performs the following sequence of steps.

\begin{enumerate}[ref={Step~\arabic*}]
    \item\label{step:sample}\emph{Sampling:} Select each element in $[n]$ independently with probability $p$. Denote by $Q$ the set of all coordinates thus obtained. If $\abs{Q} \geq 2 pn$, then $N$ outputs arbitrarily. Otherwise, $N$ queries all the coordinates in $Q$.

    \item\label{step:enumeration} \emph{Enumeration:} For every $j \in [q]$ and kernel assignment $\kappa$ to $K_j$,\footnote{Note that the algorithm \emph{does not} iterate over the case $j = 0$. We will show in \cref{sec:analysis} that this has a negligible effect.} perform the following steps. Set a counting variable $v$ to 0 before each iteration.
    \begin{enumerate}
      \item\label{step:votecount} \emph{Local view generation and vote counting:}  For every tuple $(S, a_S, 1, s) \in \tuples$ such that $S \in \daisy_j$, increment $v$ if $S \subset Q \cup K_j$ and $a_S$ assigns on $S$ the same values as $x_{\kappa}$ does.

		    In the case $j = 1$, if $\capparam$ sets have the same point outside $K_1$, disregard them in the count.\footnote{This is required for technical purposes when dealing with $K_1$.}
      \item\label{step:threshold} \emph{Threshold check:} If $v \geq \tau_j$, output $1$.
    \end{enumerate}
    \item\label{step:outzero} If the condition $v \geq \tau_j$ was never satisfied, then output $0$.
\end{enumerate}

We proceed to analyse this construction.

\subsection{Analysis}
\label{sec:analysis}

We remind that the explicit input $z$ is assumed to be fixed and is omitted from the notation. For the analysis we are interested in the behaviour of the algorithm $M'$ on a fixed input $x$.
For this purpose, we use the distribution $\mu_x$ from \cref{def:distb}.

For $x\in \Sigma^n$ we define $\mu_x$ to be the uniform distribution over the multi-collection of sets
\begin{equation}\label{equation:mu_xSet}
 \set{S: (S, a_S, b, s) \text{ is a description tuple such that } a_S = x_{|S}},
\end{equation}
where a description tuple is as appears in \cref{sec:construction}.
We note that this implies that $\supp(\mu_x)$ has exactly one set for each decision tree $M'$ may use, since when both the randomness and the input are fixed exactly one branch of the decision tree is used by $M'$. Therefore,
\begin{equation*}
\abs{\mu_x} = \abs{\Tilde{\mu}} \:.
\end{equation*}

We now list the relevant parameters in the analysis with reference to where they are obtained. By \cref{lem:prep-alg},
\begin{equation}\label{equation:sigmaVal}
	\sigma = \errorratefrac \:,
\end{equation}
and, for every $x \in \Sigma^n$,
\begin{equation}\label{equation:mux}
	\abs{\mu_x} = \abs{\Tilde{\mu}} = \suppsizefrac \:.
\end{equation}
The construction of $N$ in the previous section sets the parameters
\begin{equation}\label{equation:gammavalue}
	\gamma = \gammaValue \:,
\end{equation}
\begin{equation}\label{equation:pvalue}
	p = \sampprob \:,
\end{equation}
and, for all $j \in [q]$,
\begin{equation}\label{equation:taujvalue}
	\tau_j = \threshold{j} = \thresholdx{j} \:.
\end{equation}
(The second equality holds for all $x \in \Sigma^n$ by \cref{equation:mux}.) Finally, the size of the collection of tuples $\tuples$, which by the construction in \cref{sec:preprocessing} is the same as that of $\sets$, is bounded by the total number of branches over all decision trees in $\supp(\Tilde{\mu})$. Thus,
\begin{equation}\label{equation:tuples}
	\abs{\sets} = \abs{\tuples} \leq \abs{\Sigma}^q \cdot \abs{\Tilde{\mu}} = \abs{\Sigma}^q \cdot \abs{\mu_x} \:,
\end{equation}
for every $x \in \Sigma^n$.

For our result we need an upper bound on the sizes of the kernels that algorithm $N$ enumerates over, which we show next.
\newcommand{\kernelSizeUB}[1]{\gamma\cdot q^2\cdot n^{1-\max\{1,{#1}\}/q}}
\newcommand{\kernelNonzeroSizeUB}[1]{\gamma\cdot q^2\cdot n^{1-{#1}/q}}
\newcommand{\kernelZeroSizeUB}{\gamma\cdot q^2 \cdot n^{1-1/q}}
\begin{claim}\label{clm:kernels}
	Let $\set{K_i : 0 \leq i \leq q}$ be the kernels of the daisy partition $\set{\daisy_i}$ of $\sets$ used by the algorithm $N$. For every $i \in \set{0,1,\dots,q}$, the kernel $K_i$ is such that $\abs{K_i} \leq \kernelSizeUB{i}$ and, for $n$ sufficiently large, $\abs{K_i} < \rho n/2$.
\end{claim}
\begin{proof}
By~\cref{thm:daisy-partition},	for every $i \in \set{0,1,\dots,q}$,
\begin{flalign*}
&&\abs{K_i} &\leq q\abs{\sets}n^{-\max\{1,i\}/q} &&\\
&&&\leq q\abs{\Sigma}^q \abs{\mu_x} \cdot n^{-\max\{1,i\}/q} &&\left(\text{by \cref{equation:tuples}, } \abs{\sets}\leq \abs{\Sigma}^q \cdot \abs{\mu_x}\right)\\
&&&\leq q\abs{\Sigma}^q \cdot \suppsizefrac \cdot n^{-\max\{1,i\}/q} &&\left(\text{by \cref{equation:mux}, }\abs{\mu_x}\leq \suppsizefrac \right)\\
&&&= 24\abs{\Sigma}^q \cdot \ln\abs{\Sigma} \cdot q^2 \cdot n^{-\max\{1,i\}/q} &&\left(\text{by \cref{equation:sigmaVal}, } \sigma = \errorratefrac\right)\\
&&&=\kernelSizeUB{i} &&\left(\text{by \cref{equation:gammavalue}, }\gamma = \gammaValue \right)
\end{flalign*}
It remains to prove the second part of the claim.
By the calculation above, since $\rho$ is constant and $\abs{\Sigma}^q \ln\abs{\Sigma}\cdot q^2 = o(n^{-1/q})$ (recall that $\abs{\Sigma} \leq n^{1/q^4}$ and $q$ is sub-logarithmic), for sufficiently large $n$,
\begin{equation}
\abs{K_0} \leq \kernelZeroSizeUB = \left(\gammaValue \cdot q^2 \cdot n^{-1/q}\right)n\leq \rho n/2  .
\end{equation}
By \cref{thm:daisy-partition}, $K_q\subseteq K_{q-1},\dots,K_1 = K_0$, and hence the claim follows.
\end{proof}

Next, we provide a number of definitions emanating from algorithm $N$.
We define, for every $x \in \Sigma^n$, the multi-collection
\begin{align*}
	\ones^x \coloneqq \set{S : (S, a_S, 1, s) \in \tuples \text{ and } x_{|S} = a_S},
\end{align*}
where $\tuples$ is defined as in \cref{sec:construction}. Note that the definition of this collection depends only on the algorithm $M$ and not on the function $f$ it computes. Hence, it is well-defined for every $x$, and in particular for points that are $\rho$-close to a $\rho$-robust point of the domain (where $f$ may not be defined).
We note that, since $\mu_x$ is defined over the collection in \cref{equation:mu_xSet} we know that
\begin{equation}\label{equation:onesxMux}
\ones_x \subseteq \supp(\mu_x) \:.
\end{equation}

Since the ``capping parameter'' $\capparam$ is used numerous times, we set
\newcommand{\capa}{\alpha}
\begin{equation}\label{equation:capparam}
\capa = \capparamfrac \:.
\end{equation}

We refer to the act of incrementing $v$ as counting a vote.
For each $j \in [q]$, we define the \emph{vote counting function} $v_j\colon \Sigma^n \rightarrow \N$ to be a random variable (determined by $Q$) as follows. If $j > 1$,
\begin{equation*}
  v_j(x) \coloneqq \abs{\set{S \in \ones^x \cap \daisy_j : S \subseteq Q \cup K_j}} \:,
\end{equation*}
and $v_1(x)$ is defined likewise, with the exception that, when more than $\capa$ sets intersect in a point outside $K_1$, they are discarded.
\begin{claim}
	\label{clm:v}
	Let $x\in \Sigma^n$, $j \in [q]$ and $\kappa$ an assignment to $K_j$.
	Then $v_j(x_\kappa)$ is equal (as a function of $Q$) to the maximum value of the counter $v$ computed by $N$ on input $x$ with kernel $K_j$ and the kernel assignment $\kappa$ to $K_j$.
\end{claim}
\begin{proof}
Fix $x \in \Sigma^n$.
Recall that when algorithm $N$ computes $v$ for a $j \in \{2,3,\dots,q\}$ and a kernel assignment $\kappa$ to $K_j$ in \ref{step:votecount}, it only increases $v$ if it encounters a tuple $(S,a_S,1,s)$ where $S \in \daisy_j$, $S \subset Q \cup K_j$ and $a_S$ assigns on $S$ the same values as $x_{\kappa}$ does.
	Thus, by the definition of $\ones_x$, the algorithm $N$ counts exactly all the tuples $(S,a_S,1,s)$ such that $S \in\ones^x$  and $S \subset Q \cup K_j$.
	These are precisely the sets that comprise the collection whose cardinality is $v_j(x_\kappa)$.
	Note that the same holds when $j=1$ due to the additional condition in  \ref{step:votecount} and the corresponding restriction in the definition of $v_1(x_\kappa)$.
\end{proof}

We now proceed to the main claims. The algorithm $N$ only counts votes for output 1, i.e. tuples with 1 as their third value, and hence it suffices to prove that: (1) \emph{when $f(x) = 1$ and the kernel assignment is  $\kappa = x_{|K_j}$} (the value of $x$ on the indices in $K_j$) \emph{for some daisy $\daisy_j$}, the number of votes is high enough to cross the threshold $\tau_j$; and that (2) \emph{when $f(x) = 0$, then every kernel assignment $\kappa$} is such that the number of votes is smaller than the threshold. These conditions are shown to hold with high probability in \cref{sec:correctness-non-robust,sec:correctness-robust}, respectively, and we show how the theorem follows from them in \cref{sec:conclude}.

\subsection{Correctness on non-robust inputs}
\label{sec:correctness-non-robust}

\begin{claim}
\label{clm:completeness}
Let $Q$ be the coordinates sampled by $N$ and fix $x \in \Sigma^n$ such that $f(x) = 1$. There exists $j \in [q]$ such that, with the kernel assignment $\kappa = x_{|K_j}$, the vote counting function satisfies $v_j(x_\kappa) = v_j(x) \geq \tau_j$ with probability at least $9/10$.
\end{claim}

\begin{proof}
  For ease of notation, let us fix $x$ as in the statement and denote $\ones \coloneqq \ones^x = \ones^{x_\kappa}$. When $j > 1$, define the subcollection of $\ones$ in $\daisy_j$ by $\ones_j \coloneqq \ones^x \cap \daisy_j$; when $j = 1$, define $\ones_1 \coloneqq (\ones^x \cap \daisy_j) \setminus \mathcal{C}$, where $\mathcal{C} \subseteq \ones^x \cap \daisy_1$ contains every  $S\in \ones^x \cap \daisy_1$ for which there exist at least $\capa-1$ other sets in $S' \in \ones^x \cap \daisy_1$ that have the same petal as $S$, i.e., such that $S\setminus K_1 = S'\setminus K_1$. We also take $n$ to be sufficiently large when necessary for an inequality to hold.

For the claim to hold we require the existence of $j\in [q]$ such that $\ones_j$ is a sufficiently large portion of $\ones$.
Since $\bigcup_{0 < j \leq q} \ones_j = \ones \setminus (\ones_0\cup \mathcal{C})$, in  order to achieve this goal, we only need to bound the sizes of both $\ones_0$ and $\mathcal{C}$. As a first step, we bound $\mu_x(\ones_0)$ and $\mu_x(\mathcal{C})$, which give us a lower bound on $\mu_x\left(\bigcup_{0 < j \leq q} \ones_j\right)$, which we then use in order to lower bound $\abs{\bigcup_{0 < j \leq q} \ones_j}$.

We start with $\mu_x(\ones_0)$. All the sets in $\daisy_0$ are subsets of $K_0$, and $\abs{K_0}  < \rho n/2$ by \cref{clm:kernels}.
This implies that the cardinality of $\cup \ones_0$ (the union of all sets in $\ones_0$) is strictly less than $\rho n$, and consequently, by the volume lemma (\cref{lem:volume}, which applies because $f(x) = 1$), we have $\mu_x(\ones_0) \leq 2 \sigma$.

We now proceed to bound $\mu_x(\mathcal{C})$. As $\mathcal{C} \subseteq \daisy_1$, every set in $\mathcal{C}$ has exactly one element in $[n] \setminus K_1$ and repeats at least $\capa$ times, the cardinality of $\cup \mathcal{C}$ is at most $\abs{K_1} + \frac{\abs{\mathcal{C}}}{\capa}$.
By \cref{clm:kernels}, $\abs{K_1} < \rho n/2$, and it follows that
\begin{flalign*}
&&\abs{K_1} + \frac{\abs{\mathcal{C}}}{\capa} &< \frac{\rho n}{2} + \frac{\abs{\ones}}{\capa} && \\
&&&\leq \frac{\rho n}{2} + \frac{\abs{\mu_x}}{\capa}  &&\left(\text{by \cref{equation:onesxMux}, } \ones =\ones_x \subseteq \supp(\mu_x)\right)\\
&&&\leq \frac{\rho n}{2} + \abs{\mu_x} \cdot \capparamInv  &&\left(\text{by \cref{equation:capparam}, } \capa^{-1} = \capparamInv\right)\\
&&&= \frac{\rho n}{2} + \suppsizefrac\cdot\capparamInv  &&\left(\text{by \cref{equation:mux}, } \abs{\mu_x} = \suppsizefrac\right)\\
&&& \leq \rho n \:.
\end{flalign*}
Consequently, by \cref{lem:volume}, $\mu_x(\mathcal{C}) \leq 2 \sigma$.

By the definition of error rate, $\mu_x(\ones) \geq 1-\sigma$. Since $\set{\ones_j : 0 \leq j \leq q}$ is a partition of $\ones$ (because $\set{\daisy_j}$ is a partition),
	\begin{equation*}
\mu_x\left(\bigcup_{0 < j \leq q} \ones_j\right) = \mu_x(\ones) - \mu_x(\ones_0) - \mu_x(\mathcal{C}) \geq 1 - 5\sigma \:.
	\end{equation*}
As $\mu_x$ is uniform, each element of the multi-collection $\ones$ has weight exactly $1/\abs{\mu_x}$. Therefore,
	\begin{equation*}
		\sum_{j = 1}^q\abs{\ones_j} = \abs{\mu_x} \cdot \mu_x(\cup_{j \in [q]} \ones_j) \geq \abs{\mu_x} (1 - 5 \sigma) \geq \abs{\mu_x}/2 \:,
	\end{equation*}
where the last inequality follows from the assumption that $5\sigma  \leq 1/2$ (which follows, e.g., from $q \geq 3$, which \cref{lem:prep-alg} ensures).
Let $j$ be such that
\begin{equation}\label{equation:onesj}
\abs{\ones_j} \geq \frac{\abs{\mu_x}}{2q} \:;
\end{equation}
by averaging, such a $j$ indeed exists.
Our goal now is to show that with probability at least $9/10$, there are at least $\tau_j$ sets $S\in \ones_j$ whose petal is in $Q$, i.e., such that $S\setminus K_j \subseteq Q$.

Instead of proving this directly on $\ones_j$, we do so on collections that form a partition of $\ones_j$ and have a useful structure.
The sets in $\ones_j$ are also in $\daisy_j$, so that $\ones_j$ is also a daisy with kernel $K_j$. By~\cref{clm:hbound},
for every set $S\in \ones_j$, there exist at most $2n^{\max\set{1,j-1}/q}-1$ sets $S' \in \ones_j \setminus \set{S}$ whose petals have a non-empty intersection with the petal of $S$, i.e, such that $(S \cap S') \setminus K_j \neq \varnothing$.
This enables us to apply \cref{lem:simple-daisy} to $\ones_j$, partitioning it into $\set{\sets_i : i \in [t]}$ simple daisies of equal sizes, where
\begin{equation}\label{equation:Ct}
t \leq 2n^{\max\set{1,j-1}/q} \:.
\end{equation}  Thus, for every $i\in [t]$,
\begin{equation}\label{equation:setsj}
\abs{\sets_i} = \frac{\abs{\ones_j}}{t} \:.
\end{equation}

Let $\ones_j'$ be the multi-collection of all sets $S\in \ones_j$ such that $S\setminus K_j \subseteq Q$.
In the same manner, for every $i\in [t]$, let $\sets_i'$ be the multi-collection of all sets $S\in \sets_i$ such that $S\setminus K_j \subseteq Q$.

By construction, the collections $\set{\sets_i'}$ are pairwise disjoint. Also, by the definition of $v_j$, we have  $v_j(x)= \abs{\ones_j'} = \sum_{i = 1}^{t}\abs{\sets_i'}$. Therefore, the event that $v_j(x) \leq \tau_j$ can only occur if there exists $i \in [t]$ such that $\abs{\sets_i'} \leq \tau_j/t$.

Consequently, we obtain
\begin{flalign*}
&&\Pr\left[v_j(x) \leq \tau_j\right] & \leq \Pr\left[\abs{\sets_i'} \leq \frac{\tau_j}{t} \text{ for some } i \in [t]\right]\\
&&&\leq \sum_{i = 1}^{t}\Pr\left[\abs{\sets_i'} \leq \frac{\tau_j}{t} \right] && \text{(union bound)} \\
&&&\leq t \Pr\left[\abs{\sets_1'} \leq \frac{\tau_j}{t} \right]. && \text{(all $\sets_i$ have equal size)}
\end{flalign*}

We show afterwards that the probability of the event $\abs{\sets_1'} \leq \frac{\tau_j}{t}$ is strictly less than $1/10t$, which by the inequality above implies the claim.

We later use the Chernoff bound with $\sets_1,$ and hence we start by bounding $\E[\abs{\sets_1'}]$ from below.
Recall that the petal of every set $S\in \sets_1 \subseteq \daisy_j$ has size $j$ (i.e., $\abs{S \setminus K_j} = j$), and therefore is in $\sets_1'$ with probability exactly $p^j$.
So
\begin{flalign}
&&\E[\abs{\sets_1'}] = \abs{\sets_1}p^j &= \frac{\abs{\ones_j}p^j}{t}&&\left(\text{by \cref{equation:setsj}, }\abs{\sets_1} = \abs{\ones_j}/t\right)\nonumber\\
\label{equation:ES1}&&&\geq \frac{\abs{\mu_x}p^j}{2tq}&&\left(\text{ by \cref{equation:onesj}, } \abs{\ones_j} \geq \frac{\abs{\mu_x}}{2q}\right)\\
&&&= \frac{\tau_j}{t}.&&\left(\text{by \cref{equation:taujvalue}, } \tau_j = \thresholdx{j}\right)\nonumber
\end{flalign}
Thus,
  \begin{equation*}
\Pr\left[\abs{\sets_1'} \leq \frac{1}{2}\E[\abs{\sets_1'}]\right]
\geq \Pr\left[\abs{\sets_1'} \leq \frac{\tau_j}{t} \right] \:.
\end{equation*}
Next we show that the probability of the event $\abs{\sets_1'} \leq \frac{1}{2}\E[\abs{\sets_1'}$ is at most $1/10t$, which concludes the proof.
Since  $\sets_1$ is a simple daisy, the petals of sets in $\sets_1$ are pairwise disjoint and hence the events $S \setminus K_j \subset Q$, for every $S\in \sets_1$, are all independent. This enables us to use the Chernoff bound to get that
\newcommand{\setoneEvent}{\exp\left(-\frac{\abs{\mu_x}p^j}{32tq}\right)}
\begin{flalign*}
&&&\Pr\left[\abs{\sets_1'} \leq \frac{1}{2}\E[\abs{\sets_1'}]\right] &&\\
&&&\qquad \leq \exp\left(-\frac{\E[\abs{\sets_1'}]}{8}\right) &&\left(\text{Chernoff bound}\right)\\
&&&\qquad\leq \exp\left(-\frac{\abs{\mu_x}p^j}{16tq}\right) &&\left(\text{by \cref{equation:ES1}}\right)\\
&&&\qquad\leq \exp\left(-\suppsizefrac \cdot \frac{p^j n}{16tq}\right)&&\left(\text{by \cref{equation:mux}, }\abs{\mu_x} = \suppsizefrac\right)\\
&&&\qquad\leq \exp\left(-\errorrateInv \cdot \frac{3\ln \abs{\Sigma}p^j n}{8tq}\right)&&\left(\text{by \cref{equation:sigmaVal}, } \sigma^{-1} = \errorrateInv\right)\\
&&&\qquad\leq \exp\left(-\frac{3\ln \abs{\Sigma}p^j}{4} \cdot n^{1-\frac{\max\set{1,j-1}}{q}}\right)&&\left(\text{by \cref{equation:Ct}, }t\leq 2n^{\frac{\max\set{1,j-1}}{q}}\right)\\
&&&\qquad\leq \frac{1}{10t}\exp\left(-\gamma^j\cdot\frac{3\ln\abs{\Sigma}}{4} \cdot n^{1 - \frac{\max\set{1, j-1}}{q} - \frac{j}{2q^2}} + \ln (20t)\right)&&\left(p=\sampprob\right)\\
&&&\qquad\leq \frac{1}{10t}\exp\left(-\gamma\cdot\frac{3\ln\abs{\Sigma}}{4} \cdot n^{ \frac{1}{q} - \frac{1}{2q}} + \ln (10t)\right)&&\left(1 \leq j \leq q\right)\\
&&&\qquad\leq \frac{1}{10t} \:,
\end{flalign*}
where the last inequality follows for $n$ sufficiently large because $\ln t\leq \frac{\max\set{1,j-1}}{q}\ln n + 1 = o(n^{1/2q})$ and $\gamma\cdot \ln\abs{\Sigma} = \Omega(1)$.
\end{proof}

Note that, although a success probability of $9/10$ suffices to ensure correctness of a single run of $N$, \cref{clm:completeness} yields a much stronger result: the failure probability is \emph{exponentially small}. This is because \cref{clm:completeness} does not enumerate over kernel assignments. Moreover, the analysis for the case $j = 1$ can be improved significantly (as will be necessary in \cref{clm:soundness}), but this does not yield in an overall improvement in our results.

\subsection{Correctness on robust inputs}
\label{sec:correctness-robust}
In the following claim we note that $\abs{K_1}/n$-robustness suffices for the analysis, since it ensures all kernel assignments $\kappa$ lead $x_\kappa$ to also output $f(x) = 0$.

\begin{claim}
    \label{clm:soundness}
    Suppose the input $x \in \Sigma^n$ is $\abs{K_1}/n$-robust for $M'$ and $f(x) = 0$. Then, for every $j \in [q]$ and every assignment $\kappa$ to the kernel $K_j$, the vote count satisfies $v_j(x_\kappa) < \tau_j$ with probability at least $1 - \abs{\Sigma}^{\abs{K_j}}/(10q)$.
\end{claim}

\begin{proof}
	For ease of notation, fix $j\in [q]$, an assignment $\kappa$ to $K_j$ and $x$ as in the statement, and let $\ones \coloneqq  \ones^{x_\kappa}$.
	If $j > 1$, define the subcollection of $\ones$ in $\daisy_j$ by $\ones_j \coloneqq \ones \cap \daisy_j$; if $j = 1$, define $\ones_1 \coloneqq (\ones \cap \daisy_1) \setminus \mathcal{C}$, where $\mathcal{C} \subseteq \ones \cap \daisy_1$ contains every $S\in \ones \cap \daisy_1$ for which there exist at least $\capa-1$ other sets $S' \in \ones \cap \daisy_1$ that have the same petal as $S$, i.e., such that $S\setminus K_1 = S'\setminus K_1$.
  We also take $n$ to be sufficiently large when necessary for an inequality to hold.

  Note that $x_\kappa$ may not be in the domain of $f$, but the robustness of $x$ allows us to bound the size of $\ones = \ones^{x_\kappa}$ regardless. Moreover, since $f(x) =0$, we know that $\mu_x(\ones) \leq \sigma$.
  As $\mu_x$ is uniform, each element of the multi-collection has $\ones$ weight exactly $1/\abs{\mu_x}$. Therefore, for every $i\in [q]$,
  \begin{equation}\label{equation:Sonesj}
   \abs{\ones_i} \leq \sigma\abs{\mu_x} \:.
  \end{equation}

 	Our goal now is, for every $j\in [q]$, to upper bound the probability that there are at least $\tau_j$ sets $S\in \ones_j$ whose petal is in $Q$, i.e., such that $S\setminus K_j \subseteq Q$.

   	For every $j\in [q]$, let $\beta_j$ be such that for every set $S\in \ones_j$ there exist at most $\beta_j - 1$ other distinct sets $S'\in \ones_j$ whose petal intersects the petal of $S$, i.e., $(S \setminus K_1) \cap (S' \setminus K_1) \neq \varnothing$.

   	For the time being let us fix $j\in [q]$.
   	By applying \cref{lem:simple-daisy}, we partition $\ones_j$ into $\set{\sets_i : i \in [\beta_j]}$, such that for every $i\in [q]$,
   	\begin{equation}\label{equation:SSetsj}
   	\abs{\sets_i} = \frac{\abs{\ones_i}}{\beta_j} \leq \frac{\sigma\abs{\mu_x}}{\beta_j} \:,
   	\end{equation}
   	where the inequality follows from \cref{equation:Sonesj}, and each $\sets_i$ is a simple daisy of size $\abs{\ones_1}/\beta_j$.

   	Let $\ones_j'$ be the multi-collection of all sets $S\in \ones_j$ such that $S\setminus K_j \subseteq Q$.
 	In the same manner, for every $i\in [\beta_j]$, let $\sets_i'$ be the multi-collection of all sets $S\in \sets_i$ such that $S\setminus K_j \subseteq Q$. By the definition of $v_j$ and the fact that $\set{\sets_i}$ is a partition
 	\begin{equation*}
 	  v_j(x_\kappa) = \abs{\ones_j'} = \sum_{i = 1}^{\beta_j} \abs{\sets_i'}.
 	\end{equation*}
 	Since the event $v_1(x_\kappa) \geq \tau_j$ can only occur if $\abs{\sets_i'} \geq \frac{\tau_j}{\beta_j}$ for some $i\in[\beta_j]$, we obtain
 	\begin{flalign*}
 	&&\Pr\left[v_j(x) \geq \tau_j\right] & \leq \Pr\left[\abs{\sets_i'} \geq \frac{\tau_j}{\beta_j} \text{ for some } i \in [\beta_j]\right]\\
 	&&&\leq \sum_{i = 1}^{\beta_j}\Pr\left[\abs{\sets_i'} \geq \frac{\tau_j}{\beta_j} \right] && \text{(union bound)} \\
 	&&&\leq \beta_j\cdot\Pr\left[\abs{\sets_1'} \geq \frac{\tau_j}{\beta_j} \right] \:. && \text{(all $\sets_i$ have equal size)}
 	\end{flalign*}

   	Now our goal is to show that the event that $\abs{\sets_1'} \geq \frac{\tau_j}{\beta_j}$ happens with probability at most $\frac{\abs{\Sigma}^{-\abs{K_1}}}{10q\beta_j}$. Note that this is sufficient for proving the claim because plugging this into the previous equation gives $\Pr\left[v_j(x) \geq \tau_j\right] \leq \abs{\Sigma}^{-\abs{K_j}}/(10q)$.

   	Since the sets in $\sets_1$ are pairwise disjoint, we can and do use the Chernoff bound. In order to do so we first bound the value of $\E[\abs{\sets_1'}]$ from above.
   	Recall that the petal of every set $S\in \sets_j \subseteq \daisy_j$ has size $j$ (i.e., $\abs{S \setminus K_j} = j$), and therefore $S$ is in $\sets_j'$ with probability exactly $p^j$. So,
   	\begin{flalign*}
   	&&\E[\abs{\sets_1'}] & =  \abs{\sets_1}\cdot p^j&& \\
   	&&&\leq \frac{\sigma\cdot\abs{\mu_x}\cdot p^j}{\beta_j} &&\left(\text{by \cref{equation:SSetsj}}\right)\\
   	&&&= \frac{\tau_j}{2\beta_j} \:. &&\left( \text{by \cref{equation:sigmaVal} and \cref{equation:taujvalue}, }\sigma = \errorratefrac \text{ and } \tau_j = \thresholdx{j}\right)
   	\end{flalign*}

   	We now use the Chernoff bound, stopping at a partial result and providing separate analyses for the cases $j = 1$ and $j>1$.

   	\begin{flalign*}
   	&&\Pr\left[\abs{\sets_1'} \geq \frac{\tau_j}{\beta_j} \right] & = \Pr\left[\abs{\sets_1'} \geq \frac{\tau_j}{\beta_j\E[\abs{\sets_1'}]}\E[\abs{\sets_1'}] \right]&&\\
   	&&& \leq \exp\left(-\left(\frac{\tau_j}{\beta_j\E[\abs{\sets_1'}]}-1\right)^2\cdot\frac{\E[\abs{\sets_1'}]}{3}\right)  &&\left(\text{Chernoff bound}\right) \\
   	&&& \leq \exp\left(-\frac{\tau_j}{3\beta_j}\right)&&\left(\text{explained aferwards}\right)\\
   	&&& = \exp\left(-\frac{\abs{\mu_x}\cdot p^j}{6q\beta_j}\right)  &&\left(\text{by \cref{equation:taujvalue}, }\tau_j = \thresholdx{j} \right)\\
   	&&& = \exp\left(-\frac{\ln\abs{\Sigma} n p^j}{q\beta_j\cdot\sigma}\right) \:,  &&\left(\text{by \cref{equation:mux}, }\abs{\mu_x} = \suppsizefrac\right)
   	\end{flalign*}
   	where the second inequality follows from $\left(\frac{\tau_j}{\beta_j\E[\abs{\sets_1'}]}-1\right)^2\cdot\frac{\E[\abs{\sets_1'}]}{3}$ being minimal when $\E[\abs{\sets_1'}]$ is at its upper bound of $\frac{\tau_j}{2\beta_j}$.
   	We next proceed to the first of the two cases.

	Take $j = 1$.
	In this case, by the construction of the daisy partition (\cref{thm:daisy-partition}), every set $S\in \ones_1$ has a petal $S\setminus K_1$ of cardinality exactly $1$.
	By the definition of $\ones_1$, each set $S \in \ones_1$ has at most $\capa - 1$ other sets $S' \in \ones_1$ whose petal intersects the petal of $S$, i.e., $(S \setminus K_1) \cap (S' \setminus K_1) \neq \varnothing$ (and thus $S \setminus K_1 = S' \setminus K_1$, since both petals have size 1).
	Therefore, at most $\beta_1 - 1 = \alpha - 1$ distinct sets of $\ones_1$ intersect each $S \in \ones_1$, which follows from \cref{equation:capparam}.
	Now,
\begin{flalign*}
   	&&&\exp\left(-\frac{\ln\abs{\Sigma} n p}{q\capa\cdot\sigma}\right) &&\\
   	&&&\qquad = \exp\left(-\frac{n\cdot p\cdot \rho}{12q}\right)  && \left(\mbox{by \cref{equation:capparam}, }\alpha = \capparamfrac\right)\\
   	&&&\qquad = \exp\left(-\gamma\cdot\frac{\rho}{12q} \cdot n^{1-1/2q^2}\right)&&\left(p=\sampprob\right)\\
   	&&&\qquad = \frac{1}{10q}\exp\left(-\gamma\cdot n^{1-1/q}\cdot\frac{\rho\cdot n^{\frac{1}{q}-\frac{1}{2q^2}}}{12q}  +\ln(10q)\right)&&\\
   	&&&\qquad \leq \frac{1}{10q}\exp\left(-\ln \abs{\Sigma} \cdot \kernelZeroSizeUB\right)&&\left(\text{large enough } n\right)\\
   	&&&\qquad \leq \frac{\abs{\Sigma}^{-\abs{K_1}}}{10q} \:,&&
\end{flalign*}
where the last inequality follows because $\abs{K_1} \leq \kernelZeroSizeUB$ by \cref{clm:kernels} (and $\ln \abs{\Sigma} \leq \log n$).

	Now, take $j > 1$.
	By \cref{clm:hbound}, $\beta_j = 2h({j-1}) = 2n^{(j-1)/q}$, which implies the first equality in the following.
\begin{flalign*}
&&\exp\left(-\frac{\ln\abs{\Sigma} \cdot n p^j}{q\beta_j\cdot\sigma}\right) & = \exp\left(-\frac{\ln\abs{\Sigma} \cdot n p^j}{2q\cdot\sigma}\cdot n^{-\frac{j-1}{q}}\right)  && \\
&&&=\exp\left(-2 \ln\abs{\Sigma} \cdot p^j\cdot n^{1-\frac{j-1}{q}}\right)&&\left(\text{by \cref{equation:sigmaVal}, } \sigma = \errorratefrac\right)\\
&&&= \exp\left(-2 \ln\abs{\Sigma}\cdot\gamma^j \cdot n^{1-\frac{j-1}{q}-\frac{j}{2q^2}}\right)&&\left(p=\sampprob\right)\\
&&&= \exp\left(-2 \ln\abs{\Sigma}\cdot\gamma^j \cdot n^{1-\frac{j}{q}+\frac{2q - j}{2q^2}}\right)&&\\
&&&\leq \frac{1}{10q}\exp\left(-\ln\abs{\Sigma}\cdot\gamma \cdot n^{1-\frac{j}{q}} \cdot 2 n^{\frac{1}{2q}}+\ln (10q) \right)&&\left(1 < j \leq q \right)\\
&&&\leq \frac{1}{10q}\exp\left(-\ln\abs{\Sigma}\cdot\kernelNonzeroSizeUB{j}\right)&&\left(\text{large enough } n\right)\\
&&&\leq \frac{\abs{\Sigma}^{-\abs{K_j}}}{10q} \:,
\end{flalign*}
where the last inequality follows because $\abs{K_j} \leq \kernelNonzeroSizeUB{j}$ by \cref{clm:kernels}.
\end{proof}

\subsection{Concluding the proof}
\label{sec:conclude}
We conclude the proof \cref{thm:main} by applying the two previous claims. Recall that we transformed a $\rho$-robust local algorithm $M$ for a function $f$, with query complexity $\ell$, into a $\rho$-robust local algorithm $M'$ with query complexity $q = O(\query)$ and suitable error rate. Then we transformed $M'$ into a sample-based algorithm $N$ with sample complexity $n^{1-1/O(q^2)} = n^{1-1/O(\querysq)}$, an upper bound guaranteed by the sampling step (\ref{step:sample}) in the construction of $N$. It remains to show correctness of the algorithm on every input in the domain of $f$.

We first consider errors that may arise in the sampling step. By the Chernoff bound, it chooses more than $2pn = 2n^{1-j/(2q^2)}$ points to query and thus outputs arbitrarily with probability at most $1/10$. Otherwise, it proceeds to the next steps.

In the next part of the proof we analyse $v_j(x)$ instead of analyzing $v$ (of \ref{step:votecount}) in algorithm $N$; this is sufficient, since by \cref{clm:v}, they are distributed identically over $Q$.

Suppose the input $x \in \Sigma^n$ is such that $f(x) = 0$.
Since $x$ is $\rho$-robust, it is in particular $\abs{K_1}/n$-robust (because $\abs{K_1} = o(n)$). Then \cref{clm:soundness} ensures that, for every $j \in [q]$ and kernel assignment $\kappa$ to $K_j$, the vote counter satisfies $v_j(x_\kappa)\geq \tau_j$ with probability at most $\abs{\Sigma}^{-\abs{K_j}}/(10q)$. A union bound over all $j \in [q]$ and $\abs{\Sigma}^{\abs{K_j}}$ assignments to the kernel $K_j$ ensures the probability this happens, causing $N$ to output $1$ in the threshold check step (\ref{step:threshold}), is at most $1/10$; otherwise, $N$ will enumerate over every assignment and then (correctly) output $0$ in \ref{step:outzero}.

Now suppose $x \in \Sigma^n$ is such that $f(x) = 1$. Then \cref{clm:completeness} ensures that, for some $j \in [q]$, the kernel assignment $\kappa = x_{|K_j}$ will make the vote count satisfy $v_j(x) \geq \tau_j$ with probability at least $9/10$, in which case $N$ (correctly) outputs 1 in the threshold check step (\ref{step:threshold}).

Therefore, $N$ proceeds beyond the sampling step with probability $9/10$ and outputs correctly (due to \cref{clm:soundness} and \cref{clm:completeness}) with probability at least $9/10 - 1/10 \geq 2/3$. This concludes the proof of \cref{thm:main}.

\begin{remark}
  \label{rem:tight-error-prob}
  Notice that the claims actually prove a stronger statement: the failure probability is not merely $1/3$, but \emph{exponentially small}. For each $j \in [q]$, the error probability is
  \begin{equation*}
  \exp\left(-\Omega\left(n^{1-\frac{j}{q}+\frac{2q-j}{2q^2}}\right)\right),
  \end{equation*}
  but it must withstand a union bound over $\exp\left(O\left(n^{1- j/q}\right)\right)$ events (corresponding to the assignments to the kernel $K_j$). The smallest slackness is in the case $j = q$, where the success probability is still $\exp\left(-\Omega\left(n^{1/2q}\right)\right)$; this implies that correctness holds for $\exp\left(c \cdot n^{1/2q}\right)$ many executions, if the constant $c$ is sufficiently small. Therefore, \emph{the same samples can be reused for exponentially many runs of possibly different algorithms.}
\end{remark}

\section{Applications}
\label{sec:applications}

In this section, we derive applications from \cref{thm:main} which range over three fields of study: property testing, coding theory, and probabilistic proof systems. We first give a brief overview of the applications in the following paragraphs, then proceed to the proofs in \cref{sec:test,sec:rldc,sec:map}.

\paragraph{Query-to-sample tradeoffs for adaptive testers\ifsicomp\else.\fi} The application to property testing is an immediate corollary of \cref{thm:main}: since an $\eps/2$-tester is a $(\eps/2,0)$-robust algorithm for the problem of testing with proximity parameter $\eps/2$, \cref{cor:test} shows that any $\eps/2$-tester making $q$ adaptive queries can be transformed into a sample-based $\eps$-tester with sample complexity $n^{1-1/O(q^2\log^2 q)}$. In addition, we also show an application to multi-testing (\cref{cor:multi-test}).

\paragraph{Relaxed LDC lower bound\ifsicomp\else.\fi} By a straightforward extension of our definition of robust local algorithms to allow for outputting a special failure symbol $\bot$, our framework captures \emph{relaxed} LDCs (see \cref{sec:rldc}). We remark that, although standard LDCs have \emph{two-sided} robustness, the treatment of relaxed LDCs is analogous to one-sided robust algorithms.

By applying \cref{thm:main} to a relaxed local decoder \emph{once for each bit to be decoded}, we obtain a \emph{global} decoder that decodes uncorrupted codewords with $n^{1-1/O(q^2 \log^2 q)}$ queries; by a simple information-theoretic argument, we obtain a rate lower bound of $n = k^{1 + 1/O(q^2 \log^2 q)}$ for relaxed LDCs (see \cref{cor:rldcsample,cor:rldcbound}).

\paragraph{Tightness of the separation between MAPs and testers\ifsicomp\else.\fi} \cref{thm:main} applies to the setting of \emph{Merlin-Arthur proofs of proximity} (MAPs) via a description of MAPs as coverings by partial testers (\cref{clm:mapequiv}). In \cref{sec:map}, we show that the existence of an adaptive MAP for a property $\Pi$ with query complexity $q$ and proof length $m$ implies the existence of a sample-based tester for $\Pi$ with sample complexity $m \cdot n^{1-1/O(q^2 \log^2q)}$ (\cref{thm:mappt}).

This implies that there exists no property admitting a MAP with query complexity $q=O(1)$ and logarithmic proof length (in fact, much longer proof length) that requires at least $n^{1-1/\omega(q^2 \log^2 q)}$ queries for testers, showing the (near) tightness of the separation from \cite{GR18}.

\paragraph{Optimality of \cref{thm:main}\ifsicomp\else.\fi} We conclude \cref{sec:map} with a direct corollary of the tightness of the aforementioned separation between MAPs and testers of \cite{GR18}, we obtain that the general transformation in \cref{thm:main} is optimal, up to a quadratic gap in the dependency on the sample complexity. This follows simply because a transformation with smaller sample complexity could have been used to improve \cref{thm:mappt}, yielding a tester with query complexity that contradicts the lower bound (see \cref{thm:optimality}).

\subsection{Query-to-sample tradeoffs for adaptive testers}
\label{sec:test}

Recall that a property tester $T$ for property $\Pi \subseteq \Sigma^n$ is an algorithm that receives explicit access to a proximity parameter $\eps > 0$, query access to $x \in \Sigma^n$ and \emph{approximately decides} membership in $\Pi$: it accepts if $x \in \Pi$ and rejects if $x$ is $\eps$-far from $\Pi$, with high probability.

By \cref{clm:testrla}, an $\eps$-tester with $\eps \in (0,1)$ is an $\eps$-robust local algorithm that computes the function $f\colon \Pi \cup \overline{B_{2\eps}(\Pi)} \rightarrow \bitset$ defined as follows.
\begin{equation*}
	f(x) = \left\{\begin{array}{ll}1, & \text{if } x \in \Pi\\0, & \text{if } x \in \overline{B_{2\eps}(\Pi)}.\end{array}\right.
\end{equation*}

Note, moreover, that a local algorithm that solves $f$ is by definition a $2\eps$-tester, accepting elements of $\Pi$ and rejecting points that are $2\eps$-far from it with high probability. A direct application of \cref{thm:main} thus yields the following corollary, which improves upon the main result of \cite{FLV15}, by extending it to the two-sided adaptive setting.

\begin{corollary}
  \label{cor:test}
  For every fixed $\eps > 0, q \in \N$, any $\eps$-testable property of strings in $\Sigma^n$ with $q$ queries admits a sample-based $2\eps$-tester with sample complexity $n^{1-1/O(q^2 \log^2 q)}$.
\end{corollary}

This also immediately extends an application to multitesters in \cite{FLV15}. By standard error reduction, for any $k \in \N$, an increase of the sample complexity by a factor of $O(\log k)$ ensures each member of a collection of $k$ sample-based testers errs with probability $1/(3k)$. A union bound allows us to \emph{reuse the same samples for all testers}, so that all will output correctly with probability $2/3$. Taking $k = \exp\left(n^{1/\omega(q^2 \log^2 q)}\right)$, the sample complexity becomes $n^{1-1/O(q^2 \log^2 q)} \cdot n^{1/\omega(q^2 \log^2 q)} = o(n)$, which yields the following corollary.

\begin{corollary}
 \label{cor:multi-test}
  If a property $\Pi \subseteq \Sigma^n$ is the union of $k = \exp\left(n^{1/\omega(q^2 \log^2q)}\right)$ properties $\Pi_1, \ldots, \Pi_k$, each $\eps$-testable with $q$ queries, then $\Pi$ is $2\eps$-testable via a sample-based tester with sublinear sample complexity.
\end{corollary}

A tester for the union simply runs all (sub-)testers, accepting if and only if at least one of them accepts. A proof for a generalisation of this corollary, which holds for \emph{partial testers}, is given in the \cref{sec:map}.

\subsection{Stronger relaxed LDC lower bounds}
\label{sec:rldc}

Relaxed LDCs are codes that relax the notion of LDCs by allowing the local decoder to abort on a small fraction of the indices, yet crucially still avoid errors. This seemingly modest relaxation turns out to allow for dramatically better parameters (an exponential improvement on the rate of the best known $O(1)$-query LDCs). However, since these algorithms are much stronger, obtaining lower bounds on relaxed LDCs is significantly harder than on standard LDCs. Indeed, the first lower bound on relaxed LDCs \cite{GL21} was only shown more than a decade after the notion was introduced; this bound shows that to obtain query complexity $q$, a relaxed LDC $C\colon\bitset^k \to \bitset^n$ must have blocklength
\begin{equation*}
    n \geq k^{1 + \frac{1}{O\left(2^{2q} \cdot \log^2 q\right)}} \:,
\end{equation*}

In this section, we use \cref{thm:main} to obtain an improved lower bound with an exponentially better dependency on the query complexity. We begin by recalling the definition of relaxed LDCs.

\begin{definition}[\cref{def:rldc}, restated]
  A code $C\colon\bitset^k \to \bitset^n$ whose distance is $\delta_C$ is a $q$-local relaxed LDC with success rate $\rho$, decoding radius $\delta \in (0,\delta_C/2)$ and error rate $\sigma \in (0, 1/3]$ if there exists a randomised algorithm $D$, known as a \emph{relaxed decoder}, that, on input $i\in [k]$, makes at most $q$ queries to an oracle $w$ and satisfies the following conditions.

  \begin{enumerate}
  \item \textsf{Completeness}: For any $i\in [k]$ and $w = C(x)$, where $x\in\bitset^k$,
  \begin{equation*}
    \Pr[D^{w}(i) = x_i] \ge 1 - \sigma\;.
  \end{equation*}

  \item \textsf{Relaxed Decoding}: For any $i\in [k]$ and $w \in\bitset^{n}$ that is $\delta$-close to a (unique) codeword $C(x)$,
    \begin{equation*}
      \Pr[D^{w}(i) \in \{x_i,\bot\}] \ge 1 - \sigma\;.
    \end{equation*}

    \item \textsf{Success Rate}: There exists a constant $\rho>0$ such that, for any $w\in\bitset^n$ that is $\delta$-close to a codeword $C(x)$, there exists a set $I_w\subseteq [k]$ of size at least $\rho k$ such that for every $i\in I_w$,
    \begin{equation*}
        \Pr[D^w(i)=x_i]\geq 2/3 \;.
    \end{equation*}
  \end{enumerate}
\end{definition}

\begin{remark}
  The first two conditions imply the latter, as shown by \cite{BGHSV06}. Therefore, it is not necessary to show the success rate condition when verifying that an algorithm $D$ is a relaxed local decoder.
\end{remark}

Note that, whenever $D^w$ outputs $\bot$, it detected that the input is \emph{not valid}, since it is inconsistent with \emph{any} codeword $C(x)$. We slightly generalise local algorithms (\cref{def:localalg}) to capture this behaviour, by allowing them to output $\bot$ as well as the correct function evaluation $f(z, x)$ (except for a prescribed set of valid inputs). Formally,

\begin{definition}[Relaxed local algorithm]
  Let $\Sigma$ be a finite alphabet, $Z$ a finite set and $\set{\mathcal{P}_z : z \in Z}$ a family of sets $\mathcal{P}_z \subseteq \Sigma^n$ indexed by $Z$. With $\mathcal{P} \coloneqq \set{(z, x) : z \in Z, x \in \mathcal{P}_z}$, let $f\colon \mathcal{P} \to \bitset$ be a partial function.
  
  A \emph{relaxed $q$-local algorithm $M$ for computing $f$} with \emph{valid input set} $V \subseteq \Sigma^n$ and \emph{error rate} $\sigma$ receives explicit access to $z \in Z$, query access to $x \in \mathcal{P}_z$, makes at most $q$ queries to $x$ and satisfies
  \begin{equation*}
		\Pr\big[M^x(z) \in \set{f(z,x), \bot}\big] \geq 1 - \sigma.
  \end{equation*}
  Moreover, if $x \in V$, then $M$ satisfies
  \begin{equation*}
    \Pr[M^x(z) = f(z,x)] \geq 1 - \sigma.
  \end{equation*}
\end{definition}

We shall also need to generalise the notion of robustness (\cref{def:robust}) accordingly.
\begin{definition}[Robustness]
  Let $\rho > 0$. A local algorithm $M$ for computing $f \colon \mathcal{P} \to \bitset$ is \emph{$\rho$-robust} at the point $(z, x) \in \mathcal{P}$ if $\Pr[M^w(z) \in \set{f(z,x), \bot}] \geq 1 -\sigma$ for all $w \in B_\rho(x)$. We say that $M$ is \emph{$(\rho_0, \rho_1)$-robust} if, for all $z \in Z$ and $b \in \bitset$, $M$ is $\rho_b$-robust at every $x$ such that $f(z,x) = b$.
\end{definition}
We remark that robustness for algorithms that allow aborting allows the correct value to change to $\bot$ (but, crucially, not to the wrong value) even if only one bit is changed. This makes the argument more involved than an argument for LDCs, and indeed, our theorem for \emph{relaxed} LDCs relies on the full machinery of \cref{thm:main}.

Note that an algorithm that ignores its input and always outputs $\bot$ fits both definitions above, but has no valid inputs and clearly does not display any interesting behaviour. We also remark that the set of valid inputs captures completeness (but \emph{not} the success rate) in the case of relaxed LDCs.

With these extensions, a relaxed local decoder $D$ with decoding radius $\delta$ fits the definition of a (relaxed) local algorithm that receives $i \in [k]$ as explicit input, where the code $C$ comprises the valid inputs and every $x \in C$ is $\delta$-robust for $D$.

While a relaxed local algorithm is very similar in flavour to a standard local algorithm, it may not be entirely clear whether a transformation analogous to \cref{thm:main} holds in this case as well. We next show that one indeed does: with small modifications to the algorithm constructed in \cref{sec:construction}, we leverage the same analysis of \cref{sec:analysis} to prove the following variant of \cref{thm:main}.

\begin{theorem}
\label{thm:rldcmain}
  Suppose there exists a $(\rho_0, \rho_1)$-robust relaxed local algorithm $M$ for computing the function $f\colon \mathcal{P} \rightarrow \bitset$ (where $Z \times \Sigma^n$) with query complexity $\ell = O(1)$ and $\rho_0, \rho_1 = \Omega(1)$. Let $V \subseteq \Sigma^n$ be the valid inputs of $M$.
  Then, there exists a \emph{sample-based} relaxed local algorithm $N$ for $f$ with sample complexity $n^{1-1/O(\ell^2 \log^2 \ell)}$ with the same set $V$ of valid inputs.
\end{theorem}

\begin{proof}
  Throughout the proof, we assume the explicit input to be fixed and omit it from the notation.
  First, note that error reduction (\cref{clm:err-red}) and randomness reduction (\cref{clm:rand-red}) apply in the relaxed setting: the analysis is identical on valid inputs, and holds likewise for the remainder of the domain of $f$ (with correctness of $M$ relaxed to be $M^x \in \set{f(x), \bot}$). Thus \cref{lem:prep-alg} enables the transformation of $M$ into another robust algorithm $M'$ with small error rate that uniformly samples a decision tree from a multi-collection of small size.

  Recall that the construction of the sample-based algorithm in \cref{sec:construction} uses a collection of triplets obtained from the behaviour of $M'$ when it outputs 1. A corresponding collection can be obtained for the case where $M'$ outputs 0. Denote by $\tuples_b$ the collection that corresponds to output $b \in \bitset$, and let $N_b$ be the sample-based algorithm that
  \begin{itemize}[noitemsep]
    \item uses the triplets $\tuples_b$ to construct its daisy partition in the preprocessing step;
    \item outputs $b$ if the counter crosses the threshold in \ref{step:threshold}; and
    \item outputs $\bot$ in \ref{step:outzero} if the threshold is never reached;
  \end{itemize}
  but is otherwise the same as the construction of \cref{sec:construction}.

  The analysis of \cref{sec:analysis} applies to $N_b^x$: if $x \in V$, the analysis of \cref{clm:completeness} is identical; while if $x$ is robust and $f(x) = \neg b$, \cref{clm:soundness} requires a lower bound on the probability that $M'$ outputs $b$ when its input is $x$ (and enables an application of the volume lemma), which holds by the definition of error rate of a relaxed local algorithm. Therefore, the probability of each the following events is bounded by $1/10$:
  \begin{enumerate}[label=(\roman*),noitemsep]
  \item $N_b^x$ outputs arbitrarily in the sampling step;
  \item $N_b^x$ outputs $\bot$ when $f(x) = b$; and
  \item $N_b^x$ outputs $b$ when $x$ is robust and $f(x) = \neg b$.
  \end{enumerate}

  Finally, the relaxed sample-based algorithm $N$ simply executes the sampling step of $N_0$ then the enumeration steps of $N_0$ and $N_1$ on these samples, outputs $b$ if exactly one of $N_b$ outputs $b$, and outputs $\bot$ otherwise. Then, $N^x = f(x)$ if $x \in V$ and $N^x = \bot$ if $x \notin V$, with probability $7/10 \geq 2/3$.
\end{proof}

By casting a relaxed decoder as a robust relaxed local algorithm and applying \cref{thm:rldcmain}, we obtain the following corollary.

\begin{corollary}
  \label{cor:rldcsample}
  Any binary code $C\colon \bitset^k \rightarrow \bitset^n$ that admits a relaxed local decoder $D$ with decoding radius $\delta$ and query complexity $q = O(1)$ also admits a sample-based relaxed local decoder $D'$ with decoding radius $\delta/2$ and sample complexity $n^{1-1/O(q^2 \log^2 q)}$.
\end{corollary}

We are now ready to state the following corollary, which improves on the previous best rate lower bound for relaxed LDCs \cite{GL21} by an application of the theorem above to the setting of relaxed local decoding. This follows from the construction of a global decoder (which is able to decode the entire message) that is only guaranteed to succeed with high probability when its input is \emph{a perfectly valid codeword}.

\begin{corollary}
  \label{cor:rldcbound}
  Any code $C\colon \bitset^k \rightarrow \bitset^n$ that is relaxed locally decodable with $q = O(1)$ queries satisfies
  \begin{equation*}
		n = k^{1 + \frac{1}{O(q^2 \log^2 q)}}.
	\end{equation*}
\end{corollary}
\begin{proof}
  Let $D'$ be the sample-based relaxed LDC with sample complexity $q'$ obtained by \cref{cor:rldcsample} from a relaxed LDC with query complexity $q$ for the code $C$. Reduce the error rate of $D'$ to $1/3k$ by repeating the algorithm $O(\log k)$ times and taking the majority output, thus increasing the sample complexity to $O(q' \cdot \log k) = n^{1-1/t}$ with $t = O(q^2 \log^2 q)$.

  Now, consider the \emph{global decoder} $G$ defined as follows: on input $w$, execute the sampling stage once and the enumeration stages of $D^w(1), \ldots, D^w(k)$ on the same samples. A union bound ensures that, with probability at least $2/3$, the outputs satisfy $D^w(i) = x_i$ for all $i$ if $w = C(x)$.

  The global decoder $G$ obtains $k$ bits of information from $n^{1-1/t}$ bits with probability above $1/2$. Information-theoretically, we must have
  \begin{equation*}
		k \leq \frac{n^{1-1/t}}{2} = \frac{n^{\frac{t-1}{t}}}{2},
	\end{equation*}
  so that $n \geq (2k)^{1 + \frac{1}{t-1}} \geq 2 k^{1 + \frac{1}{t-1}}$. Since $t = O(q^2 \log^2 q)$, it follows that $n = k^{1 + 1/O(q^2 \log^2 q)}$.
\end{proof}

\subsection{A maximal separation between testers and proofs of proximity}
\label{sec:map}

Recall that a Merlin-Arthur proof of proximity (MAP, for short) for property $\Pi$ is a local algorithm that receives explicit access to a proximity parameter $\eps > 0$ and a purported proof string $\pi$, as well as query access to a string $x \in \Sigma^n$. It uses the information encoded in $\pi$ to decide which coordinates of $x$ to query, accepting if $x \in \Pi$ and $\pi$ is a valid proof for $x$, and rejecting if $x$ is $\eps$-far from $\Pi$. In particular, a MAP with proof length $0$ is simply a tester. The complexity of a MAP is defined as the sum of its proof length and query complexity. For simplicity, we consider the proximity parameter $\eps$ to be a fixed constant in the following discussion.

A \emph{partial tester} $T$ is a relaxation of the standard definition of a tester, that accepts inputs inside a property $\Pi_1$ and rejects inputs that are far from a \emph{larger} property $\Pi_2$ that contains $\Pi_1$ (standard testing is the case where $\Pi_2 = \Pi_1$). We first formalise an observation made in \cite{FGL14}, which shows an equivalence between MAPs and coverings by partial testers.
\begin{claim}
	\label{clm:mapequiv}
	A MAP $T$ for property $\Pi \subseteq \Sigma^n$ with proof complexity $m$, error rate $\sigma$ and query complexity $q = q(n, \eps)$ is equivalent to a collection of \emph{partial testers} $\set{T_i : i \in [\abs{\Sigma}^m]}$. Each $T_i(\eps)$ accepts inputs in the property $\Pi_i$ and rejects inputs that are $\eps$-far from $\Pi$, with the same query complexity $q$ and error rate $\sigma$ as $T$. The properties $\Pi_i$ satisfy $\Pi_i \subseteq B_\eps(\Pi)$ and $\Pi \subseteq \cup_i \Pi_i$.
\end{claim}
\begin{proof}
	Consider a MAP $T$ with parameters as in the statement, and define $T_i(\eps) \coloneqq T(\eps, i)$ for each purported proof $i \in [\abs{\Sigma}^m]$. Clearly the query complexity and error rate of $T_i$ match those of $T$, and these testers reject points that are $\eps$-far from $\Pi$. The property $\Pi_i$ is, by definition, the set of inputs that $T_i$ accepts (with probability at least $1 - \sigma$), which is contained in $B_\eps(\Pi)$ (since elements of $\overline{B_\eps(\Pi)}$ are rejected), and may possibly be empty. But since the definition of a MAP ensures that, for each $x \in \Pi$, the tester $T_i^x$ accepts for some proof $i$ (with probability $1 - \sigma$), we have $\Pi \subseteq \cup_i \Pi_i$.

	Consider, now, a collection of testers $\set{T_i : i \in [\abs{\Sigma}^m]}$ as in the statement, and define a MAP $T$ that simply selects the tester indexed by the received proof string; i.e., $T(\eps, i) \coloneqq T_i(\eps)$. Then, with probability at least $1 - \sigma$, the MAP $T$ rejects inputs that are $\eps$-far from $\Pi$ and accepts $x \in \Pi$ when its proof string is $i \in [\abs{\Sigma}^m]$ such that $x \in \Pi_i$.
\end{proof}

As discussed in the introduction, one of the most fundamental questions regarding proofs of proximity is their relative strength in comparison to testers; that is, whether verifying a proof for an approximate decision problem can be done significantly more efficiently than solving it. This can be cast as an analogue of the P versus NP question for property testing.

Fortunately, in the setting of property testing, the problem of verification versus decision is very much tractable: one of the main results in \cite{GR18} shows the existence of a property $\Pi$ which: (1) admits a MAP with proof length $O(\log n)$ and query complexity $q=O(1)$; and (2) requires at least $n^{1-1/\Omega(q)}$ queries to be tested without access to a proof. (The lower bound of \cite{GR18} is stated in a slightly weaker form. However, it is straightforward to see that the stronger form holds; see discussion at the end of this section.)

While this implies a nearly exponential separation between the power of testers and MAPs, it remained open whether the aforementioned sublinear lower bound on testing is an artefact of the techniques, or whether it is possible to obtain a stronger separation, where the property is harder for testers.

\cref{clm:mapequiv} and \cref{thm:main} allow us to prove the following corollary, which shows that the foregoing separation is nearly tight.

\begin{theorem}
	\label{thm:mappt}
	If a property $\Pi \subseteq \Sigma^n$ admits a MAP with query complexity $q$, proof length $m$ and proximity parameter $\eps = \Omega(1)$, then it admits a sample-based $2\eps$-tester with sample complexity $m \cdot n^{1-1/O(q^2 \log^2q)}$.
\end{theorem}

Applying \cref{thm:mappt} to the special case of MAPs with \emph{logarithmic} proof length, we obtain a sample-based tester with sample complexity $n^{1-1/O(q^2 \log^2q)}$, showing that the separation in \cite{GR18} is nearly optimal, and in particular that there cannot be a fully exponential separation between MAPs and testers.

\begin{proof}[Proof (of \cref{thm:mappt})]
	Let $\Pi$ be a property and $T$ be a MAP with proof length $m$ as in the statement. By \cref{clm:mapequiv}, there exists a collection of partial testers $\set{T_i : i \leq \abs{\Sigma}^m}$ with query complexity $q$ that satisfy the following. Each $T_i$ accepts inputs in a property $\Pi_i$ and rejects inputs that are $\eps$-far from $\Pi$, with $\Pi \subseteq \cup_i \Pi_i$. By applying \cref{cor:test} to each of these testers, we obtain a collection of sample-based testers $\set{S_i}$ with sample complexity $q' = n^{1-1/O(q^2 \log^2q)}$ for the same partial properties, but which only reject inputs that are $2\eps$-far from $\Pi$.

  The execution of each of the $S_i$ proceeds in two steps, as defined in \cref{sec:construction}: \emph{sampling} (\ref{step:sample}) and \emph{enumeration} (\ref{step:enumeration}). Note that the sampling step is exactly the same for every $S_i$.

  Let $k = O(m \log \abs{\Sigma})$ such that taking the majority output from $k$ repetitions of $S_i$ yields an error rate of $1/(3\abs{\Sigma}^m)$. We define a new sample-based algorithm $S$ that repeats the following steps $k$ times:
  \begin{enumerate}
    \item\label{item:partial-sample} Execute both steps of $S_1$ (sampling and enumeration), recording the output.
    \item For all $1 < i \leq \abs{\Sigma}^m$, \emph{only execute the enumeration step} of $S_i$ on the samples obtained in \cref{item:partial-sample}, and record the output.
  \end{enumerate}
  After all $k$ iterations have finished, check if at least $k/2$ outputs of $S_i$ were 1 for some $i$. If so, output 1, and output 0 otherwise.

   First suppose $S$ receives an input $x \in \Pi$, and let $i \leq \abs{\Sigma}^m$ such that $x \in \Pi_i$. Then the majority output of the enumerations steps of $S_i$ is 1 with probability $1 - 1/(3\abs{\Sigma}^m) \geq 2/3$. Now suppose $S$ receives an input $x$ that is $2\eps$-far from $\Pi$. Then, for each $i$, the majority output of the enumeration step of $S_i$ is 1 with probability at most $1/(3\abs{\Sigma}^m)$. A union bound over all $i \leq \abs{\Sigma}^m$ ensures this happens with probability at least $1/3$, in which case $S$ correctly outputs 0.

   $S$ is therefore a $2\eps$-tester for the property $\Pi$ with sample complexity $k \cdot q' = m \cdot n^{1-1/O(q^2 \log^2q)}$ (recall that $\log \abs{\Sigma} \leq \log n$), and the theorem follows.
\end{proof}

Interestingly, as a direct corollary of \cref{thm:mappt}, we obtain that the general transformation in \cref{thm:main} is optimal, up to a quadratic gap in the dependency on the sample complexity, as a transformation with a smaller sample complexity could have been used to transform the MAP construction in the MAPs-vs-testers separation of \cite{GR18}, yielding a tester with query complexity that contradicts the lower bound in that result.

\begin{theorem}
  \label{thm:optimality}
   There does \emph{not} exist a transformation that takes a robust local algorithm with query complexity $q$ and transforms it into a \emph{sample-based} local algorithm with sample complexity at most $ n^{1- 1/o(q)}$.
\end{theorem}

\begin{proof}
  Let $\Pi$ be the \emph{encoded intersecting messages} property considered in \cite[Section 3.1]{GR18}, for which it was shown that $\Pi$ has a MAP with query complexity $q$ and logarithmic proof complexity, but every tester for $\Pi$ requires at least $n^{1-1/\Omega(q)}$ queries. Suppose towards contradiction that a transformation as in the hypothesis exists. Then, applying the transformation to the aforementioned MAP (as in \cref{thm:mappt}) yields a tester for $\Pi$ with query complexity $n^{1-1/o(q)}$, in contradiction to the lower bound.
  \end{proof}

\paragraph{On the lower bound in \cite{GR18}\ifsicomp\else.\fi}
The separation between MAPs and testers in \cite{GR18} is proved with respect to a property of strings that are encoded by relaxed LDCs; namely, the \emph{encoded intersecting messages property}, defined as
\begin{equation*}
  \textrm{EIM}_C = \left\{ \big( C(x),C(y) \big) \;:\; x,y \in \bitset^k,\; k
    \in \N \text{ and } \exists i\in[k] \text{ s.t. } x_i \neq 0
    \text{ and } y_i \neq 0 \right\},
\end{equation*}
where $C \colon \bitset^k \to \bitset^n$ is a code with linear distance, which is both a relaxed LDC and an LTC. In \cite{GR18} it is shown that there exists a MAP with proof length $O(\log n)$ and query complexity $q=O(1)$, and crucially for us, that any tester requires $\Omega(k)$ queries to be tested without access to a proof. The best constructions of codes that satisfy the aforementioned conditions \cite{BGHSV06,CGS22,AS21} achieve blocklength $n = O(k^{1 + 1/q}) = k^{1 + 1/\Omega(q)}$, and hence the stated lower bound follows.

\cleardoublepage
\ifsicomp
  \bibliographystyle{siamplain}
  \bibliography{references}
\else
\phantomsection
  \addcontentsline{toc}{section}{References}
  \bibliographystyle{alpha}
  \bibliography{references}
\fi

\clearpage
\appendix
\section{Deferred proofs}
\label{sec:deferred}

In this appendix we provide the proofs of two claims and a lemma obtained from their combination, which were deferred in \cref{sec:lemmas}: \cref{clm:err-red} provides an amplification procedure and \cref{clm:rand-red} a randomness reduction procedure for local algorithms, while \cref{lem:prep-alg} obtains both simultaneously. We remark that both claims follow from a straightforward adaptation of standard techniques, and include these proofs for completeness. We begin with the amplification procedure.

\begin{claim}[\cref{clm:err-red} restated]
Let $M$ be a $(\rho_0, \rho_1)$-robust algorithm for computing $f\colon \mathcal{P} \rightarrow \bitset$ (where $\mathcal{P} \subset Z \times \Sigma^n$) with error rate $\sigma \leq 1/3$, query complexity $q$ and randomness complexity $r$.

For any $\sigma' > 0$, there exists a $(\rho_0, \rho_1)$-robust algorithm $N$ for computing the same function with error rate $\sigma'$, query complexity $108 q \log(1/\sigma')/\sigma$ and randomness complexity $108 r \log(1/\sigma')/\sigma$.
\end{claim}
\begin{proof}
  Define $N$ as the algorithm that makes $t = 108 \log(1/\sigma')/\sigma$ independent runs of $M$ and outputs the most frequent symbol, resolving ties arbitrarily. The query and randomness complexities of $N$ clearly match the statement, and we must now prove that the error rate is indeed $\sigma'$ and that $N$ is $(\rho_0,\rho_1)$-robust.

  Fix $z \in Z$ and $x \in \Sigma^n$ in the domain of $f$ and let $b \coloneqq f(z,x)$. As $M$ is $\rho_b$-robust at $x$, the algorithm satisfies $\Pr[M^y(z) = b] \geq 1 - \sigma$ for all $y \in B_{\rho_b}(x)$. By the Chernoff bound,
  \begin{equation*}
		\Pr\left[M^y(z) \neq b \text{ for at least } (\sigma + 1/6) t \text{ runs}\right] \leq e^{-\frac{\sigma t}{3 \cdot 36}} = e^{-\log \frac{1}{\sigma'}} < 2^{-\log \frac{1}{\sigma'}} = \sigma'.
	\end{equation*}
  The majority rule will thus yield outcome $b$ with probability at least $1 - \sigma'$, since at least $1 - (\sigma + 1/6) t \geq t/2$ runs output $b$ (except with probability at most $\sigma'$). As $x, z$ and $y \in B_{\rho_b}(x)$ are arbitrary, the result follows.
\end{proof}

We proceed to the randomness reduction transformation.

\begin{claim}[\cref{clm:rand-red}, restated]
    Let $M$ be a $(\rho_0, \rho_1)$-robust algorithm for computing $f\colon \mathcal{P} \rightarrow \bitset$ (where $\mathcal{P} \subset Z \times \Sigma^n$) with error rate $\sigma$, query complexity $q$ and randomness complexity $r$.

    There exists a $(\rho_0, \rho_1)$-robust algorithm $N$ for computing the same function with error rate $2\sigma$ and query complexity $q$, whose distribution $\Tilde{\mu}^N$ has support size $3n \ln \abs{\Sigma}/\sigma$. In particular, the randomness complexity of $N$ is bounded by $\log (n/\sigma) + \log \log \abs{\Sigma} + 2$.
\end{claim}

\begin{proof}
  Fix any explicit input $z \in Z$. Let $\set{x_j}$ be an enumeration of the inputs in $\Sigma^n$ such that $\Pr[M^{x_j}(z) = b_j] \geq 1 - \sigma$ for some $b_j \in \bitset$. Note that this includes points in the neighbourhood of a point at which $M$ is robust which are not necessarily in the domain of $f$, so it suffices to show $\Pr[N^{x_j}(z) = b_j] \geq 1 - 2\sigma$ to prove the claim for $N$ with the required query complexity and distribution.

  Define the $2^r \times \abs{\set{x_j}}$ matrix $E$ with entries in $\bitset$ as follows. Denote by $b_{ij} \in \bitset$ the output of $M^{x_j}(z)$ when it executes according to the decision tree indexed by (the binary representation of) $i \in [2^r]$. Then,
  \begin{equation*}
		E_{i,j} = \left\{\begin{array}{cl}1, &\text{if } b_{ij} \neq b_j\\0, &\text{otherwise.}\end{array}\right.
	\end{equation*}
  Note that $E_{i,j}$ simply indicates whether $M^{x_j}(z)$ outputs incorrectly on input when the outcome of the algorithm's coin flips is (the binary representation of) $i$. By construction, for each fixed $j$, a fraction of at most $\sigma$ indices $i \in [2^r]$ are such that $E_{i,j} = 1$.

  Let $t = 3 n \ln \abs{\Sigma}/\sigma$ and $I_1, \ldots, I_t$ be independent random variables uniformly distributed in $[2^r]$. For each fixed $j \leq \abs{\set{x_j}} \leq \abs{\Sigma}^n$ and $k \leq t$, we have $\E[E_{I_k,j}] \leq \sigma$. By the Chernoff bound,
  \begin{equation*}
		\Pr\left[\sum_{k = 1}^t E_{I_k,j} \geq 2\sigma t\right] \leq e^{-\frac{\sigma t}{3}} = e^{-n \ln \abs{\Sigma}} < \abs{\Sigma}^{-n}.
	\end{equation*}

  Applying the union bound over all $j \leq \abs{\Sigma}^n$, we obtain
  \begin{equation*}
		\Pr\left[\sum_{k = 1}^t E_{I_k,j} \ge 2\sigma t \text{ for some } j\right] < 1.
	\end{equation*}

  We have thus shown, via the probabilistic method, the existence of a multi-set $R_z$ of size $3n \ln \abs{\Sigma}/\sigma$ such that
  \begin{equation*}
		\Pr[N^{x_j}(z) \neq b_j] \leq 2\sigma,
	\end{equation*}
  where $N$ samples its random strings uniformly from $R_z$ (rather than from $\bitset^r$), using the corresponding decision trees of $M$. The size of $S_z$ is thus $\abs{\Tilde{\mu}^N} = 3n \ln \abs{\Sigma}/\sigma$, and this sampling can be performed with $\log (n/\sigma) + \log \log \abs{\Sigma} + 2$ random coins.

  Since the decision trees of $N$ are simply a subcollection of those of $M$, the query complexity of $N$ is $q$ and the claim follows.
\end{proof}

We now conclude with the following lemma, obtained by suitably combining the foregoing claims in sequence.

\begin{lemma}[\cref{lem:prep-alg}, restated]
  Assume there exists a $\rho$-robust algorithm $M$ for computing $f$ with query complexity $\ell$, error rate $1/3$ and arbitrary randomness complexity.
  Then there exists a $\rho$-robust $q$-local algorithm $M'$ for $f$ with error rate
	\begin{equation*}
		\sigma = \errorratefrac
	\end{equation*}
	such that $\frac{q}{\log 8q} = O(\ell)$, or, equivalently,
  \begin{equation*}
    q = O(\ell \log \ell).
  \end{equation*}
  Moreover, the distribution of $M'$ is uniform over a multi-collection of decision trees of size $\suppsize$.
\end{lemma}
\begin{proof}
  We apply both transformations in order, omitting mention of parameters that are left unchanged. Recall that $M$ may have arbitrarily large randomness complexity.

  \begin{enumerate}
	  \item Apply \cref{clm:err-red} (error reduction) to $M$, obtaining $M''$ with error rate $\sigma'' = 1/8q$ and query complexity $q = O(\ell \log(1/\sigma'')) = O(\ell \log(8q))$ (as well as larger randomness complexity).
    \item Apply \cref{clm:rand-red} (randomness reduction) to $M''$, thereby obtaining a new algorithm $M'$ with error rate $\sigma = 2 \sigma'' = \errorratefrac$ and support size $3n \ln \abs{\Sigma}/\sigma'' = \suppsize$ on its distribution over decision trees. \ifsicomp\else\qedhere\fi
  \end{enumerate}
\end{proof}

\end{document}